\documentclass[11pt,a4paper,oneside,onecolumn,final,accepted=2026-01-17]{quantumarticle}
\pdfoutput=1

\usepackage[english]{babel} 
\usepackage[utf8]{inputenc}
\usepackage[T1]{fontenc}
\usepackage{microtype}
\usepackage{textcomp}
\usepackage{ragged2e} 
\usepackage{comment}
\usepackage[autostyle=true]{csquotes}

\usepackage{
    amsmath,
    amssymb, 
    amsfonts,
    amsthm,
    mathtools,
}

\usepackage{graphicx} 
\graphicspath{{./figures/}}
\usepackage[shortlabels]{enumitem} 

\usepackage[numbers,sort&compress]{natbib}
\usepackage{xcolor}
\definecolor{cblue}{rgb}{0.16, 0.32, 0.75}
\definecolor{cred}{rgb}{0.7, 0.11, 0.11}

\usepackage[
	    colorlinks=true, 
	    breaklinks=false, 
	    linkcolor=blue, 
		draft=false,
	    bookmarksopen=false, 
	    bookmarksopenlevel=0,
	    bookmarksnumbered=true,
]{hyperref}
\hypersetup{%
	,linkcolor=cred
	,citecolor=cblue
	,urlcolor=black
}

\usepackage{dsfont} 
\usepackage{braket} 
\usepackage{diffcoeff} 
\usepackage{siunitx} 
\usepackage{cancel} 

\usepackage{xparse} 

\usepackage[capitalize,nameinlink]{cleveref} 

\newcommand{\id}{\mathds{1}} 
\newcommand{\e}{\mathrm{e}}

\DeclareMathOperator{\mmod}{mod}

\newcommand{\iu}{\mathrm{i}\mkern1mu} 
\newcommand{\nnum}{\mathbb{N}}
\newcommand{\rnum}{\mathbb{R}}
\newcommand{\cnum}{\mathbb{C}}

\newcommand{\znum}{\mathbb{Z}}

\newcommand{\hilbert}{\mathcal{H}}
\newcommand{\domain}{\mathcal{D}}
\numberwithin{equation}{section} 
\newtheorem{theorem}{Theorem}[section] 
\newtheorem{corollary}{Corollary}[section] 
\newtheorem{lemma}[theorem]{Lemma} 
\newtheorem{proposition}[theorem]{Proposition} 
\theoremstyle{definition} 
\newtheorem{definition}[theorem]{Definition} 
\theoremstyle{remark} 
\newtheorem{remark}[theorem]{Remark}
\theoremstyle{question} 
\newtheorem{question}[theorem]{Question} 
\theoremstyle{definition} 
\newtheorem{example}[theorem]{Example} 

\DeclarePairedDelimiterX{\norm}[1]\lVert\rVert{
  \ifblank{#1}{\:\cdot\:}{#1}
}
\difdef { f } {}
{
  outer-Ldelim = \left . ,
  outer-Rdelim = \right |,
  sub-nudge = 0 mu
}

\title{Quantum particle in the wrong box (or: the perils of finite-dimensional approximations)}

\author{Felix Fischer}
\email{felix.o.fischer@fau.de}
\orcid{0009-0001-3040-8184}

\author{Daniel Burgarth}
\orcid{0000-0003-4063-1264}

\author{Davide Lonigro}
\orcid{0000-0002-0792-8122}
\affiliation{Department Physik, Friedrich-Alexander-Universität Erlangen-Nürnberg, Staudtstra{\ss}e 7, 91058 Erlangen, Germany}
\date{}

\begin{document}

\maketitle

\begin{abstract}\sloppy

     When numerically simulating the unitary time evolution of an infinite-dimensional quantum system, one is usually led to treat the Hamiltonian $H$ as an ``infinite-dimensional matrix'' by expressing it in some orthonormal basis of the Hilbert space, and then truncate it to some finite dimensions. However, the solutions of the Schr\"odinger equations generated by the truncated Hamiltonians need not converge, in general, to the solution of the Schr\"odinger equation corresponding to the actual Hamiltonian.

     In this paper we demonstrate that, under mild assumptions, they converge to the solution of the Schr\"odinger equation generated by a specific Hamiltonian which crucially depends on the particular choice of basis: the Friedrichs extension of the restriction of $H$ to the space of finite linear combinations of elements of the basis. Importantly, this is generally different from $H$ itself; in all such cases, numerical simulations will unavoidably reproduce the wrong dynamics in the limit, and yet there is no numerical test that can reveal this failure, unless one has the analytical solution to compare with.

     As a practical demonstration of such results, we consider the quantum particle in the box, and we show that, for a wide class of bases---which include associated Legendre polynomials as a concrete example---the dynamics generated by the truncated Hamiltonians will always converge to the one corresponding to the particle with Dirichlet boundary conditions, regardless the initial choice of boundary conditions. Other such examples are discussed.
\end{abstract}

    \section{Introduction}\label{sec:intro}

 Imagine you are assigned the following homework for an undergraduate quantum mechanics class: to numerically solve the time-independent Schrödinger equation for a particle in the box with periodic boundary conditions. You choose a basis of wavefunctions with the correct boundary conditions, and express the Hamiltonian in this basis. You then solve for the time evolution with your favorite method and programming language. You observe that the numerical simulation converges when increasing the number of basis elements, and that the solution is preserving the normalization of the initial wavefunction. You hand back the solution, yet you fail the class: you have obtained the solution for the wrong box, namely the box with hard wall (Dirichlet) boundary conditions.

    This article is about what went wrong, and why this also has implications for research level problems.

    \subsection{Convergence of Galerkin approximations}\label{sec:results_galerkin_informal}

    Let us begin by integrating our problem in a more general framework. We shall consider a quantum system whose energy is described by a self-adjoint operator $H$ (the Hamiltonian) on an infinite-dimensional Hilbert space $\mathcal{H}$; in order to make our description suitable to quantum systems with unbounded energy spectrum, like in the case of the particle in a box, $H$ will generally be an unbounded linear operator on $\hilbert$, and thus to be defined on a proper subspace $\domain(H)$ of the Hilbert space---the domain of $H$:
    \begin{equation}
        H:\mathcal{D}(H)\subset\hilbert\rightarrow\hilbert.
    \end{equation}
    We further assume $H$ to be bounded from below, which is typically the case in applications.
    
    As is well-known, the Hamiltonian $H$ generates a unitary propagator $U(t)=\e^{-\iu tH}$ providing the unique solution of the Schr\"odinger equation. However, in most realistic situations, analytically computing $U(t)$ proves to be a formidable challenge---nor can computers handle systems with infinitely many degrees of freedom.
    
    To this extent, we can consider a family $(P_n)_{n\in\mathbb{N}}$ of finite-dimensional projectors on $\mathcal{H}$, e.g. each with $n$-dimensional range, and strongly converging to the identity---that is, $P_n\psi\to\psi$ as $n\to\infty$ for every state $\psi\in\hilbert$.
    Assuming $P_n \mathcal{H}\subset\domain(H)$ for every $n$, one can define a family of finite-dimensional truncations of $H$ via $H_n = P_n H P_n$, each of them being a bounded operator on $\hilbert$ generating a new unitary propagator $U_n(t)=\e^{-\iu tH_n}$ which can be computed numerically using standard linear algebra procedures.
    This procedure is usually referred to as \textit{Galerkin approximation} or \textit{Galerkin's method} (cf.~Section~\ref{sec:biblio}).
    Typically, one constructs $P_n$ by choosing some suitable orthonormal basis $(\phi_l)_{l\in\mathbb{N}}\subset\mathcal{D}(H)$, and defines $P_n$ as the projector on the space spanned by the first $n$ elements of the basis,
    \begin{equation}\label{eq:onto_the_first_n}
        P_n=\sum_{l=0}^{n-1}\braket{\phi_l,\cdot}\phi_l.
    \end{equation}    
    Intuitively, one would then expect that, as $n\to\infty$, the propagators $U_n(t)$ provide a ``good'' approximation of the propagator $U(t)$ generated by the original Hamiltonian $H$. This leads us to the following:

    \begin{question}
        Under which conditions do the unitary propagators $U_n(t)$ generated by $H_n=P_nHP_n$ strongly converge to the unitary propagator $U(t)=\e^{-\iu tH}$, in the following sense:
        \begin{equation}
            \lim_{n\to\infty}\|U_n(t)\psi-U(t)\psi\|=0\qquad\text{for all }\psi\in\hilbert?
        \end{equation}
    \end{question}
    One case in which the answer is always affirmative is the following: $H$ has a discrete spectrum, and the projectors $P_n$ are defined via Eq.~\eqref{eq:onto_the_first_n} by choosing the vectors $(\phi_l)_{l\in\mathbb{N}}$ as a basis of eigenvectors of $H$ itself (Remark~\ref{rem:eigenbasis-convergence}). However, in practical applications, such an explicit knowledge of the eigenvalues and corresponding eigenvectors of $H$ is not at our disposal, save from very particular cases. As such, it makes sense to worry about this question for more general families of projectors. 

    As it turns out, the mere fact that $P_n\to\id$ strongly is generally not sufficient to ensure that the question above has an affirmative answer. Quite the contrary, in this paper we will show (Proposition~\ref{prop:friedrichs}) that, in general,
    \begin{equation}
            \lim_{n\to\infty}\|U_n(t)\psi-\tilde{U}(t)\psi\|=0\qquad\text{for all }\psi\in\hilbert,
    \end{equation}
    where $\tilde{U}(t)$ is the unitary propagator generated by a self-adjoint operator $\tilde{H}$, generally distinct from $H$, obtained through the following procedure:
    \begin{enumerate}
        \item one restricts $H$ to the dense subspace $\hilbert_{\rm fin}=\cup_nP_n\hilbert$, i.e., the space of all vectors $\psi\in\hilbert$ such that $P_n\psi=\psi$ for sufficiently large $n$;
        \item one then considers the so-called \textit{Friedrichs extension} $\tilde{H}$ of this restriction (cf.~Definition~\ref{prop:def-friedrich}).
    \end{enumerate}
    Consequently, and taking into account that the relation between self-adjoint operators and their corresponding unitary propagators is one-to-one, the question has an affirmative answer if and only if $H$ coincides with the operator $\tilde{H}$ as defined above. Since in general the restriction of $H$ to $\hilbert_{\rm fin}$ may have infinitely many self-adjoint extensions---with $H$ being only one of them---in general $H\neq\tilde{H}$, and thus $U_n(t)\psi\not\to U(t)\psi$. In this case, all numerical simulations of the Schr\"odinger equation based on this choice of projectors will dramatically fail to reproduce the desired dynamics in the limit $n\to\infty$. Merely choosing the projectors in such a way that $P_n\hilbert\subset\domain(H)$ (for example, constructing them via an orthonormal basis entirely contained in $\domain(H)$) is not sufficient, in general, to avoid this phenomenon.

   \subsection{Galerkin approximations of the particle in a box}\label{sec:results_box_informal}

   We can now revisit the ``thought homework'' presented at the beginning of the paper in light of these general results. To this extent, it will be useful to briefly recall the mathematical description of a quantum particle in a one-dimensional box. 
   
   Without loss of generality, the box shall be identified with the real interval $(-1,1)$, its boundary thus consisting of the two points $x=\pm1$. The Hilbert space of the system is $\hilbert = L^2((-1,1))$, that is, the space of square-integrable, complex-valued functions on $(-1,1)$. One then proceeds to identify the Hamiltonian of the particle. As is well-known, this operator should act on wavefunctions $\psi$ as minus the second derivative
with respect to the position $x$ of the particle---the one-dimensional Laplace operator:
\begin{equation}
    (H\psi)(x)=-\psi''(x)=-\frac{\mathrm{d}^2}{\mathrm{d}x^2}\psi(x),
\end{equation}
where the mass $m$ of the particle was fixed to $1/2$  and $\hbar = 1$. However, as already remarked in the general case, one must also specify a proper \textit{domain} $\mathcal{D}(H)$ for our operator. Here, specifically,
\begin{itemize}
    \item $\mathcal{D}(H)$ must contain functions satisfying the minimal condition $\psi''\in\hilbert=L^2((-1,1))$. The space of such functions is the second Sobolev space $\mathrm{H}^2((-1,1))$, cf.~Definition~\ref{def:sobolev_spaces};
    \item Besides, because of the presence of a boundary, $\mathcal{D}(H)$ must contain functions implementing the specific boundary conditions satisfied by the wavefunctions.
\end{itemize}
The last piece of information is particularly relevant: distinct boundary conditions correspond to distinct Hamiltonians and, ultimately, to distinct physical systems whose boundaries have physically distinct properties---or even different geometries. 
As is well-known, there are infinitely many admissible choices of boundary conditions, each parametrized by a $2\times2$ unitary matrix $W\in\mathrm{U}(2)$ (cf.~Section~\ref{sec:particle-in-a-box} and references therein), each boundary condition reading
\begin{equation}\label{eq:bc}
    \iu(I+W)\begin{pmatrix}
        -\psi'(-1)\\\psi'(+1)
    \end{pmatrix}=(I-W)
    \begin{pmatrix}
        \psi(-1)\\\psi(+1)
    \end{pmatrix}\, .
\end{equation}
Each choice of $W$ will correspond to a distinct operator $H_W:\mathcal{D}(H_W)\subset\hilbert\rightarrow\hilbert$, representing a distinct physical situation---and thus, generating a distinct unitary propagator $U_W(t)=\e^{-\iu tH_W}$.

Coming back to Galerkin approximations, the following question arises:

\begin{question}
    Consider an orthonormal basis $(\phi_l)_{l\in\mathbb{N}}\subset\mathrm{H}^2((-1,1))$ of the Hilbert space $L^2((-1,1))$ with the following property: for every $l\in\mathbb{N}$,
    \begin{equation}\label{eq:blind}
        \phi_l(\pm1)=0,\qquad\phi'_l(\pm1)=0.
    \end{equation}
    Let $H_n$ be the $n$th Galerkin approximation of the Laplace operator corresponding to this choice of basis, and $U_n(t)$ the corresponding unitary propagator. Then:
    \begin{itemize}
        \item will $U_n(t)$ converge strongly to some unitary propagator?
        \item if so, does this propagator coincide with the unitary propagator generated by some specific operator $H_W$, i.e., corresponding to some specific choice of boundary conditions?
    \end{itemize}
\end{question}
Note the counterintuitive fact that Eq.~\eqref{eq:blind} does not prevent $(\phi_l)_{l \in \nnum}$ from being a complete orthonormal set in $L^2((-1,1))$: such bases can be found, and they are as good as any other. The question itself is nontrivial for the following reason: an orthonormal basis satisfying Eq.~\eqref{eq:blind} is essentially blind to the choice of boundary conditions---that is, it satisfies Eq.~\eqref{eq:bc} for any possible choices of $W$. As such, there is no heuristic reason why finite-dimensional approximations constructed this way should ``privilege'', in the limit $n\to\infty$, one specific choice of boundary conditions.

Instead, we will show the following results:
\begin{itemize}
    \item Under some assumptions, the unitary propagator $U_n(t)$ corresponding to a ``boundary-blind'' orthonormal basis $(\phi_l)_{l\in\mathbb{N}}$ converges strongly to the evolution generated by the Laplace operator with \textit{Dirichlet} boundary conditions (Proposition~\ref{prop:legendre-dirichlet-general}; see Figure~\ref{fig:numerics-nt-alt});
    \item Under the same conditions, it is further possible to modify the orthonormal basis above in such a way that $U_n(t)$ converges strongly to the evolution generated by the Laplace operator with other boundary conditions of our choice---specifically, periodic boundary conditions and some generalizations (Proposition~\ref{prop:alpha-periodic-unitary}). This can be done in a minimal way: adding to the basis a single function satisfying the desired boundary conditions can suffice.
\end{itemize}
In particular, we will offer a concrete example of ``boundary-blind'' orthonormal bases such that the phenomena listed above happen: the normalized associated Legendre polynomials $(p^m_l)_{l\geq m}$ of order $m\geq4$ (Theorems~\ref{thm:legendre-dirichlet} and~\ref{thm:alpha-periodic-legendre}).
We visualize our results for this choice of basis in Figure~\ref{fig:numerics-nt-alt}, where the error $\norm{U_W(t) \psi_0 - U_n(t) \psi_0}$ is plotted as a function of $n$ for two choices of $W$: Dirichlet and periodic boundary conditions. In agreement with Theorem~\ref{thm:legendre-dirichlet}, the approximation error converges to $0$ for Dirichlet boundary conditions, whereas it stays nonzero (and relatively big) for periodic boundary conditions.
We refer to Appendix~\ref{app:expansion-coefficients} for the lengthy calculation of all quantities involved in evaluating the error.
\begin{figure}[ht!]
    \centering
    \includegraphics[width=0.9\linewidth]{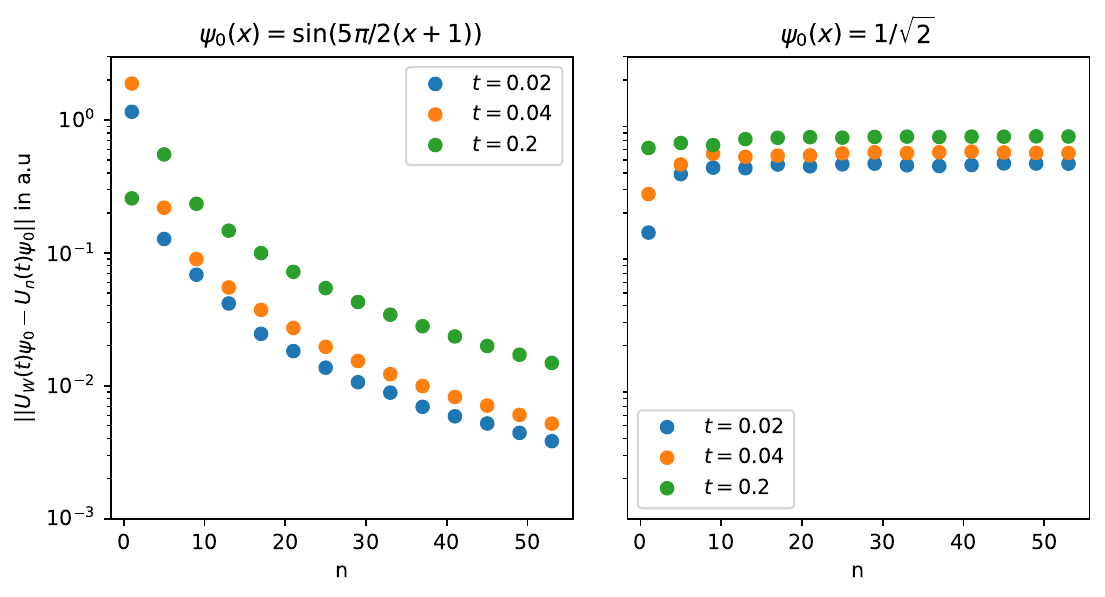}
    \caption{Approximation error $\norm*{U_W(t)\psi_0-U_n(t)\psi_0}$ for the Dirichlet time evolution $U_{\mathrm{Dir}}$ of 
    a Dirichlet eigenvector $\psi_0(x) = \sin(5/2 \pi (x+1))$ (left) and for the periodic time evolution $U_{\mathrm{per}}(t)$ of a periodic eigenvector $\psi_0(x)=1/\sqrt{2}$ (right) over the 
    number of basis elements $n$ and for different times $t$. In both cases, finite-dimensional truncations are constructed by means of associated Legendre polynomials $p_l^m$ with $m=4$. The error only converges to $0$ in the Dirichlet case, in agreement with our analytical results (cf.~Theorem~\ref{thm:legendre-dirichlet}).}
    \label{fig:numerics-nt-alt}
\end{figure}

This apparently contradictory phenomenon has, in fact, a clear mathematical interpretation in the light of the abstract results informally presented in Section~\ref{sec:results_galerkin_informal}: in both cases sketched above, the specific boundary conditions that are ``selected'' in the limit $n\to\infty$ are precisely those corresponding to the Friedrichs extension of the Laplace operator initially defined on the space of finite linear combinations of these basis vectors. If the desired boundary conditions differ from the one corresponding to the Friedrichs extension, one will always obtain a wrong result---and yet there is no numerical test that can reveal this failure, unless (like here) one has the analytical solution to compare with.

\subsection{Some bibliographic remarks}\label{sec:biblio}

    Galerkin approximations and equivalent finite-dimensional truncations are constantly used in the literature, cf.~\cite[p.~273]{tannor-introductionquantummechanics-2007},~\cite[p.~18]{shizgal-spectralmethodschemistry-2015},~\cite[p.~67]{boyd-chebyshevfourierspectral-2000}.
    In quantum chemistry, see e.g.~\cite[pp. 202--204]{gatti-applicationsquantumdynamics-2017} and references therein, they are also often called \textit{finite basis representation}.
    They are also a fundamental tool in quantum control theory, where (again under the name Galerkin approximations) they are employed in order to transfer controllability results from finite-dimensional approximations to the full, infinite-dimensional control system~\cite{chambrion_periodic_2012,boussaid_weakly-coupled_2013,balmaseda_quantum_2023,balmaseda_global_2024,balmaseda_sharper_2024}, see in particular~\cite[Theorem~4.7]{balmaseda_sharper_2024}.

    In the mathematical literature, the convergence of Galerkin approximations is often treated in the Hilbert space framework introduced here; see, in particular,~\cite[p.~101,~243]{miklavcic-appliedfunctionalanalysis-2001}.
    In the more general framework of Banach spaces, a variety of upper-bounds for the approximation errors in the respective norms can be deduced~\cite{parter-rolesstabilityconvergence-1980,helfrich-fehlerabschaetzungenfuergalerkinverfahren-1974,gottlieb-numericalanalysisspectral-1977}: 
    in particular, parabolic differential equations are discussed in~\cite{douglas-galerkinmethodsparabolic-1970,wheeler-prioril_2error-1973,baker-singlestepgalerkin-1977,bramble-convergenceestimatessemidiscrete-1977,evans-partialdifferentialequations-2010}.
    Furthermore, Cea's Lemma~\cite{cea-approximationvariationnelleproblemes-1964,brenner-mathematicaltheoryfinite-2008} provides such an error estimate for finite dimensional approximations of stationary problems: for the specific case of approximating the Dirichlet Laplacian with a Sobolev--complete basis, see~\cite{zhikov-galerkinapproximationsproblems-2016,pastukhova-galerkinapproximationsdirichlet-2019}.
    
    In quantum chemistry, on the other hand, approximation errors and the question of convergence are often treated by means of heuristic arguments or with the aid of numerical simulations, e.g. by comparing numerical results with reference systems~\cite[p.~387]{lebris-computationalchemistryperspective-2005}.
    In particular, no general convergence statements are usually involved~\cite{gatti-applicationsquantumdynamics-2017,tannor-introductionquantummechanics-2007,shizgal-spectralmethodschemistry-2015}.
    Notable exceptions are~\cite{klahn-convergencerayleighritzmethod-1977,klahn-convergencerayleighritzmethod-1977a,kato-fundamentalpropertieshamiltonian-1951} and~\cite[p.~387]{lebris-computationalchemistryperspective-2005}, where the convergence of the lowest discrete eigenvalues of molecular Hamiltonians is shown for multiple typical basis sets used in quantum chemistry.
    Cases in which Galerkin's method fails to reproduce the correct dynamics in the limit are usually not mentioned in the literature, an exception being~\cite{klahn-convergencerayleighritzmethod-1977} again.

    Galerkin approximations of bosonic systems were considered rigorously in a number of recent works~\cite{fischer-selfadjointrealizationshigherorder-2025,ashhab-finitedimensionalapproximationsgeneralized-2026,robin-convergenceanalysisgalerkin-2025,etienney-posteriorierrorestimates-2025,arzani-effectivedescriptionsbosonic-2025}.
    In particular, \cite{fischer-selfadjointrealizationshigherorder-2025,ashhab-finitedimensionalapproximationsgeneralized-2026} highlight that finite--dimensional approximations of higher--order squeezing operators always lead to unexpected oscillations in the truncation size.

    The potential relevance of such numerical issues extend far beyond its mere educational value. Even the particular case study presented in this paper, the particle in a box with hard walls, provides first important insights into the simulation of systems of physical relevance: it serves as a simple model for confined particles, such as electrons in molecules and crystals~\cite{autschbach-whyparticleboxmodel-2007,vos-particleboxmomentumdensities-2002,anderson-lasersynthesislinear-2008,ruedenberg-freeelectronnetworkmodel-1953,scherr-freeelectronnetworkmodel-1953,anjos-quantummechanicsparticles-2024,fillaux-neutronscatteringstudies-2012}, nuclei in the nucleus~\cite[p.~77]{basdevant-fundamentalsnuclearphysics-2005}, quantum dots~\cite{schmid-nanoparticlestheoryapplication-2010,landry-simplesynthesescdse-2014,rice-quantumdotspolymer-2008} or quantum well lasers~\cite{holonyak-quantumwellheterostructurelasers-1980}. On the other hand, Cooper pair boxes are described by a particle in a box with periodic boundary conditions~\cite{koch-chargeinsensitivequbitdesign-2007,girvin-circuitqedengineering-2009,bladh-singlecooperpairbox-2005}.
    As the associated Legendre polynomials are often used in quantum chemistry in the form of spherical harmonics~\cite{gatti-applicationsquantumdynamics-2017,yuan-usingfourierseries-2005,shizgal-spectralmethodschemistry-2015}, convergence issues similar to those presented here might also appear in other quantum chemistry problems.    
    Finally, regarding quantum boundary conditions, the different boundary conditions of a particle in a box are discussed in a more mathematical setting in~\cite{bonneau-selfadjointextensionsoperators-2001,asorey-dynamicalcompositionlaw-2013}, the latter providing a composition law for time-dependent, alternating boundary conditions.
    A more general case on Riemannian manifolds is considered in~\cite{asorey-globaltheoryquantum-2005a}, highlighting the importance of boundary conditions in a variety of fields like the Quantum Hall effect~\cite{halperin-quantizedhallconductance-1982,john-renormalizationgroupquantum-1995}, cosmology~\cite{vilenkin-boundaryconditionsquantum-1986} and quantum field theory~\cite{casimir-attractiontwoperfectly-1948,manton-schwingermodelits-1985}.
    The predominance of the Dirichlet boundary condition is discussed in~\cite[Prop 4.14.1, p.~759]{zagrebnov-trotterkatoproductformulae-2024} and~\cite{deoliveira-mathematicalpredominancedirichlet-2012}.
    Finally, in~\cite{berry-quantumfractalsboxes-1996} the time evolution of a particle in a box with hard walls is analytically calculated, and it is shown that initial states which do not fulfill Dirichlet boundary conditions evolve into fractal functions.
    
\subsection{Outline of the paper}

The remainder of the paper will be devoted to translating the results presented above into a rigorous mathematical framework. The paper is organized as follows:
\begin{itemize}
    \item in Section~\ref{sec:galerkin-general} we list, together with relevant definitions, our main results (informally presented in Section~\ref{sec:results_galerkin_informal}) about the convergence of the dynamics generated by Galerkin approximations of general self-adjoint operators;
    \item in Section~\ref{sec:particle_results} we list our main results (informally presented in Section~\ref{sec:results_box_informal}) concerning the convergence of the dynamics generated by Galerkin approximations of the Hamiltonian of a quantum particle in a box;
    \item Sections~\ref{sec:proofs_sec2} and~\ref{sec:proofs_sec3} contain the proofs of all results listed in Sections~\ref{sec:galerkin-general} and~\ref{sec:particle_results};
    \item finally, we gather some final remarks in Section~\ref{sec:conclusion}.
\end{itemize}

\section{Galerkin approximations: general results}
\label{sec:galerkin-general}

We will start by presenting some general results about the convergence of the unitary dynamics generated by Galerkin approximations of a given self-adjoint operator $H$. These results constitute the abstract substrate to the concrete convergence theorems for Galerkin approximations of the particle in a box that will be presented in Section~\ref{sec:particle_results}. As these results are of interest by themselves---and potentially applicable to situations of physical interest not covered in this work---we shall present them explicitly. The reader exclusively interested in the specific case of the particle in a box might directly jump to the next section.

For the remainder of this section, $\hilbert$ will be an infinite-dimensional, complex, and separable---that is, admitting countable complete orthonormal sets---Hilbert space; the scalar product on $\hilbert$, and its associated norm, will be denoted by $\braket{\cdot,\cdot}$ and $\|\cdot\|$. The domain of an unbounded linear operator $A$ on $\hilbert$ will be denoted by $\mathcal{D}(A)$, and the adjoint of $A$, whenever well-defined, will be denoted by $A^*$.

\begin{definition}[Galerkin approximation]\label{def:galerkin_approximation}
    Let $H:\domain(H)\subset\hilbert\rightarrow\hilbert$ be a self-adjoint operator on $\hilbert$, and $(P_n)_{n\in\mathbb{N}}$ a family of finite-dimensional orthogonal projectors satisfying the following properties:
    \begin{itemize}
\item $P_n\to\id$ strongly, that is, $P_n\psi\to\psi$ for all $\psi\in\hilbert$;
\item $P_n\hilbert\subset\domain(H)$ for every $n\in\mathbb{N}$.
    \end{itemize}
    The $n$th \textit{Galerkin approximation} of $H$ is the bounded operator on $\hilbert$ defined by
        \begin{equation}\label{eq:galerkin_approximation}
            H_n:\hilbert\rightarrow\hilbert,\qquad H_n:=P_nHP_n=\hat{H}_n\oplus0_{\hilbert_n^\perp},
        \end{equation}
        with $\hat{H}_n$ being the restriction of $H_n$ to the finite-dimensional Hilbert space $\hilbert_n = P_n\hilbert$.
      \end{definition} 

Above, in Eq.~\eqref{eq:galerkin_approximation}, we exploited the direct sum decomposition of $\hilbert$ given by $\hilbert \simeq \hilbert_n \oplus \hilbert_n^\perp$.
The decomposition $H_n = \hat{H}_n\oplus0_{\hilbert_n^\perp}$ is pictorially depicted in Figure~\ref{fig:matrix}. The decomposition~\eqref{eq:galerkin_approximation} of $H_n$ immediately means that the corresponding unitary propagator $U_n(t)=\e^{-\iu tH_n}$ decomposes as follows:
\begin{equation}
    \label{eq:decomposition-un}
    U_n(t) = \hat{U}_n(t) \oplus \id_{\hilbert_n^\perp} \, ,\qquad \hat{U}_n(t)=\e^{-\iu t\hat{H}_n}.
\end{equation}
\begin{figure}
    \centering
    \includegraphics[width=0.6\linewidth]{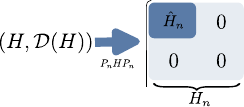}
    \caption{Pictorial representation of the Galerkin approximation. The unbounded operator $H$ on the infinite-dimensional space $\hilbert$ is truncated by means of a family of finite-dimensional projectors $(P_n)_{n\in\mathbb{N}}$: this yields the bounded operator $H_n = P_n H P_n$ on $\hilbert$, which only acts nontrivially on the finite-dimensional subspace $\hilbert_n=P_n\hilbert$, where its action is described by $\hat{H}_n$, and zero everywhere else.}
    \label{fig:matrix}
\end{figure}

We are interested in sufficient conditions under which $U_n(t)$ converges strongly to $U(t)$, that is, $U_n(t)\psi\to U(t)\psi$ for all $\psi\in\hilbert$ and $t\in\mathbb{R}$. We begin by recalling a simple sufficient criterion: Galerkin's method works when the subspace of $\hilbert$ defined by
\begin{equation}\label{eq:hilbert_fin}
    \hilbert_{\mathrm{fin}}:=\bigcup_{n \in \nnum} P_n \hilbert
\end{equation}
is a \textit{core} of $H$, that is: the restriction of $H$ to this space is still essentially self-adjoint, i.e. it admits $H$ as its unique self-adjoint extension:

\begin{proposition}
    \label{prop:convergence-core}
    Let $H$ be a self-adjoint operator on $\hilbert$, $(P_n)_{n \in \nnum}$ a family of projectors as in Definition~\ref{def:galerkin_approximation}, additionally satisfying $P_n P_m = P_n$ for all $m \geq n$, and $U_{n}(t)=\e^{-\iu tH_{n}}$ the unitary propagator associated with the $n$th Galerkin approximation $H_n$ of $H$. Assume that the restriction $H_{\mathrm{fin}}$ of $H$ onto the space $\hilbert_{\mathrm{fin}}$ as per Eq.~\eqref{eq:hilbert_fin} is essentially self-adjoint. Then, for all $t\in\mathbb{R}$ and $\psi\in\hilbert$,
    \begin{equation}\label{eq:dynamical_convergence}
        \lim_{n\to\infty}\left\|U(t)\psi-U_{n}(t)\psi\right\|  = 0
    \end{equation}
    where $U(t) = \e^{-\iu tH}$ is the unitary propagator generated by $H$.
\end{proposition}
\begin{proof}
We recall that the property~\eqref{eq:dynamical_convergence} is equivalent to strong resolvent convergence of $H_n$ to $H$ (cf.~\cite[287]{reed-mmmp1-funkana-1980}), and that, by~\cite[Proposition~10.1.18]{oliveira-intermediatespectraltheory-2009}, a sufficient condition for strong resolvent convergence of a family of self-adjoint operators $(H_n)_{n\in\mathbb{N}}$ to a self-adjoint operator $H$ is the following~\cite[Proposition~10.1.18]{oliveira-intermediatespectraltheory-2009}: there exists $\mathcal{D}_0\subset\domain(H)$ such that
\begin{enumerate}
\item[(i)] $\mathcal{D}_0$ is a core of $H$;
\item[(ii)] $\mathcal{D}_0\subset\domain(H_n)$ for all $n\in\mathbb{N}$;
\item[(iii)] for all $\psi\in\domain_0$, $H_n\psi\to H\psi$.
\end{enumerate}
We choose $\domain_0=\hilbert_{\rm fin}$. Then (i) is true by assumption, and (ii) is obvious since $\domain(H_n)=\hilbert$ for every $n\in\mathbb{N}$. To prove (iii), let $\psi \in \hilbert_{\mathrm{fin}}$.
    Then, as $P_n P_m = P_n$ for $m \geq n$, there exists $N \in \nnum$ such that $\psi \in P_n \hilbert$ for all $n \geq N$, and
    \begin{equation}
        \lim_{n \to \infty} H_n \psi = \lim_{n \to \infty} P_n H P_n \psi = \lim_{n \to \infty} P_n H  \psi = H \psi\, .
    \end{equation}
    Thus $\lim_{n \to \infty}H_n \psi = H \psi$ for all $\psi \in \hilbert_{\mathrm{fin}}$, which proves that (iii) holds. This completes the proof.
\end{proof}

\begin{remark}
\label{rem:eigenbasis-convergence}
    Proposition~\ref{prop:convergence-core} covers the particular case in which $H$ has a purely discrete spectrum, thus admitting an orthonormal basis of eigenvectors\footnote{For definiteness, whenever considering an arbitrary orthonormal basis, the index $l$ will range on all natural numbers. However, when considering specific bases (like the associated Legendre polynomials, cf.~Definition~\ref{def:legendre}), the index $l$ could range on a different set of integers.} $(\phi_l)_{l\in\mathbb{N}}$, and the projectors $(P_n)_{n\in\mathbb{N}}$ correspond precisely to the projection on the span of the first $n$ eigenvectors of $H$,
    \begin{equation}
        P_n=\sum_{l=0}^{n-1}\braket{\phi_l,\cdot}\phi_l,\qquad H\phi_l=E_l\phi_l,
    \end{equation}
    with $E_j\in\mathbb{R}$ being the $j$th eigenvalue of $H$.
    Each of these projectors has a finite-dimensional span, and $P_n\hilbert\subset\domain(H)$ since each eigenvector is in $\domain(H)$.
    Then $\hilbert_{\mathrm{fin}}$, as defined in Proposition~\ref{prop:convergence-core}, is a core of $H$~\cite[p.~78]{teschl-mathematicalmethodsquantum-2009}. As such, as anticipated in Section~\ref{sec:results_galerkin_informal}, the Galerkin method always works when truncating $H$ in its own eigenbasis.
\end{remark}

As already pointed out in Section~\ref{sec:results_galerkin_informal}, here we are instead interested in the situation in which the sufficient criterion presented above does not apply---that is, $H$ is not essentially self-adjoint on $\hilbert_{\rm fin}$.
In such cases, more sophisticated arguments will be necessary. We begin by recalling two basic definitions:
\begin{definition}\label{def:coercive-semibounded}
Let $A:\domain(A)\subset\hilbert\rightarrow\hilbert$ be a densely defined symmetric operator on $\hilbert$. Then:
\begin{itemize}
    \item $A$ is \textit{bounded from below} (or semibounded) if there exists $\gamma\in\mathbb{R}$ such that
    \begin{equation}
        \braket{\psi,A\psi}\geq\gamma\|\psi\|^2\qquad\text{for all }\psi\in\domain(A);
    \end{equation}
    \item $A$ is \textit{coercive} if the equation above holds with $\gamma>0$.
\end{itemize}    
\end{definition}
\begin{proposition}\label{prop:coercive-implies-invertibility}
    Let $A$ be a coercive self-adjoint operator. Then $A$ admits a bounded inverse with $\|A^{-1}\|\leq\frac{1}{\gamma}$.
\end{proposition}
\begin{proof}
    As $A$ is coercive and self-adjoint, $0$ is not in its spectrum~\cite[Theorem 2.19]{teschl-mathematicalmethodsquantum-2009}, hence $A$ admits a bounded inverse $A^{-1}$.
    Let $0\neq\psi\in\domain(A)$. By the Cauchy--Schwarz inequality,
    \begin{equation}
    \label{proofeq:coercive-bound}
        \norm{A \psi} = \frac{\norm{\psi} \norm{A \psi}}{\norm{\psi}} \geq \frac{\braket{\psi,A\psi}}{\norm{\psi}} \geq \gamma \norm{\psi} \, .
    \end{equation}
    Furthermore, given $\psi\in\hilbert$,
    \begin{equation}
        \gamma \norm{A^{-1}\psi} \leq  \norm{A A^{-1} \psi} = \norm{\psi} \quad \text{for all }\psi \in \hilbert,
    \end{equation}
    where we applied \cref{proofeq:coercive-bound} onto $A^{-1}\psi\in\domain(A)$.
\end{proof}

\begin{remark}\label{rem:waiving_coerciveness}
A central object in the study of Galerkin approximations are \textit{Galerkin projectors} (Definition~\ref{def:galerkin_projector}), which require the operator to admit a bounded inverse. As such, to keep the notation simple, we will only define them for coercive self-adjoint operators, which automatically admit a bounded inverse by Proposition~\ref{prop:coercive-implies-invertibility}. However, in the more general case of operators bounded from below with $\gamma\in\mathbb{R}$, one can simply study the operator $A-\gamma+1$, which is clearly coercive:
\begin{equation}
    \braket{\psi,(A-\gamma+1)\psi}=\braket{\psi,A\psi}+(1-\gamma)\|\psi\|^2\geq\|\psi\|^2,
\end{equation}
since the dynamics induced by $A$ and $A-\gamma+1$ only differ by an immaterial phase term: $\e^{-\iu t(A-\gamma+1)}=\e^{-\iu(1-\gamma)t}\e^{-\iu tA}$. We will use this simple argument in the proof of Proposition~\ref{prop:galerkin_criterion}.
\end{remark}

\begin{definition}[Galerkin projector]\label{def:galerkin_projector}
Let $H:\domain(H)\subset\hilbert\rightarrow\hilbert$ be a coercive self-adjoint operator on $\hilbert$, and $P_n$, $H_n$, $\hat{H}_n$ as in Definition~\ref{def:galerkin_approximation}. The $n$th \textit{Galerkin projector} is the operator on $\hilbert$ defined by
        \begin{equation}\label{eq:galerkin_projector}
            Q_n : \domain(H) \to \hilbert\, ,  \quad Q_n \psi =  R_n H \psi  \,,
        \end{equation}
        where
         \begin{equation}\label{eq:rn}
            R_n : \hilbert \to \hilbert\, , \quad R_n = \hat{H}_n^{-1} \oplus 0_{\hilbert_n^\perp} \, .
        \end{equation}
\end{definition}
 We point out that $Q_n$ is generally a non-self-adjoint projector (cf.~Lemma~\ref{lem:rn-prop}), that is, $Q_n^2=Q_n\neq Q_n^*$. The definition above is indeed well-posed: by Proposition~\ref{prop:coercive-implies-invertibility} $H$ admits a bounded inverse, thus necessarily each operator $\hat{H}_n$ is also invertible in the finite-dimensional space $\hilbert_n$. Pay attention to the fact that, instead, $H_n$ is never invertible in $\hilbert$ (also see Remark~\ref{rem:galerkin_projection}).

We present our first general convergence statement about Galerkin approximations: it consists of an adaptation to the Hilbert space scenario of a result by Parter, cf.~\cite{parter-rolesstabilityconvergence-1980}:
\begin{proposition}\label{prop:galerkin_criterion}
     Let $H:\domain(H)\subset\hilbert\rightarrow\hilbert$ be a coercive self-adjoint operator on $\hilbert$, and $H_n$ and $Q_n$ the Galerkin approximation and projector corresponding to a family of projectors $(P_n)_{n\in\mathbb{N}}$ as per Definition~\ref{def:galerkin_approximation}. Assume
    \begin{equation}
        \lim_{n\to\infty}Q_n\psi=\psi\qquad\text{for all }\psi\in\domain(H).
    \end{equation}
    Then, for all $t\in\mathbb{R}$,
    \begin{equation}
        \lim_{n\to\infty}U_n(t)\psi=U(t)\psi\qquad\text{for all }\psi\in\hilbert,
    \end{equation}
    where $U_n(t)$ and $U(t)$ are the unitary propagators generated respectively by $H_n$ and $H$, i.e. $U_n(t)=\e^{-\iu tH_n}$ and $U(t)=\e^{-\iu tH}$.
\end{proposition}
We provide an explicit proof in Section~\ref{sec:galerkin-criterion-proof}.

Proposition~\ref{prop:galerkin_criterion} provides a sufficient condition for the convergence of the dynamics generated by Galerkin approximations: given a coercive self-adjoint $H$, and its Galerkin approximations $(H_n)_{n\in\mathbb{N}}$ obtained through a family of projectors $(P_n)_{n\in\mathbb{N}}$ such that $P_n\hilbert\subset\mathcal{D}(H)$, then a sufficient condition from $U_n(t)$ converging strongly to $U(t)$ is $Q_n\psi\to\psi$ for all $\psi\in\mathcal{D}(H)$, with $Q_n$ being the $n$th Galerkin approximation associated with $P_n$. 
In practical applications, computing $Q_n$ explicitly---and thus, using Proposition~\ref{prop:galerkin_criterion} to show convergence of Galerkin approximations---might be unfeasible; it would then be nice to have sufficient criteria for the property $Q_n\psi\to\psi$ to hold. We will present one such condition that crucially involves the concept of \textit{Friedrichs extension} of a symmetric operator. To this end, we shall briefly recall some basic facts about the relation between sesquilinear forms and linear operators. 

Let $A:\mathcal{D}(A)\subset\hilbert\rightarrow\hilbert$ be a (possibly not self-adjoint) symmetric operator bounded from below, cf.~Definition~\ref{def:coercive-semibounded}. As such, the sesquilinear form $q_A:\mathcal{D}(A)\times\mathcal{D}(A)\rightarrow\mathbb{C}$ defined by $q_A(\psi,\varphi):=\braket{\psi,A\varphi}$ satisfies $q_A(\psi,\psi)\geq\gamma\|\psi\|^2$, or equivalently, satisfies $q_{A-\gamma}(\psi,\psi):=\braket{\psi,(A-\gamma)\psi}\geq0$. As such, one can consider a stronger norm on $\mathcal{D}(A)$ defined by
\begin{equation}\label{eq:plusnorm}
\|\varphi\|^2_{+}:=q_{A-\gamma}(\varphi,\varphi)+\|\varphi\|^2,\qquad\varphi\in\mathcal{D}(A).
\end{equation}
We remark that, for positive operators, this norm is equivalent to the graph norm of $A^{1/2}$ \cite[p.~227]{schmudgen-unboundedselfadjointoperators-2012}. 
The following proposition allows us to construct a distinct extension of $A$.
\begin{proposition}
\label{prop:def-friedrich}
\emph{(\cite[p.~177]{reed-mmmp2-fourier-1975}~\cite[p.~70]{teschl-mathematicalmethodsquantum-2009})} The following facts hold:
\begin{itemize}
    \item  $q_{A-\gamma}$ admits a unique extension $\tilde{q}_{A-\gamma}$ to a sesquilinear form defined on $\mathcal{D}(\tilde{q}_{A-\gamma}):=\overline{\mathcal{D}(A)}^{\|\cdot\|_+}$, the closure of $\mathcal{D}(A)$ with respect to the norm $\|\cdot\|_+$;
    \item There is a unique self-adjoint operator $\tilde{A}:\mathcal{D}(\tilde{A})\subset\hilbert\rightarrow\hilbert$, with $\mathcal{D}(\tilde{A})\subset\mathcal
    D(\tilde{q}_{A-\gamma})$, such that $\tilde{q}_{A-\gamma}(\psi,\varphi)=\braket{\psi,(\tilde{A}-\gamma)\varphi}$ for all $\psi,\varphi\in\mathcal{D}(\tilde{A})$.
\end{itemize}
    Besides, $\tilde{A}$ is a self-adjoint extension of $A$ which is also bounded from below by $\gamma$, and it is its unique self-adjoint extension such that $\mathcal{D}(\tilde{A})\subset\domain(\tilde{q}_{A-\gamma})$.
\end{proposition}

\begin{definition}\label{def:friedrichs_extension}
    Let $A:\mathcal{D}(A)\subset\hilbert\rightarrow\hilbert$ be a densely defined, symmetric operator bounded from below. Then the self-adjoint extension $\tilde{A}$ as given by Proposition~\ref{prop:def-friedrich} is the \textit{Friedrichs extension} of $A$.
\end{definition}

The following proposition showcases the central role of the Friedrichs extension in approximation theory:
\begin{proposition}
    \label{prop:friedrichs}
   Let $H:\domain(H)\subset\hilbert\rightarrow\hilbert$ be a self-adjoint operator bounded from below, and $(P_n)_{n \in \nnum}$ a family of projectors as in Definition~\ref{def:galerkin_approximation} additionally satisfying $P_n P_m = P_n$ for all $n \geq m$.
   Let $\hilbert_{\rm fin}$ be the subspace of $\hilbert$ defined in Eq.~\eqref{eq:hilbert_fin}, and $H_{\rm fin}$ be the restriction of $H$ to $\hilbert_{\rm fin}$.   
    Then $H_{\rm fin}$ is densely defined and symmetric, and for all $\psi\in\hilbert$ and $t\in\mathbb{R}$ we have
    \begin{equation}
        \lim_{n\to\infty}U_n(t)\psi=\tilde{U}(t)\psi,
    \end{equation}
where $\tilde{U}(t)$ is the unitary propagator generated by the Friedrichs extension $\tilde{H}_{\rm fin}$ of $H_{\rm fin}$, cf.~Definition~\ref{def:friedrichs_extension}.
\end{proposition}

This proposition could be derived from~\cite[Theorem~2.12.6, p.~101]{miklavcic-appliedfunctionalanalysis-2001}, which deals with Galerkin approximations of a more general class of sectorial operators on a Hilbert space. However, such a derivation is not immediate and we provide a direct proof in Section~\ref{sec:proof-friedrichs}.

As anticipated in Section~\ref{sec:results_galerkin_informal}, one clearly sees that, in general, the dynamics induced by the Galerkin approximations of $H$ will not converge to $U(t)=\e^{-\iu tH}$, unless $H$ actually coincides with $\tilde{H}_{\rm fin}$. 
In agreement with Proposition~\ref{prop:convergence-core}, this is always the case if $H_{\mathrm{fin}}$ is essentially self-adjoint, i.e. $\hilbert_{\mathrm{fin}}$ is a core of $H$: in this case $H$ and $\tilde{H}_{\rm fin}$ are both self-adjoint extensions of the essentially self-adjoint operator $H_{\mathrm{fin}}$, and thus equal. In such cases, Galerkin's method always works. If, instead, $H_{\rm fin}$ has multiple---and thus, infinitely many---self-adjoint extensions, Galerkin's method will produce convergence to the dynamics generated by one privileged self-adjoint extension of $H_{\rm fin}$ (the Friedrichs one), which may or may not coincide with the operator $H$ whose dynamics we want to reproduce. In the following section we will examine a situation in which, indeed, $H_{\rm fin}$ has infinitely many self-adjoint extensions---so that the validity of Galerkin's method is at stake.

\begin{remark}\label{rem_nomenclature}
    When the projectors $P_n$ are obtained by an orthonormal basis $\Phi=(\phi_l)_{l\in\mathbb{N}}$ through
    \begin{equation}
        P_n=\sum_{l=0}^{n-1}\braket{\phi_l,\cdot}\phi_l,
    \end{equation}
    then the space $\hilbert_{\rm fin}$ as defined in Eq.~\eqref{eq:hilbert_fin} is simply $\operatorname{Span}\Phi$, the space of finite linear combinations of elements of the basis. This is exactly the scenario we will encounter in Section~\ref{sec:particle_results}, where, to stress the dependence on the choice of basis, the operator $H_{\rm fin}$ shall be simply denoted by $H_\Phi$.
\end{remark}

\section{Galerkin approximations of the particle in a box}\label{sec:particle_results}

 In this section we shall apply the results of the previous section to the quantum particle in a box. We will prove the existence of a class of orthonormal bases satisfying Dirichlet boundary conditions and possibly other boundary conditions, and having the following property: the dynamics generated by the corresponding Galerkin approximation always converge to the one corresponding to Dirichlet boundary conditions (Proposition~\ref{prop:legendre-dirichlet-general}); these include, as a notable example of ``boundary-blind'' bases, associated Legendre polynomials (Theorem~\ref{thm:legendre-dirichlet}).
 
 Furthermore, we show that such bases can be minimally modified in such a way that the dynamics generated by said approximations converge to the one corresponding to other boundary conditions---namely, $\alpha$-periodic boundary conditions (Proposition~\ref{prop:alpha-periodic})---despite again satisfying other boundary conditions.

\subsection{Preliminaries}
\label{sec:particle-in-a-box}
Following the discussion in Section~\ref{sec:results_box_informal}, we now set $\hilbert=L^2((-1,1))$, and we use the symbols $\braket{\cdot,\cdot}$, $\|\cdot\|$ to denote the standard $L^2$ scalar product and associated norm. In order to set up the notation---and for the convenience of readers not familiar with this language---we will recall here some known properties of Sobolev spaces, see e.g.~\cite{evans-partialdifferentialequations-2010,jost-partialdifferentialequations-2002,brezis2011functional,mazya-sobolevspacesapplications-2011}, and self-adjoint realizations of the Laplace operator in a one-dimensional interval; our results will be presented in Section~\ref{sec:dirichlet}--\ref{sec:alpha-periodic}.

\begin{definition}[Sobolev spaces]\label{def:sobolev_spaces}
    Let $k\in\mathbb{N}$. The Sobolev space of order $k$, $\mathrm{H}^k((-1,1))$, is defined by
\begin{equation}
    \mathrm{H}^k((-1,1))=\left\{\varphi\in L^2((-1,1)):\varphi^{(\alpha)}\in L^2((-1,1)),\;\alpha=1,\dots,k\right\},
\end{equation}
\end{definition}
Above, $\varphi^{(\alpha)}$ is the $\alpha$th (distributional) derivative of $\varphi$, which is well-defined since all functions in $L^2((-1,1))$ are locally integrable. For $\alpha=1,2$, we shall denote the corresponding derivatives by $\varphi'$ and $\varphi''$. Each of these spaces, equipped with the scalar product and associated norm defined by
\begin{equation}\label{eq:sobolev}
    \Braket{\psi,\varphi}_{\mathrm{H}^k}=\sum_{\alpha=0}^k\braket{\psi^{(\alpha)},\varphi^{(\alpha)}},\qquad \|\varphi\|_{\mathrm{H}^k}=\left(\sum_{\alpha=0}^k\|\varphi^{(\alpha)}\|^2\right)^{1/2},
\end{equation}
is a Hilbert space.

Recall that the elements of $L^2((-1,1))$ are, strictly speaking, equivalence classes of square-integrable functions that are almost everywhere equal, i.e., each class contains functions that differ in a set of zero Lebesgue measure. The following known result tells us that each element of $\mathrm{H}^1((-1,1))$ admits a unique continuous representative:

\begin{proposition}\label{prop:continuous_representative}\emph{\cite[Theorem~8.2]{brezis2011functional}}
    Let $\varphi\in\mathrm{H}^1((-1,1))$. Then there exists a unique continuous function $\tilde{\varphi}:[-1,1]\rightarrow\mathbb{C}$ that belongs to the equivalence class of $\varphi$, i.e., $\tilde{\varphi}=\varphi$ almost everywhere in $(-1,1)$.
\end{proposition}
As a consequence, each element of $\mathrm{H}^1((-1,1))$ (and so for higher-order Sobolev spaces) can be identified, with a slight abuse of notation, with a continuous function. In the following we will not distinguish between $\varphi$ and its continuous representative; this will allow us to write expressions like $\varphi(x)$ for $x\in[-1,1]$, including the two boundary points $\pm1$. We can thus provide the following additional definition:
\begin{definition}\label{def:sobolev_h10}
    The space $\mathrm{H}^1_0((-1,1))$ is defined by
    \begin{equation}
        \mathrm{H}^1_0((-1,1))=\left\{\varphi\in\mathrm{H}^1((-1,1)):\varphi(\pm1)=0\right\}.
    \end{equation}
\end{definition}
This is a closed subspace of $\mathrm{H}^1((-1,1))$, corresponding to the closure of the space of differentiable and compactly supported functions with respect to the norm of $\mathrm{H}^1((-1,1))$~\cite[Theorem~8.12]{brezis2011functional}, and is thus itself a Hilbert space. Functions in this space satisfy an inequality which will turn useful for our purposes:
\begin{proposition}[Poincaré's inequality]
    \label{prop:poincare}
    Let $\psi \in \mathrm{H}^{1}_0((-1,1))$. Then 
    \begin{equation}
        \norm{\psi} \leq \norm{\psi'} \, .
    \end{equation}
\end{proposition}
For completeness, we provide an explicit proof here. Analogous inequalities also hold in higher dimensions, see e.g.~\cite[p.~166]{jost-partialdifferentialequations-2002}.
\begin{proof}
    As $\psi\in\mathrm{H}^1((-1,1))$ and $\psi(-1)=0$, we have
    \begin{align}
        \psi(x) & = \int_{-1}^x \psi'(s) \dl s \qquad \forall x \in (-1,0],
    \end{align}
    whence we get
    \begin{align}\label{eq:sup}
        \sup_{x\in(-1,0]}|\psi(x)|\leq\int_{-1}^0|\psi'(s)|\;\mathrm{d}s,
    \end{align}
    and therefore
    \begin{align}\label{eq:ineq-left}
        \int_{-1}^0|\psi(s)|^2\;\mathrm{d}s&\leq\sup_{x\in(-1,0]}|\psi(x)|^2\int_{-1}^0\;\mathrm{d}s\nonumber\\
        &\leq\left(\int_{-1}^0|\psi'(s)|\;\mathrm{d}s\right)^2\nonumber\\
        &\leq\int_{-1}^0|\psi'(s)|^2\;\mathrm{d}s\;\int_{-1}^0\mathrm{d}s'
        =\int_{-1}^0|\psi'(s)|^2\;\mathrm{d}s,
    \end{align}
    where we used Eq.~\eqref{eq:sup} and, in the last step, the Cauchy--Schwarz inequality. Since we also have $\psi(+1)=0$, repeating the same procedure in $[0,1)$ yields
    \begin{align}\label{eq:ineq-right}
        \int_{0}^1|\psi(s)|^2\;\mathrm{d}s\leq\int_{0}^1|\psi'(s)|^2\;\mathrm{d}s.
    \end{align}
    Combining Eqs.~\eqref{eq:ineq-left}--\eqref{eq:ineq-right} yields the desired inequality.
\end{proof}
Notice that, as a direct result of this inequality, the function $\psi\mapsto\|\psi'\|$ defines a norm on $\mathrm{H}^1_0((-1,1))$ equivalent to $\|\cdot\|_{\mathrm{H}^1}$.

We now turn to the realizations of the Laplace operator on $L^2((-1,1))$, reprising the discussion in Section~\ref{sec:results_box_informal}. Let us recall some known facts:

\begin{proposition}\label{prop:basic_facts_laplacian}
    For every $2\times2$ unitary matrix $W\in\mathrm{U}(2)$, the operator $H_W:\domain(H_W)\subset L^2((-1,1))\rightarrow L^2((-1,1))$ defined by
\begin{align}
\label{eq:hamiltonian-boundaries}
    \domain(H_W)&=\left\{\varphi\in\mathrm{H}^2((-1,1)): \iu(\id+W)\begin{pmatrix}
        -\varphi'(-1)\\\varphi'(+1)
    \end{pmatrix}=(\id-W)
    \begin{pmatrix}
        \varphi(-1)\\\varphi(+1)
    \end{pmatrix}\, 
    \right\};\\
    H_W\varphi&=-\varphi''
\end{align}
 is self-adjoint, and has a purely discrete spectrum bounded from below.
\end{proposition}
\begin{proof}
   Self-adjointness of $H_W$ is a standard result, see e.g.~\cite{bonneau-selfadjointextensionsoperators-2001,asorey-dynamicalcompositionlaw-2013,asorey-globaltheoryquantum-2005a}, which comes from von Neumann's theory of self-adjoint extensions. For the spectrum in the general case, 
   see e.g.~\cite{bonneau-selfadjointextensionsoperators-2001}.
\end{proof}
 Each operator $H_W$ is a distinct self-adjoint realization of the Laplace operator, corresponding to a specific choice of \textit{boundary conditions} to be satisfied by all functions in their domain; as such, each generates a distinct unitary group $U_W(t)=\e^{-\iu tH_W}$, cf.~Example~\ref{ex:dirichlet_vs_neumann}.
 For our purposes, let us introduce some specific kinds of self-adjoint realizations of the Laplace operator.
 \begin{definition}\label{def:particular_bcs}
     We define the following boundary conditions:
     \begin{equation}
\begin{array}{rll}
    \text{Dirichlet b.c.:}\quad & W=-\id &\Rightarrow \varphi(-1) = \varphi(1) = 0;\\
    \text{Neumann b.c.:}\quad & W=+\id &\Rightarrow \varphi'(-1) = \varphi'(1) = 0;\\
    \text{$\alpha$-periodic b.c.:}\quad &  W=W(\alpha) &\Rightarrow \varphi(-1) = \e^{\iu\alpha}\varphi(1), \;\varphi'(-1) = \e^{\iu\alpha}\varphi'(1),
\end{array}
\end{equation}
where $0\leq\alpha<2\pi$ and
\begin{equation} \label{eq:w}
    W(\alpha)=\begin{pmatrix}
        0&\e^{-\iu\alpha}\\ \e^{\iu\alpha}&0
    \end{pmatrix}.
\end{equation}
In particular, $\alpha$-periodic boundary conditions corresponding to $\alpha=0$ and $\alpha=\pi$ are respectively denoted as \textit{periodic} and \textit{antiperiodic} boundary conditions.
\end{definition}
\begin{example}[Eigenvalues and eigenvectors of different boundary conditions]
\label{ex:bc-eigenfunctions}
   We recall that, for the three particular boundary conditions as per Definition~\ref{def:particular_bcs}, one can explicitly compute eigenvalues and eigenvectors:
    \begin{equation}
        \begin{array}{rlll}
            \text{Dirichlet b.c.:} & E_j = \frac{j^2 \pi^2}{4} & \psi_j(x) = \sin\left(\frac{j\pi}{2}(x+1)\right) & j \geq 1;\\
            \text{Neumann b.c.:} & E_j = \frac{j^2 \pi^2}{4} & \psi_j(x) = \cos\left(\frac{j\pi}{2}(x+1)\right) & j \geq 0;\\
            \text{Periodic b.c.:} & E_j = j^2 \pi^2 & \psi_j(x) = \frac{1}{\sqrt{2}}\e^{\iu j \pi (x+1)} & j \in \znum.
        \end{array}
    \end{equation}
\end{example}
 \begin{example}[Evolution induced by distinct boundary conditions]\label{ex:dirichlet_vs_neumann}
    In order to stress how drastically the choice of boundary conditions can influence the dynamics, let us compare the evolutions corresponding to Dirichlet and periodic boundary conditions, cf.~Definition~\ref{def:particular_bcs}. For definiteness, let us adopt the self-explanatory notations $H_{\rm Dir},H_{\rm per}$ for the two Hamiltonians, and $U_{\rm Dir}(t),U_{\rm per}(t)$ for the corresponding unitary propagators. Given $\psi_0\in L^2((-1,1))$, let
    \begin{equation}
        \psi_{\mathrm{Dir}}(x;t) = (U_{\mathrm{Dir}}(t)\psi_0)(x),\qquad  \psi_{\mathrm{per}}(x;t) = (U_{\mathrm{per}}(t)\psi_0)(x).
    \end{equation}
    For times $t = \frac{p}{q} \frac{4}{\pi}$ with rational prefactors $\frac{p}{q}$, $p,q \in \nnum$, $\psi_{\mathrm{Dir}}(x;t)$ is explicitly calculated in~\cite{berry-quantumfractalsboxes-1996}.
In particular, at $t = 4/ \pi$, the time evolutions are given by 
\begin{equation}
    \psi_{\mathrm{Dir}}(x;4/ \pi) = -\psi_0(-x) \quad \quad \psi_{\mathrm{per}}(x;4/ \pi) = \psi_0(x) \, ,
\end{equation}
see Appendix~\ref{sec:app-exact-timeevolution} for an explicit proof. 
As such, 
the time evolutions of an initially symmetric wavefunction $\psi_0(-x) = \psi_0(x)$ at $t = 4/\pi$
are maximally different, that is, 
$\norm*{\psi_{\mathrm{Dir}}(x;4/ \pi)-\psi_{\mathrm{per}}(x;4/ \pi)} = 2$, cf.~Corollary~\ref{corr:app-dirichlet-periodic-maxdifferent}.
The two processes are visually depicted in Figure~\ref{fig:dir_periodic_sin}.
This result agrees with our intuition: hard walls reflect the particle, causing the wavefunction to undergo a phase change of $\pi$ and being mirrored, which eventually results in $\psi_{\mathrm{Dir}}(x;4/ \pi) = -\psi_0(-x)$.
Instead, at the same time, the particle in a periodic box returns to its original position $\psi_{\mathrm{per}}(x;4/ \pi) = \psi_0(x)$. 
\begin{figure}[ht]
    \centering
    \includegraphics[width=0.65\linewidth]{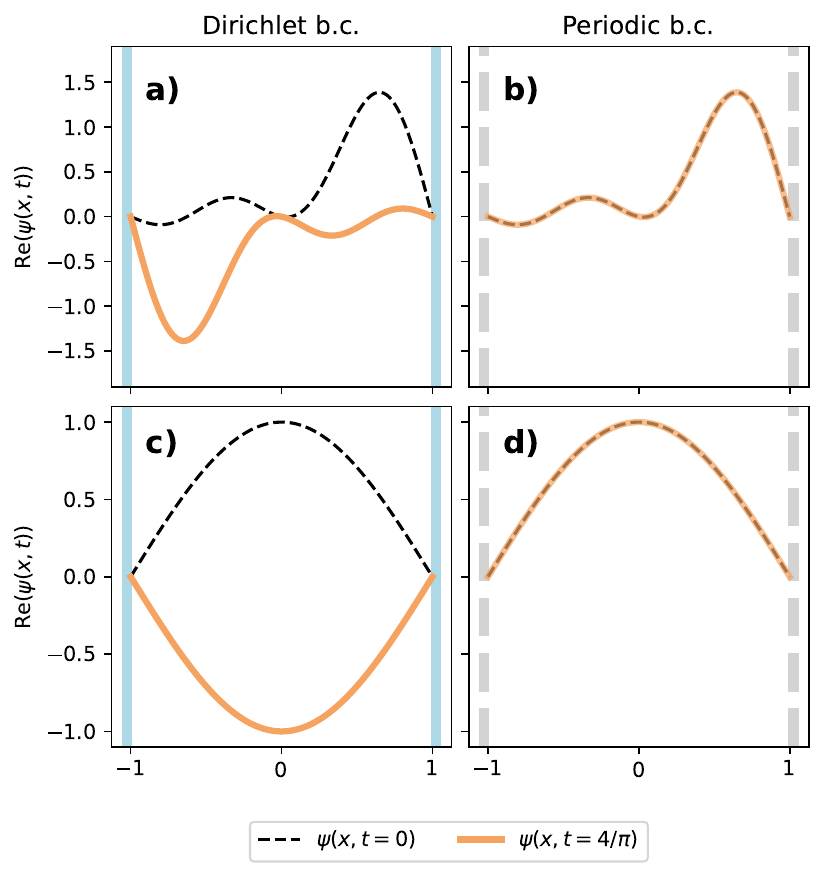}
    \caption{Time evolution of an arbitrary initial wavefunction $\psi_0(x)$ (\textbf{a)} and \textbf{b)}) and the specific symmetric initial wavefunction $\psi_0(x) = \sin(\pi/2(x+1))$  (\textbf{c)} and \textbf{d)}) in a box with hard walls (\textbf{a)} and \textbf{c)}) and with periodic boundary conditions (\textbf{b)} and \textbf{d)}).
    All functions are evaluated at $t = 4/ \pi$, and the imaginary part of all wavefunctions is zero.}
\label{fig:dir_periodic_sin}
\end{figure}
\end{example}

\begin{remark}\label{rem:dirichlet}
     We note that the domain of the Dirichlet Laplacian $H_{\rm Dir}$, cf.~Definition~\ref{def:particular_bcs}, corresponds to
     \begin{equation}
         \domain(H_{\rm Dir})=\mathrm{H}^1_0((-1,1))\cap\mathrm{H}^2((-1,1)),
     \end{equation}
     since $\mathrm{H}^2((-1,1))\subset\mathrm{H}^1((-1,1))$ and the boundary condition $\varphi(\pm1)=0$ corresponds precisely to the one defining the subspace $\mathrm{H}^1_0((-1,1))$ of $\mathrm{H}^1((-1,1))$, cf.~Definition~\ref{def:sobolev_h10}.
 \end{remark}

\subsection{Main results: convergence to the Dirichlet Laplacian}\label{sec:dirichlet}

Given $W\in\mathrm{U}(2)$, we now consider Galerkin approximations $H_{n,W} = P_n H_W P_n$ of the corresponding realization $H_W$ of the Laplacian. Given a complete orthonormal set $(\phi_l)_{l \in \nnum}$ of $\hilbert$, we define $P_n$ as the projector onto the $n$-dimensional vector space spanned by the first $n$ vectors of the basis,
\begin{equation}\label{eq:proj_by_basis}
    P_n=\sum_{l=0}^{n-1}\braket{\phi_l,\cdot}\phi_l.
\end{equation}
The first result we present is a sufficient condition under which Galerkin approximations of the Laplace operator corresponding to some basis $(\phi_l)_{l\in\mathbb{N}}$ converge to the Dirichlet Laplacian---regardless whether they also satisfie other boundary conditions. We start by the following crucial property:
\begin{proposition}
\label{prop:friedrichs-dirichlet}
    Let $\Phi=(\phi_l)_{l\in\mathbb{N}}$ be a complete orthonormal set of $L^2((-1,1))$ satisfying the following additional properties:
    \begin{itemize}
        \item[(i)] $\Phi\subset\mathrm{H}^1_0((-1,1))\cap\mathrm{H}^2((-1,1))$;
        \item[(ii)] $\overline{\operatorname{Span}\Phi'}=\{1\}^\perp$,
    \end{itemize}
    where $\Phi'=(\phi'_l)_{l\in\mathbb{N}}$, and $1$ denotes the function with constant value $1$. Then the symmetric operator $H_\Phi$ defined by
    \begin{equation}
        \domain (H_\Phi)=\operatorname{Span}\Phi,\qquad H_\Phi\psi=-\psi''
    \end{equation}
    has a Friedrichs extension equal to the Dirichlet Laplacian $H_{\mathrm{Dir}}$ (cf.~Remark~\ref{rem:dirichlet}):
    \begin{equation}
        \domain (H_{\mathrm{Dir}})=\mathrm{H}^1_0((-1,1))\cap\mathrm{H}^2((-1,1)),\qquad H_{\mathrm{ Dir}}\psi=-\psi''.
    \end{equation}
\end{proposition}
This result is proven in Section~\ref{sec:proof_prof_friedrichs-dirichlet}.

\begin{remark}
    Notice that $\Phi'=(\phi'_l)_{l\in\mathbb{N}}$ is a legitimate set in $L^2((-1,1))$ since each $\phi_l$ is in the first Sobolev space; besides, each $\phi'_l$ satisfies
    \begin{equation}
        \braket{1,\phi'_l}=\int_{-1}^{+1}\phi'_l(x)\;\mathrm{d}x=\phi(1)-\phi(-1)=0,
    \end{equation}
    whence the inclusion $\overline{\operatorname{Span}\Phi'}\subset\{1\}^\perp$ already follows from (i) as direct consequence of Dirichlet boundary conditions; the converse inclusion is, instead, to be assumed.
\end{remark}

\begin{remark}
The fact that the Dirichlet realization of the Laplacian corresponds to the Friedrichs extension of the Laplacian initially defined on $\operatorname{Span}\Phi$ mirrors---but is not implied by---the known fact that the Dirichlet Laplacian also coincides with the Friedrichs extension of the Laplacian when initially defined on the space $C^\infty_0((-1,1))$ of smooth compactly supported functions, see e.g.~\cite[Example 1]{reed-mmmp2-fourier-1975}.
\end{remark}

 \begin{proposition}
        \label{prop:legendre-dirichlet-general}
        Let $\Phi=(\phi_l)_{l\in\mathbb{N}}$ be an orthonormal basis of $L^2((-1,1))$ satisfying the following conditions:
        \begin{itemize}
        \item[(i)] $\Phi\subset\mathrm{H}^1_0((-1,1))\cap\mathrm{H}^2((-1,1))$;
        \item[(ii)] $\overline{\operatorname{Span}\Phi'}=\{1\}^\perp$,
    \end{itemize}
    and let $H_n$ be the $n$th corresponding Galerkin approximation of the Laplace operator, cf.~Definition~\ref{def:galerkin_approximation}. Then, for all $\psi\in L^2((-1,1))$ and $t\in\mathbb{R}$,
        \begin{equation}
            \lim_{n \rightarrow \infty}\norm{\left(U_n(t)-U_{\mathrm{Dir}}(t)\right)\psi} = 0 \,
        \end{equation}
        where $U_{\mathrm{Dir}}(t)=\e^{-\iu tH_{\rm Dir}}$.
    \end{proposition}
    \begin{proof}
        As the Dirichlet Laplacian is a self-adjoint operator bounded from below (cf. Proposition~\ref{prop:basic_facts_laplacian}), we can apply Proposition~\ref{prop:friedrichs}, and the Galerkin approximations $U_n(t)=\e^{-\iu H_n t}$ converge strongly to $\tilde{U}(t)=\e^{-\iu \tilde{H}_{\mathrm{fin}} t}$, where $\tilde{H}_{\mathrm{fin}}$ is the Friedrichs extension of $H_{\mathrm{fin}}$, the restriction of $H$ to $\operatorname{Span}\Phi$ (cf. Remark~\ref{rem_nomenclature}).
        But $\tilde{H}_{\mathrm{fin}} = H_{\mathrm{Dir}}$ by Proposition~\ref{prop:friedrichs-dirichlet}, and we obtain the claim.
    \end{proof}
   As anticipated in Section~\ref{sec:results_box_informal}, functions in $\Phi$ can, in principle, also satisfy other boundary conditions---possibly all admissible boundary conditions---which makes the content of Proposition~\ref{prop:legendre-dirichlet-general} not obvious a priori. To prove this point, we shall now specifically consider truncations of the Laplace operator constructed by means of associated Legendre polynomials~\cite{dlmf147-legendre-rodrigues}:
    \begin{definition}\label{def:legendre}
        Let $l,m\in\mathbb{N}$ such that $l\geq m$. The \textit{Legendre polynomial}s $P_l$ and the \textit{associated Legendre polynomials} $P_l^m$ are defined by
        \begin{align}
            P_l(x) & = \frac{1}{2^l l!} \diff*[l]{(x^2-1)^l}{x}\,, 
            \label{eq:def-legendre}\\
            P_l^m(x) & = (-1)^m (1-x^2)^{m/2} \diff*[m]{P_l(x)}{x} \,.
            \label{eq:def-assoc-legendre}
        \end{align}
    \end{definition}
    The Legendre polynomials $P_l$ are the result of the Gram--Schmidt orthogonalization of the polynomials~\cite[p.~208]{zeidler-appliedfunctionalanalysis-1999} with respect to the scalar product of $L^2(-1,1)$, whence
    \begin{equation}
        \label{eq:legendre-poly-gramschmidt}
        \operatorname{Span}\,(P_l)_{l=0}^{n}=\operatorname{Span}\,(x^l)_{l=0}^{n}\qquad\text{for every }n\in\mathbb{N}\, 
    \end{equation}
    For even $m$, each associated Legendre polynomial $P_l^m$ is also a polynomial function of degree $l$. It is known that, for each fixed $m \in \nnum$, $(P_l^m)_{l\geq m}$ is a complete orthogonal set of $L^2((-1,1))$ with normalization factor $\norm*{P_l^m} = \sqrt{2(l+m)!/(2l+1)(l-m)!}$~\cite[p.~226]{teschl-mathematicalmethodsquantum-2009}, whence the set of \textit{normalized} associated Legendre polynomials $(p^m_l)_{l\in\mathbb{N}}$, where
    \begin{equation}
        p^m_l(x)=\frac{1}{\|P^m_l\|}P^m_l(x),
    \end{equation}
    is a complete orthonormal set of $L^2((-1,1))$.
    
    Most importantly for our purposes, for every $m\geq 4$ the set $(p_l^m)_{l \geq m}$ is an example of ``boundary-blind'' basis as defined in Section~\ref{sec:results_box_informal}: it satisfies every admissible boundary condition, in the following sense:
    \begin{proposition}\label{prop:boundary_data_legendre}
        Let $m \in \nnum$, $m \geq 4$. Then $p_l^m(\pm 1) = (p_l^m)'(\pm 1)= 0$.
    \end{proposition}
    \begin{proof}
    Let $m \geq 4$. Using Eq.~\eqref{eq:def-assoc-legendre}, $P_l^m(\pm 1) = 0$.
    Furthermore,
    \begin{multline}
        \diff*{P_l^m(x)}{x}[\bar{x} = \pm 1] =  (-1)^m \frac{m}{2} (1-x^2)^{m/2-1} (-2x) \diff*[m]{P_l(x)}{x} \\
        \left. + (-1)^m (1-x^2)^{m/2}\diff*[m+1]{P_l(x)}{x} \right|_{\bar{x} = \pm 1} = 0\, ,
    \end{multline}
    which proves $P^m_l(\pm1)=(P^m_l)'(\pm1)=0$ and therefore $p^m_l(\pm1)=(p^m_l)'(\pm1)=0$
    \end{proof}
    In Figure~\ref{fig:intro-legendre}, the normalized associated Legendre polynomials and their derivatives are shown for $m = 4$ and $l \in \{4,5,8\}$.
    It is clearly visible that both the $p_l^m(x)$ and their derivatives evaluate to zero at $x = \pm 1$, as shown in the proposition above; thus, the corresponding matrix elements $\braket{p_l^m, -(p_k^m)''}$ with $l,k \geq m$ are the same for all $H_W$. They are explicitly calculated in Appendix~\ref{app:expansion-coefficients}.
   
    \begin{figure}[ht]
        \centering
        \includegraphics[width=0.8\linewidth]{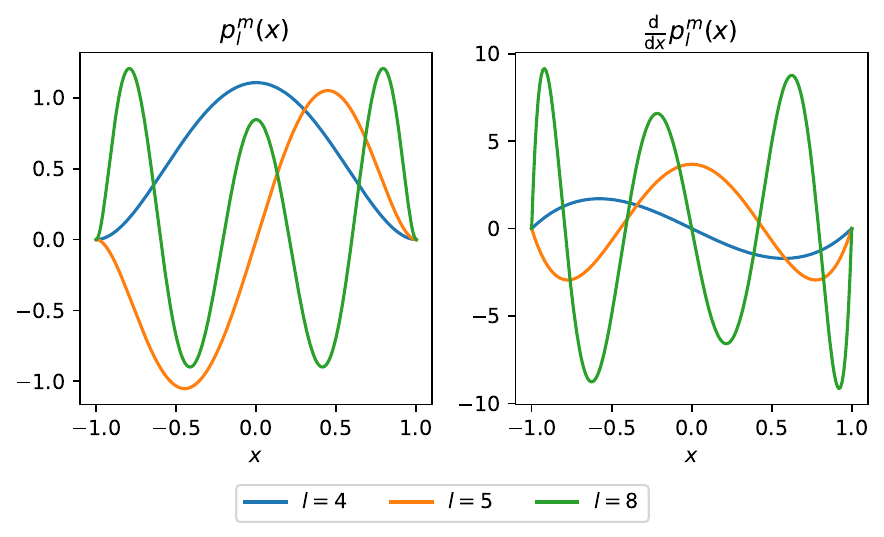}
        \caption{The normalized associated Legendre polynomials $p_l^m$ (left) and their derivatives $\diff*{p_l^m(x)}{x}$ (right) for $l \in \{4,5,8\}$ and $m = 4$.}
        \label{fig:intro-legendre}
    \end{figure}
    
    The following proposition shows that the conditions of Proposition~\ref{prop:legendre-dirichlet-general} are fulfilled by the associated Legendre polynomials:
    \begin{proposition}
        \label{prop:assoc_legendre_diff_1perp}
        Let $P_l^m$ be the associated Legendre polynomials and $m \geq 4$. 
        Then 
        \begin{equation}
            \overline{\operatorname{Span}\left(p^m_l\right)'_{l \geq m}} = \{1\}^\perp \, .
        \end{equation}       
    \end{proposition}
    The proof of this statement is provided in Section~\ref{sec:assoc_legendre_diff_1perp}. 
    Combining this result with Proposition~\ref{prop:legendre-dirichlet-general}, our first concrete convergence result follows:
    \begin{theorem}
        \label{thm:legendre-dirichlet}
        Let $\hilbert = L^2((-1,1))$, $H_n$ be the $n$th Galerkin approximation of the Laplace operator corresponding to any basis of normalized associated Legendre polynomials $(p^m_l)_{l\geq m}$ with $m\geq 4$, and $U_n(t)=\e^{-\iu tH_n}$ the corresponding evolution group. Then
        \begin{equation}
            \lim_{n \rightarrow \infty}\norm{\left(U_n(t)-U_{\mathrm{Dir}}(t)\right)\psi} = 0 \, \quad \forall t \geq 0\, , \, \, \psi \in \hilbert \, ,
        \end{equation}
        where $U_{\mathrm{Dir}}(t)=\e^{-\iu tH_{\rm Dir}}$; that is, $U_n(t)$ strongly converges to the evolution group generated by the Laplace operator with Dirichlet boundary conditions.
    \end{theorem}

    \begin{proof}
    For all $l \geq m$, $p_l^m(\pm 1) = 0$ (see Proposition~\ref{prop:boundary_data_legendre}), and thus $p_l^m \in \mathrm{H}^1_0((-1,1))\cap\mathrm{H}^2((-1,1))$.
    By Proposition~\ref{prop:assoc_legendre_diff_1perp}, 
    \begin{equation}
            \overline{\operatorname{Span}\left(p^m_l\right)'_{l \geq m}} = \{1\}^\perp \, ,
    \end{equation}
    whence the claim follows from Proposition~\ref{prop:legendre-dirichlet-general}.
\end{proof}

\subsection{\texorpdfstring{Main results: convergence to the $\alpha$-periodic Laplacian}{Main results: convergence to the alpha-periodic Laplacian}}
\label{sec:alpha-periodic}

Proposition~\ref{prop:legendre-dirichlet-general} shows that there exist orthonormal bases of $L^2((-1,1))$ such that the corresponding Galerkin approximations of the Laplace operator will always generate a dynamics approximating the one generated by Dirichlet boundary conditions; in particular, Theorem~\ref{thm:legendre-dirichlet} shows that this is the case for Galerkin approximations constructed via associated Legendre polynomials---despite the fact that such functions satisfy all boundary conditions. 

Here we shall focus on cases in which the selected basis $(\phi_l)_{l \in \nnum}$ fulfills not all boundary conditions, but still more than one boundary condition.
Specifically, we will focus on $\alpha$-periodic boundary conditions, cf.~Definition~\ref{def:particular_bcs}. For every $\alpha\in[0,2\pi)$, the Laplace operator with $\alpha$-periodic boundary conditions has domain
\begin{equation}
    \domain(H_{W(\alpha)}) =  \left\{ f\in \mathrm{H}^2((-1,1)), \, f(-1) = \e^{\iu \alpha} f(1),\, f'(-1) = \e^{\iu \alpha} f'(1) \right\}\, ,
\end{equation}
with $W(\alpha)$ as per Eq.~\eqref{eq:w}.

Similarly to Proposition~\ref{prop:legendre-dirichlet-general}, we seek sufficient conditions for a basis $\Phi = (\phi_l)_{l \in \nnum}$ so that the corresponding Galerkin approximations of the Laplace operator reproduce, in the limit $n\to\infty$, the exact dynamics generated by $H_{W(\alpha)}$.

\begin{proposition}
    \label{prop:alpha-periodic-unitary}
    Let $\alpha\in[0,2\pi)$, and $\Phi = (\phi_l)_{l \in \nnum}$ a complete orthonormal set of $L^2((-1,1))$ satisfying the following additional properties:
    \begin{itemize}
        \item[(i)] $\Phi\subset \domain(H_{W(\alpha)})$;
        \item[(ii)] There exists $\phi_0 \in  \domain(H_{W(\alpha)})$ with $\phi_0(1) \neq 0$ such that $\overline{\operatorname{Span}\tilde{\Phi}'}=\{1\}^\perp$,
    \end{itemize}
    where 
    \begin{equation}\label{eq:modified_basis}
        \tilde{\Phi} =(\tilde{\phi}_l)_{l\in\mathbb{N}},\qquad\tilde{\phi}_l(x)=\phi_l(x)-\frac{\phi_l(1)}{\phi_0(1)} \phi_0(x),
    \end{equation}
    and $\tilde{\Phi}'=(\tilde{\phi}'_l)_{l\in\mathbb{N}}$. Then, for all $\psi\in\hilbert$ and $t\in\mathbb{R}$,
    \begin{equation}
        \lim_{n \to \infty} \norm*{\left(U_{n}(t) \psi - U_{W(\alpha)}(t)\right) \psi} = 0 \, .
    \end{equation}
\end{proposition}
This is proven in Section~\ref{sec:proof-alpha-periodic-unitary}.

Like in the Dirichlet case, we will construct explicit examples of bases satisfying the conditions of Proposition~\ref{prop:alpha-periodic-unitary}. The basic idea is the following:
\begin{itemize} \item[(i)] as a starting point, take the normalized associated Legendre polynomials $(p^m_l)_{l\geq m}$, which have vanishing boundary data (Proposition~\ref{prop:boundary_data_legendre}) and therefore $(p^m_l)_{l\geq m}\subset\domain(H_{W(\alpha)})$;
    \item[(ii)] add to this set a new function $\phi_0$, with $\|\phi_0\|=1$ and $\phi_0(1)\neq0$, which specifically satisfies the desired $\alpha$-periodic boundary conditions, and define $\Phi$ as the set obtained by applying the Gram--Schmidt algorithm (using the scalar product of $L^2(-1,1)$) to $\{\phi_0\}\cup(p^m_l)_{l\geq m}$ starting from $\phi_0$, which we will compactly denote by
    \begin{equation}
        \label{eq:def-gs}
        \Phi=\operatorname{GS}\left(\phi_0,(p^m_l)_{l\geq m}\right).
    \end{equation}
    Each vector in $\Phi$ is a linear combination of $\phi_0$ and finitely many associated Legendre polynomials, whence $\Phi\subset\domain(H_{W(\alpha)})$ as well.
\end{itemize}
The following theorem shows that a basis $\Phi$ constructed this way also fulfills condition \textit{(ii)} of Proposition~\ref{prop:alpha-periodic-unitary}, and thus the corresponding Galerkin approximations converge to the evolution generated by the $\alpha$-periodic Hamiltonian: 
\begin{theorem}
    \label{thm:alpha-periodic-legendre}
    Let $f_0 \in \mathrm{H}^2((-1,1))$ be a normalized and $\alpha$-periodic function with $f_0(1) \neq 0$. Let $\Phi = (\phi_l)_{l \in \nnum} = \mathrm{GS}(f_0,(p_l^m)_{l \geq m})$ for $m \geq 4$.
    Then the approximate time evolution in the basis $\Phi$, $U_{n,\Phi}(t)$ converges strongly to the $\alpha$-periodic time evolution $U_{W(\alpha)}(t) = \e^{-\iu H_{W(\alpha)} t}$:
    \begin{equation}
        \lim_{n \to \infty} \norm{U_{W(\alpha)}(t) \psi - U_{n,\Phi}(t) \psi} = 0 \quad \forall \psi \in \hilbert\, .
    \end{equation}
\end{theorem}
The theorem is proven in Section~\ref{sec:proof-thm-alpha-periodic-legendre}.
In the following we investigate two examples.
\begin{example}[Periodic boundary conditions]
    We add the \textit{periodic} function $f_0(x) = \frac{1}{\sqrt{2}}$ to the associated Legendre polynomials.
    $f_0(1) \neq 0$, and therefore using Theorem~\ref{thm:alpha-periodic-legendre}
    the Galerkin approximation in the basis $\Phi=\operatorname{GS}\left(f_0,(p^m_l)_{l\geq m}\right)$, $U_{n,\Phi}(t)$, converges strongly to the time evolution of the periodic Laplacian, $U_{W(0)}(t)\equiv U_{\mathrm{per}}(t)$.
\end{example}
\begin{example}[Anti-periodic boundary conditions]
    Adding the \textit{anti-periodic} function $f_0(x) = \cos\left(\frac{\pi}{2}(x+1)\right)$ to the basis functions, the Galerkin approximation $U_{n,\Phi}(t)$ in the basis $\Phi=\operatorname{GS}\left(f_0,(p^m_l)_{l\geq m}\right)$, converges strongly to the time evolution of the anti-periodic Laplacian, $U_{W(\pi)}(t)$.
\end{example}
Notably, in both examples above, the function $f_0$ used to construct the basis $\Phi$ also satisfies Neumann boundary conditions; therefore, in both cases the Galerkin approximations of the Laplacian satisfy two distinct boundary conditions---periodic and Neumann in the first case, anti-periodic and Neumann in the second case---and yet in the limit $n\to\infty$ Neumann boundary conditions are ``discarded'' in favor of the other choice. We comment on this point in the following.

\subsection{Summary}

Let us summarize the findings of this section (see Table~\ref{tab}). When numerically simulating the particle in a box in the eigenbasis of a specific self-adjoint extension $H_W$, the dynamics of the corresponding Galerkin approximations will always converge, in the strong sense, to the exact dynamics generated by $H_W$ (see Remark~\ref{rem:eigenbasis-convergence}). This is instead not guaranteed when choosing a complete orthonormal set which does not span a core for any specific boundary conditions---as pointed out in Section~\ref{sec:galerkin-general} (see Proposition~\ref{prop:friedrichs}), in general we will obtain convergence to the dynamics generated by the Friedrichs extension of the Laplace operator on the given basis, which will generally differ from the one corresponding to the desired boundary conditions. Specifically:
\begin{itemize}
    \item Galerkin approximations corresponding to a basis $\Phi=(\phi_l)_{l\in\mathbb{N}}$ satisfying Dirichlet boundary conditions and the additional condition $\overline{\operatorname{Span}\Phi'}=\{1\}^\perp$ will always yield convergence to the dynamics corresponding to Dirichlet boundary conditions (Proposition~\ref{prop:friedrichs-dirichlet}), even if they satisfy other boundary conditions---possibly all of them. This is the case, in particular, for any basis of associated Legendre polynomials with $m\geq4$ (Theorem~\ref{thm:legendre-dirichlet}).
     \item Similarly, Galerkin approximations corresponding to a basis $\Phi=(\phi_l)_{l\in\mathbb{N}}$ satisfying $\alpha$-periodic (for example, periodic or anti-periodic) boundary conditions, plus an additional condition, will yield convergence to $\alpha$-periodic dynamics (Proposition~\ref{prop:alpha-periodic}). Such bases can be obtained again by taking associated Legendre polynomials with $m\geq4$, adding a function $f_0$ satisfying $\alpha$-periodic boundary conditions, and applying Gram--Schmidt to the resulting system (Theorem~\ref{thm:alpha-periodic-legendre}). In both cases, the basis $\Phi$ also satisfies Neumann boundary conditions.
\end{itemize}

\begin{center}
\begin{table}
    \begin{tabular}{l l l } \hline
        Basis $\Phi$ & b.c. fulfilled by $\Phi$ & b.c. of limiting dynamics\\ \hline & \\[-1.5ex]
        $\left(\sin\left(\frac{j \pi}{2}(x+1)\right)\right)_{j \geq 1}$ & Dirichlet & Dirichlet \\[0.2cm]
        $\left(\cos\left(\frac{j \pi}{2}(x+1)\right)\right)_{j \geq 0}$ & Neumann & Neumann \\[0.2cm]
        $\left(\exp\left(j \pi \iu(x+1)\right)\right)_{j \in \znum}$ & periodic & periodic \\[0.2cm]
        $(p_l^m)_{l \geq m}$ & all & Dirichlet  \\[0.2cm]
        $\mathrm{GS}\left(\frac{1}{\sqrt{2}},(p_l^m)_{l \geq m}\right)$& Neumann and periodic & periodic \\[0.2cm]
        $\mathrm{GS}\left(\cos\left(\frac{\pi}{2}(x+1)\right),(p_l^m)_{l \geq m}\right)$ & Neumann and anti-periodic & anti-periodic \\[0.2cm] \hline
    \end{tabular}
    \caption{Summary of the results of the section. The functions in the first three lines are the eigenvectors of the Laplace operator with Dirichlet, Neumann, and periodic boundary conditions (cf.~Example~\ref{ex:bc-eigenfunctions}).}\label{tab}
    \end{table}
\end{center}
We conclude with the following observation. While the domains corresponding to distinct boundary conditions are not included into each other, the \textit{form} domains corresponding to Dirichlet, $\alpha$-periodic, and Neumann boundary conditions satisfy
\begin{equation}
    \mathrm{H}^1_0((-1,1)) = \domain(q_{\mathrm{Dir}}) \subset \domain(q_{W(\alpha)}) \subset \domain(q_{\mathrm{Neu}}) = \mathrm{H}^1((-1,1)) \, ,
\end{equation}
where $q_{\mathrm{Dir}}$, $q_{W(\alpha)}$ and $q_{\mathrm{Neu}}$ are the forms associated with the respective Hamiltonians. Thus, in all cases discussed in this work, Galerkin approximations favor the boundary conditions corresponding with the smaller form domain. Interestingly, the prevalence of Dirichlet boundary conditions also appears in other problems, see~\cite[Prop 4.14.1, p.759]{zagrebnov-trotterkatoproductformulae-2024} and~\cite{deoliveira-mathematicalpredominancedirichlet-2012,asorey-dynamicalcompositionlaw-2013}.

\section{Proofs of Section~\ref{sec:galerkin-general}}\label{sec:proofs_sec2}

In this section we shall provide proofs for the general convergence properties stated in Section~\ref{sec:galerkin-general}, namely Propositions~\ref{prop:galerkin_criterion} and~\ref{prop:friedrichs}.

\subsection{Proof of Proposition~\ref{prop:galerkin_criterion}}
\label{sec:galerkin-criterion-proof}

To prove Proposition~\ref{prop:galerkin_criterion}, we will use an adaptation of the methods in Ref.~\cite{parter-rolesstabilityconvergence-1980} to the unitary case. To this purpose, we shall need some preliminary lemmata.
We begin by stating some basic properties of the Galerkin projector $Q_n$, cf.~Definition~\ref{def:galerkin_projector}:
\begin{lemma}
        \label{lem:rn-prop}
        Let $H:\domain(H)\subset\hilbert\rightarrow\hilbert$ be a coercive self-adjoint operator, $(P_n)_{n\in\mathbb{N}}$ a family of finite-dimensional orthogonal projectors satisfying $P_n\hilbert\subset\domain(H)$, and $H_n,R_n,Q_n$ be as per Definitions~\ref{def:galerkin_approximation} and~\ref{def:galerkin_projector}. Then the following identities hold:
        \begin{equation}
            H_n R_n = P_n \, , \qquad P_nQ_n=Q_n \, , \qquad H_n Q_n = P_n H \, , \qquad U_n(t) R_n = R_n U_n(t).
        \end{equation}
        Furthermore, $Q_n$ is a (not necessarily self-adjoint) projector: $Q_n^2 = Q_n$.
    \end{lemma}
    \begin{proof}
        We show these identities using the direct sum decomposition of all operators with respect to the decomposition $\hilbert=\hilbert_n\oplus\hilbert_n^\perp$, with $\hilbert_n=P_n\hilbert$, cf.~Definitions~\ref{def:galerkin_approximation} and~\ref{def:galerkin_projector}. The first one follows from
        \begin{equation}
            H_n R_n  = \left(   \hat{H}_n \oplus 0_{\hilbert_n^\perp} \right) \left(  \hat{H}_n^{-1} \oplus 0_{\hilbert_n^\perp} \right) = \left( \hat{H}_n  \hat{H}_n^{-1} \right) \oplus 0_{\hilbert_n^\perp} = \id_{\hilbert_n} \oplus 0_{\hilbert_n^\perp} = P_n \, .
        \end{equation}
    Besides,
    \begin{equation} \label{eq:pnrn_rn}
    P_nR_n=\left(\id\oplus0_{\hilbert_n^\perp}\right)\left(\hat{H}_n^{-1}\oplus0_{\hilbert_n^\perp}\right)=\left(\hat{H}_n^{-1}\oplus0_{\hilbert_n^\perp}\right)=R_n,
    \end{equation}
    thus clearly implying $P_nQ_n=P_nR_nH=R_nH=Q_n$.
    From the first identity we directly get
    \begin{equation}
        H_n Q_n \psi = H_n R_n H \psi = P_n H \psi,
    \end{equation}
    and, as $\domain(Q_n) = \domain(H)$, $H_n Q_n = P_n H$ holds.
        For the third identity, recall (cf.~Eq.~\eqref{eq:decomposition-un}) that the evolution group $U_n(t)$ also admits a decomposition $U_n(t)=\hat{U}_n(t)\oplus\id_{\hilbert_n^\perp}$. Then,
        \begin{multline}
            U_n(t) R_n = \left(   \hat{U}_n(t) \oplus \id_{\hilbert_n^\perp} \right) \left(  \hat{H}_n^{-1} \oplus 0_{\hilbert_n^\perp} \right)  = \left( \hat{U}_n(t)  \hat{H}_n^{-1} \right) \oplus 0_{\hilbert_n^\perp} \\
            = \left(   \hat{H}_n^{-1} \hat{U}_n(t) \right) \oplus 0_{\hilbert_n^\perp} =  \left(  \hat{H}_n^{-1} \oplus 0_{\hilbert_n^\perp} \right)\left(   \hat{U}_n(t) \oplus \id_{\hilbert_n^\perp} \right) = R_n U_n(t) \, .
        \end{multline}
        Finally, to show $Q_n^2 = Q_n$, we begin by noting that  $Q_n^2$ is well-defined, as $P_n \hilbert \subset \domain(H)$.
        Using Eq.~\eqref{eq:pnrn_rn} and the definition of $R_n$, we get the following identities:
        \begin{align}
            (\id - P_n) R_n & = R_n - R_n = 0 \, , \\
            \label{proofeq:rnpn_rn}
            R_n P_n & = \left(   \hat{H}_n^{-1} \oplus 0_{\hilbert_n^\perp} \right) P_n = R_n\, \\
            R_n (\id -P_n) & = \left(   \hat{H}_n^{-1} \oplus 0_{\hilbert_n^\perp} \right) (\id - P_n) = 0 \, ,
        \end{align}
        as $P_n$ acts on $\hilbert_n$ as $\id$ and as $0$ elsewhere.
        Then, inserting two identities into $Q_n^2$ and using the first identity of the lemma and Eq.~\eqref{proofeq:rnpn_rn}, we obtain
        \begin{multline}
            Q_n^2 = R_n H R_n H = R_n (P_n +\id - P_n)H (P_n + \id - P_n) R_n H = R_n P_n H P_n R_n H \\
            = R_n H_n R_n H = R_n P_n H = R_n H = Q_n \, ,
        \end{multline}
        which completes the proof.
    \end{proof}

\begin{remark}\label{rem:galerkin_projection}
We point out that, in~\cite{parter-rolesstabilityconvergence-1980}, the operator $R_n$ is denoted by $R_n=H_n^{-1}P_n$. This is an (ultimately harmless) abuse of notation, since, as previously remarked, the operator $H_n$ is not invertible---only its restriction $\hat{H}_n$ to the finite-dimensional space $\hilbert_n$ is. Such an abuse of notation essentially consists in identifying $H_n$ and $\hat{H}_n$. Here we decided to avoid this abuse and keep the distinction between $H_n$ and $\hat{H}_n$ explicit.
\end{remark}
        
   The following property is a particular case of~\cite[Lemma 4.1]{parter-rolesstabilityconvergence-1980}: here we adapt it to the unitary scenario, and we provide a proof for completeness.
    \begin{lemma}
        \label{lem:trotter-approximation-theorem}
        Let $H$ be a coercive self-adjoint operator, and $H_n$ and $R_n$ as defined above. The following equality holds for all $\psi \in \hilbert, \, t \in\mathbb{R}$:
        \begin{equation}
            R_n \left( U(t) - U_n(t) \right) H^{-1} \psi = -\iu \int_0^t U_n(t-s)\left(R_n - P_n H^{-1}\right)U(s) \psi \dl s\, .
        \end{equation}
    \end{lemma}
    \begin{proof}
        We prove the statement for $t\geq 0$; for $t\leq0$, the proof is analogous.
        To begin with, we note that
        \begin{align}
            U(t) H \phi & = H U(t) \phi \,,  \quad\quad\!\diff*{U(t)}{t}\phi  = -\iu H U(t) \phi \qquad\quad\;\forall \phi \in \domain(H), \, \\
            U_n(t)H_n  \psi & = H_n U_n(t) \psi\, , \quad \diff*{U_n(t)}{t}\psi = -\iu H_n U_n(t) \psi \quad  \quad \!\forall \psi \in \hilbert\,.
        \end{align}        
        Given $t > 0$, $\psi \in \hilbert$, define the $\hilbert$-valued function $(0,t)\ni s\mapsto G_n(s)\in\hilbert$ by 
        \begin{equation}
            G_n(s) \coloneqq U_n(t-s) R_n U(s) H^{-1} \psi \, .
        \end{equation}
        Since $H^{-1}\psi\in\domain(H)$, the function $s\mapsto U(s)H^{-1}\psi$ is differentiable; as such,       
        $G_n(s)$ is a differentiable function through $\hilbert$ for $0 < s < t$. Applying the chain rule and using Lemma~\ref{lem:rn-prop}, its derivative reads
        \begin{equation}\label{eq:int1}
            \begin{split}
                \diff*{G_n(s)}{s} & =  U_n(t-s)R_n (-\iu H)U(s)H^{-1}\psi -(-\iu H_n) U_n(t-s) R_n U(s) H^{-1} \psi \\
                & = - \iu U_n(t-s) R_n U(s) \psi+\iu U_n(t-s) P_n H^{-1} U(s) \psi  \\
                & = - \iu U_n(t-s) \left(R_n - P_n H^{-1} \right) U(s) \psi \,.
            \end{split}
        \end{equation}
On the other hand, using $R_n U_n(t) = U_n(t) R_n$ (again by Lemma~\ref{lem:rn-prop}) and the definition of $G_n(t)$,
        \begin{equation}\label{eq:int2}
            \begin{split}
                \int_0^t \diff*{G_n(s)}{s} \dl s &  = G_n(t)-G_n(0)\\&=U_n(t-t) R_n U(t) H^{-1} \psi - U_n(t-0) R_n U(0) H^{-1} \psi \\
                & = R_n \left( U(t) - U_n(t) \right) H^{-1} \psi,
            \end{split}  
        \end{equation}
whence the claimed equality is obtained by integrating Eq.~\eqref{eq:int1} and computing the left-hand side by means of Eq.~\eqref{eq:int2}.
    \end{proof}
    We apply the previous lemma to the Galerkin projection $Q_n$.
    \begin{lemma}
        \label{lem:trotter-approximation-theorem-qn}
        For all $\psi \in \domain(H^2)$ and $t\in\mathbb{R}$,
        \begin{equation}
            Q_n U(t) \psi - U_n(t) Q_n \psi = -\iu \int_0^t U_n(t-s) P_n \left(Q_n -  \id \right) U(s) H \psi \dl s \, .
        \end{equation}
    \end{lemma}
    \begin{proof}
        Let $\psi \in \domain(H^2)$.
        Then, using the identities $R_n U_n(t) = U_n(t) R_n$ and $P_n Q_n = Q_n$ (Lemma~\ref{lem:rn-prop}), and applying Lemma~\ref{lem:trotter-approximation-theorem} to the vector $H^2\psi\in\hilbert$, we have
        \begin{equation}
            \begin{split}     
                Q_n U(t) \psi - U_n(t) Q_n \psi & = R_n H U(t) H^{-1} H \psi - U_n(t) R_n H H^{-1} H \psi \\
                & = R_n \left( U(t) - U_n(t) \right) H^{-1} H^2 \psi \\
                & = -\iu \int_0^t U_n(t-s)\left(R_n -P_n H^{-1} \right) U(s) H^2 \psi \dl s \\
                & = -\iu \int_0^t U_n(t-s) \left( R_n - P_nH^{-1} \right) HU(s) H \psi \dl s \\
                & = -\iu \int_0^t U_n(t-s) \left( R_n H - P_n \right) U(s) H \psi \dl s \\
                & = -\iu \int_0^t U_n(t-s) P_n \left( Q_n - \id \right) U(s) H \psi \dl s \, ,
            \end{split}
        \end{equation}
        which is the claimed identity.
    \end{proof}

\begin{proof}[Proof of Proposition~\ref{prop:galerkin_criterion}]
    Let $H$ be a coercive self-adjoint operator, i.e., there exists $\gamma >0$ such that $\braket{\psi,H \psi} \geq \gamma \norm{\psi}^2$ for all $\psi \in \domain(H)$.
    Thus, by Proposition~\ref{prop:coercive-implies-invertibility}, $H$ is invertible and $\|H^{-1}\|\leq\frac{1}{\gamma}$; in particular, $Q_n$ is well-defined.     
    We now show that $R_n$ is also bounded by $\frac{1}{\gamma}$.
    To begin with,
    \begin{equation}
        \braket{\psi_n,H_n \psi_n} = \braket{\psi_n,H\psi_n} \geq \gamma \norm{\psi_n}^2 \quad \text{for all }\psi_n \in \hilbert_n \, .
    \end{equation}
    Let $\hat{H}_n$ and $\hat{\psi}_n$ be the restrictions of $H_n$ and $\psi_n$ on $\hilbert_n$.
    Then,
    \begin{equation}
        \braket{\hat{\psi}_n, \hat{H}_n \hat{\psi}_n} \geq \gamma \norm{\hat{\psi_n}}\,, 
    \end{equation}
    and, by the same argument as in the proof of Proposition~\ref{prop:coercive-implies-invertibility}, $\norm{\hat{H}_n^{-1}} \leq \frac{1}{\gamma}$.
    As $R_n = \hat{H}_n^{-1} \oplus 0_{\hilbert_n^\perp}$, we have $\norm{R_n} \leq \frac{1}{\gamma}$ for all $n \in \nnum$.
    We now return to proving 
    \begin{equation}
        \lim_{n \to \infty} \norm{U(t) \psi - U_n(t) \psi} = 0 \quad \forall \psi \in \hilbert \, ,\;t\in\mathbb{R}.
    \end{equation}
    As $U(t)$ and the $U_n(t)$ are unitary, and thus bounded operators, it is sufficient to prove convergence on a dense set.
    We can choose the dense subspace $\domain(H^2)$, cf.~\cite[p.~180]{reed-mmmp2-fourier-1975}.
    Then, given $\psi \in \domain(H^2)$, we have
    \begin{equation}\label{eq:bound}
        \begin{split}
            \norm{U(t)\psi -U_n(t)\psi }  \leq &\;\norm{U(t)\psi - Q_n U(t) \psi} + \norm{Q_n U(t) \psi - U_n(t)Q_n \psi} \\
            & + \norm{U_n(t) Q_n \psi -U_n(t) \psi} \\
             = &\;\norm{(Q_n - \id) U(t) \psi} + \norm{Q_n U(t) \psi - U_n(t)Q_n \psi} \\
            & + \norm{ (Q_n -\id)\psi} \, .\
        \end{split}
    \end{equation}
    Per assumption, $\lim_{n \to \infty} \norm{Q_n \phi - \phi} = 0$ for all $\phi \in \domain(H)$. Since both $\psi$ and $U(t)\psi$ are in $\domain(H^2)\subset\domain(H)$, the first and third term above converge to zero. We must still prove
    \begin{equation}
        \lim_{n \to \infty} \norm{Q_n U(t) \psi - U_n(t) Q_n \psi} = 0 \, .
    \end{equation}
    For this purpose, consider the following function:
    \begin{equation}\label{eq:fn}
        f_n(s) = \norm{(Q_n -\id) U(s) H \psi} \,, \quad s \in [0,t] \, ,
    \end{equation}
    which is well-defined since $\psi\in\domain(H^2)$ implies $H\psi\in\domain(H)$ and again $U(t)H\psi\in\domain(H)$.
    Then,
    \begin{equation}
        \lim_{n \to \infty} f_n(s) = \lim_{n \to \infty} \norm{(Q_n -\id) U(s) H \psi} = 0 \quad \forall s \in [0,t].
    \end{equation}
     Furthermore, for all $s \in [0,t]$,
    \begin{align}
        \label{eq:bound-integral}
        \norm{(Q_n -\id) U(s) H \psi} &= \norm{(R_n -H^{-1}) U(s) H^2 \psi} \\ \nonumber&\leq \norm{(R_n -H^{-1})}\norm{U(s)H^2 \psi} \leq \frac{2}{\gamma} \norm{H^2 \psi}\, ,
    \end{align}
    as $\norm{R_n} \leq \frac{1}{\gamma}$ and $\norm{H^{-1}} \leq \frac{1}{\gamma}$ as well.
    Thus, the functions $f_n(s)$ converge pointwise to $0$ as $n\to\infty$ for every $s\in[0,t]$, and are uniformly bounded by $\frac{2}{\gamma} \norm{H^2 \psi}$.
    Using Lemma~\ref{lem:trotter-approximation-theorem-qn},
    \begin{align}
     \norm{Q_n U(t) \psi - U_n(t) Q_n \psi} &\leq \int_0^t \norm{U_n(t-s) P_n \left(Q_n -  \id \right) U(s) H \psi} \dl s \nonumber\\& \leq \int_0^t \norm{ \left(Q_n -  \id \right) U(s) H \psi} \dl s\, \nonumber\\&= \int_0^t f_n(s) \dl s\, .
    \end{align}
    As $[0,t]$ is compact and the $f_n(s)$ are uniformly bounded, by Lebesgue's dominated convergence theorem~\cite[p.~24]{reed-mmmp1-funkana-1980} we get
    \begin{equation}
        \lim_{n \to \infty} \norm{Q_n U(t) \psi - U_n(t) Q_n \psi} = \lim_{n \to \infty} \int_0^t f_n(s) \dl s = 0\, ,
    \end{equation}
    which proves the claim.
\end{proof}

\begin{remark}
This result could be alternatively proven by using the Kato--Trotter approximation theorem~\cite[p.~209]{nagel-oneparametersemigroups-1999}. The proof employed in the paper has the advantage of being directly suitable to calculating state-dependent convergence rates for Galerkin approximations, which are especially useful for applications. For example, if $\psi$ is an eigenvector of $H$ with eigenvalue $E$, the functions $f_n(s)$ as defined in Eq.~\eqref{eq:fn} satisfy $f_n(s) \leq E \norm{Q_n \psi-\psi}$, whence
    \begin{equation}
        \norm{Q_n U(t) \psi - U_n(t) Q_n \psi} \leq E t \norm{(Q_n-\id) \psi} \, ;
    \end{equation}
    in such a case, the convergence rate of each of the three terms in Eq.~\eqref{eq:bound} essentially depends on the (state-dependent) convergence rate of $Q_n$ to the identity.
    
    In general, the tighter the bound on the functions $f_n(s)$, the tighter the one on the error. Applying such arguments to concrete models will be the object of future research.
\end{remark}

\subsection{Proof of Proposition~\ref{prop:friedrichs}}
\label{sec:proof-friedrichs}

We begin by noting that, while in the definition of Galerkin projectors (Definition~\ref{def:galerkin_approximation}) we required our operator $H$ to be coercive ($\braket{\psi,H\psi}\geq\gamma\|\psi\|^2$ with $\gamma>0$) and thus to admit a bounded inverse, Proposition~\ref{prop:friedrichs} only requires boundedness from below---i.e., with $\gamma$ possibly being zero or negative. As anticipated in Remark~\ref{rem:waiving_coerciveness}, this is possible since any symmetric operator bounded from below only differs from a coercive one by an immaterial shift. For this purpose, we need to check that such a shift does not affect the Friedrichs extension---that is, the Friedrichs extensions of two operators only differing by a constant differ themselves by the same constant:
\begin{lemma}
\label{lem:friedrich_plus_b}
    Let $A:\domain(A)\subset\hilbert\rightarrow\hilbert$ be a densely defined symmetric operator bounded from below, $b \in \rnum$, and $B=A+b$.
    Then the Friedrichs extensions $\tilde{A},\tilde{B}$ of $A$ and $B$ satisfy $\tilde{B}=\tilde{A}+b$.
\end{lemma}
\begin{proof}
    By assumption, there is $\gamma\in\mathbb{R}$ such that $\braket{\psi,A\psi} \geq \gamma \norm{\psi}^2$ for all $\psi \in \domain(A)$.
    Consequently,
    \begin{equation}
        \braket{\psi,B\psi} = \braket{\psi,A\psi+b\psi} \geq (\gamma+b) \norm{\psi}^2 \quad \forall \psi \in \domain(B) = \domain(A) \, .
    \end{equation}
    Let $\norm{\psi}_{+,A}$ and $\norm{\psi}_{+,B}$ be the norms associated to $A$ and $B$ as defined in Eq.~\eqref{eq:plusnorm}: for all $\psi\in\domain(A)=\domain(B)$,
    \begin{align}
        \norm{\psi}_{+,A}^2 & = q_{A-\gamma}(\psi,\psi) + \norm{\psi}^2 = \braket{\psi,(A-\gamma+1)\psi},\\
        \norm{\psi}_{+,B}^2 & = q_{B-\gamma-b}(\psi,\psi) + \norm{\psi}^2 = \braket{\psi,(B-\gamma-b+1)\psi},
    \end{align}
    Thus, $\norm{\psi}_{+,A} = \norm{\psi}_{+,B}$ for all $\psi \in \domain(A)=\domain(B)$, and so clearly are the form domains of their Friedrichs extensions:
    \begin{equation}
        \domain(\tilde{q}_{A-\gamma}) = \overline{\domain(A)}^{\norm{}_{+,A}} = \overline{\domain(B)}^{\norm{}_{+,B}} = \domain(\tilde{q}_{B-\gamma-b}) \, .
    \end{equation}
    Let $\tilde{A}$ and $\tilde{B}$ be the Friedrichs extensions of $A$ and $B$; we claim $\tilde{B}=\tilde{A}+b$. To this extent, it suffices to note that $\tilde{B}-b$ is a self-adjoint extension of $A$, since
    \begin{equation}
        (\tilde{B}-b) \psi = B\psi - b \psi = A\psi \quad \text{for all } \psi \in \domain(A) = \domain(B) \, ,
    \end{equation}
and its domain is simply 
\begin{equation}
        \domain(\tilde{B}-b) = \domain(\tilde{B}) \subset \domain(\tilde{q}_{B-\gamma-b}) =  \domain(\tilde{q}_{A-\gamma}) \, .
    \end{equation}
    But $\tilde{A}$ is the only self-adjoint extension of $A$ whose domain is a subset of $\domain(\tilde{q}_{A-\gamma})$ (cf.~Definition~\ref{def:friedrichs_extension}).
    Thus $\tilde{B}-b= \tilde{A}$.
\end{proof}

\begin{proposition}
    \label{prop:friedrich-galerkin-projection}
    Let $H:\domain(H)\subset\hilbert\rightarrow\hilbert$ be a coercive self-adjoint operator, $(P_n)_{n \in \nnum}$ a family of orthogonal projectors satisfying $P_n \hilbert \subset \domain(H)$, $P_n P_m = P_n$ for all $n \leq m$, and $P_n \psi \to \psi$ for all $\psi \in \hilbert$. Define
\begin{equation}
    \hilbert_{\rm fin}:=\bigcup_{n\in\mathbb{N}}P_n\hilbert,
\end{equation}
and let $H_{\rm fin}$ be the restriction of $H$ to $\hilbert_{\rm fin}$.
    Then $H_{\rm fin}$ is densely defined, symmetric and positive. Let $\tilde{H}_{\rm fin}$ be the Friedrichs extension  of $H_{\rm fin}$, cf.~Definition~\ref{def:friedrichs_extension}, and $\tilde{Q}_n$ be the Galerkin projection associated with $\tilde{H}_{\rm fin}$, cf.~Definition~\ref{def:galerkin_projector}. Then
    \begin{equation}
        \lim_{n \to \infty} \tilde{Q}_n \psi = \psi \quad \text{for all } \psi \in \domain(\tilde{H}_{\rm fin}) \, .
    \end{equation}
\end{proposition}
\begin{proof}
To begin with, let us check that $H_{\mathrm{fin}}$ is densely defined, symmetric and positive.
The fact that $\hilbert_{\rm fin}$ is dense follows directly from the property $P_n\psi\to\psi$: given $\psi\in\hilbert$, the sequence $(P_n\psi)_{n\in\mathbb{N}}\subset\hilbert_{\rm fin}$ converges to $\psi$.
As $H$ is symmetric, its restriction $H_{\mathrm{fin}}$ is also a symmetric operator.
Furthermore, per assumption there exists $\gamma>0$ such that $\braket{\psi,H\psi} \geq \gamma\norm{\psi}^2$ for all $\psi \in \domain(H)$, and therefore its restriction $H_{\rm fin}$ also satisfies
\begin{equation}
    q(\varphi,\varphi):=\braket{\varphi,H_{\rm fin}\varphi}\geq\gamma\|\varphi\|^2
\end{equation}
for all $\varphi\in\hilbert_{\rm fin}$. As such, $H_{\rm fin}$ has a Friedrichs extension $\tilde{H}_{\rm fin}:\domain(\tilde{H}_{\rm fin})\subset\hilbert\rightarrow\hilbert$ whose associated quadratic form $\tilde{q}$ also satisfies
\begin{equation}
\label{proofeq:tildeq-positive}
    \tilde{q} (\varphi,\varphi) = \braket{\varphi,\tilde{H}_{\rm fin} \varphi} \geq \gamma\norm{\varphi}^2 \quad \forall \varphi \in \domain(\tilde{H}_{\rm fin}) \, .
\end{equation}
Besides, as $\tilde{q}$ is a symmetric positive form on the linear space $\domain(\tilde{H}_{\rm fin})$, it satisfies the Cauchy--Schwarz inequality:
\begin{equation}
    \label{proofeq:cauchy-schwary-q}
    |\tilde{q}(\psi,\varphi)|^2 \leq \tilde{q}(\varphi,\varphi)\tilde{q}(\psi,\psi) \quad \forall \varphi,\psi \in \domain(\tilde{H}_{\rm fin}) \, .
\end{equation}
Also notice that, by the very definition of $\hilbert_{\rm fin}$, $P_n\hilbert\subset\hilbert_{\rm fin}=\mathcal{D}(H_{\rm fin})\subset\mathcal{D}(\tilde{H}_{\rm fin})$. 

The two operators $H$ and $\tilde{H}_{\rm fin}$ are, in general, distinct self-adjoint extensions of $H_{\rm fin}$. However, since $P_n\hilbert\subset\hilbert_{\rm fin}$, the corresponding Galerkin approximations, $P_nHP_n$ and $P_n\tilde{H}_{\rm fin}P_n$ actually coincide:
\begin{equation}
\label{proofeq:galerkin-approx-equal}
   H_n =  P_nHP_n=P_n\tilde{H}_{\rm fin}P_n =: \tilde{H}_{\mathrm{fin},n} \qquad\text{for all }n\in\mathbb{N},
\end{equation}
while the corresponding Galerkin projectors, cf.~Definition~\ref{def:galerkin_projector}, generally do not. In the following, let $\tilde{Q}_n$ be the $n$th Galerkin projector associated with $\tilde{H}_{\rm fin}$.  We claim the following:
\begin{equation}
    \lim_{n\to\infty}\|\tilde{Q}_n\psi-\psi\|=0\qquad\text{for all }\psi\in\domain(\tilde{H}_{\rm fin}).
\end{equation}
Let $\psi\in\domain(\tilde{H}_{\rm fin})$, $n\in\mathbb{N}$, $\psi_n = \tilde{Q}_n \psi$ and $\hilbert_n = P_n \hilbert$.
Then, using Eq.~\eqref{proofeq:galerkin-approx-equal} and Lemma~\ref{lem:rn-prop},
\begin{equation}
    H_n \psi_n = \tilde{H}_{\mathrm{fin},n} \tilde{Q}_n \psi = P_n \tilde{H}_{\mathrm{fin}} \psi\, , 
\end{equation}
and therefore
\begin{equation}
    \Braket{H_n \psi_n,\phi_n} = \Braket{\tilde{H}_{\rm fin} \psi,\phi_n} \quad \text{for all } \phi_n \in \hilbert_n \, .
\end{equation}
Thus, 
\begin{equation}
    \Braket{\phi_n,\tilde{H}_{\rm fin}(\psi - \psi_n)} = 0\qquad\text{for all }\phi_n\in\hilbert_n,
\end{equation}
and in particular
\begin{equation}
    \Braket{\psi_n,\tilde{H}_{\rm fin}(\psi - \psi_n)} = 0.
\end{equation}
As such, for all $\phi_n\in\hilbert_n$ we have
\begin{align}\label{eq:switch}
   \tilde{q}(\psi-\psi_n,\psi-\psi_n)&= \Braket{\psi-\psi_n,\tilde{H}_{\rm fin}(\psi - \psi_n) }\nonumber\\ &= \Braket{\psi-\phi_n,\tilde{H}_{\rm fin}(\psi - \psi_n) } \nonumber\\&= \tilde{q}(\psi-\phi_n,\psi-\psi_n),
\end{align}
whence, using Eq.~\eqref{eq:switch} and the Cauchy--Schwarz inequality Eq.~\eqref{proofeq:cauchy-schwary-q}, for all $\phi_n\in\hilbert_n$ we have
\begin{align}
    |\tilde{q}(\psi-\psi_n,\psi-\psi_n)|&=|\tilde{q}(\psi-\phi_n,\psi-\psi_n)|\nonumber\\
    &\leq  \sqrt{\tilde{q}(\psi-\phi_n,\psi-\phi_n)} \sqrt{\tilde{q}(\psi-\psi_n,\psi-\psi_n)},
\end{align}
and therefore
\begin{align}\label{eq:minimum}
    |\tilde{q}(\psi-\psi_n,\psi-\psi_n)|\leq|\tilde{q}(\psi-\phi_n,\psi-\phi_n)|\qquad\text{for all }\phi_n\in\hilbert_n;
\end{align}
that is, the absolute value of $\tilde{q}(\psi-\phi_n,\psi-\phi_n)$ for $\phi_n$ ranging in $\hilbert_n$ reaches its minimum at $\psi_n$. Finally, using the coercivity of $\tilde{q}$ (Eq.~\eqref{proofeq:tildeq-positive}) and Eq.~\eqref{eq:minimum}, we obtain\footnote{This is a special case of Céa's lemma~\cite{cea-approximationvariationnelleproblemes-1964}, which gives general bounds on Galerkin approximations with respect to suitable norms on Banach spaces~\cite[p.~64]{brenner-mathematicaltheoryfinite-2008}.}
\begin{align}\label{proofeq:cea}
\norm{\tilde{Q}_n\psi-\psi}^2&=\norm{\psi-\psi_n}^2 \nonumber\\&\leq   \frac{1}{\gamma}\tilde{q}(\psi-\psi_n,\psi-\psi_n)\nonumber\\&\leq\frac{1}{\gamma}\tilde{q}(\psi-\phi_n,\psi-\phi_n)\nonumber\\&\leq\frac{1}{\gamma}\|\psi-\phi_n\|^2_+\qquad\text{for all }\phi_n\in\hilbert_n,
\end{align}
with the norm $\|\cdot\|_+$ as defined in Eq.~\eqref{eq:plusnorm}.

Now, recall that, by Definition~\ref{def:friedrichs_extension}, the domain of the Friedrichs extension $\domain(\tilde{H}_{\rm fin})$ of $H_{\rm fin}$ satisfies
\begin{equation}\label{eq:crucial_inclusion}
    \domain(\tilde{H}_{\rm fin})\subset\overline{\domain(H_{\rm fin})}^{\norm{}_+} \, =\overline{\bigcup_{n \in \nnum} P_n \hilbert}^{\norm{}_+} \, ,
\end{equation}
Let $\epsilon>0$. By the equality above, there exists $N_\epsilon\in\mathbb{N}$ and $\phi_{N_\epsilon}\in\hilbert_{N_\epsilon}$ such that
\begin{equation}\label{eq:epsilon}
    \|\psi-\phi_{N_\epsilon}\|_+\leq\epsilon\sqrt{\gamma};
\end{equation}
besides, clearly $\phi_{N_\epsilon}\in\hilbert_n$ for every $n\geq N_\epsilon$ as well since we required $P_mP_n=P_m$ for all $n\geq m$. By combining Eqs.~\eqref{proofeq:cea} and~\eqref{eq:epsilon} we finally get, for all $n\geq N_\epsilon$,
\begin{equation}
    \norm{\tilde{Q}_n\psi-\psi}\leq\frac{1}{\sqrt{\gamma}}\|\psi-\phi_{N_\epsilon}\|_+\leq\epsilon,
\end{equation}
and thus $\lim_{n\to\infty}\|\tilde{Q}_n\psi-\psi\|=0$ as claimed.
\end{proof}

\begin{remark}
    Note that, in the proof of Proposition~\ref{prop:friedrich-galerkin-projection}, the fact that $\tilde{H}_{\rm fin}$ is precisely the Friedrichs extension of $H_{\rm fin}$ (and not any other self-adjoint extension of it) entered our argument when applying Eq.~\eqref{eq:crucial_inclusion}, since the Friedrichs extension of $H_{\rm fin}$ is its unique self-adjoint extension satisfying this property.
\end{remark}

\begin{proof}[Proof of Proposition~\ref{prop:friedrichs}]

Let us start by proving the claim for coercive $H$, i.e. there is $\gamma>0$ such that $\braket{\psi,H\psi}\geq\gamma\|\psi\|^2$ for all $\psi\in\domain(H)$. In this case, by Proposition~\ref{prop:friedrich-galerkin-projection} we know that the Galerkin projectors $\tilde{Q}_n$ associated with the Friedrichs extension $\tilde{H}_{ \rm fin}$ of $H_{\rm fin}$ satisfy
\begin{equation}
    \lim_{n\to\infty}\tilde{Q}_n\psi=\psi\qquad\text{for all }\psi\in\domain(\tilde{H}_{\rm fin}),
\end{equation}
which implies, by Proposition~\ref{prop:galerkin_criterion}, 
\begin{equation}
    \lim_{n \to \infty} U_n(t) \psi = \tilde{U}(t) \psi \quad \text{for all } \psi \in \hilbert.
\end{equation}
Let us now consider the general case in which $H$ is bounded from below by some possibly nonpositive $\gamma\in\mathbb{R}$. In such a case, let $A:=H-\gamma+1$. This is clearly a coercive operator, since $\braket{\psi,A\psi}\geq\|\psi\|^2$; besides, $A_{\rm fin}=H_{\rm fin}-\gamma+1$ and, by Lemma~\ref{lem:friedrich_plus_b}, their Friedrichs extensions also differ by the same constant,
\begin{equation}
    \tilde{A}_{\rm fin}=\tilde{H}_{\rm fin}-\gamma+1;
\end{equation}
The Galerkin approximations of the two operators satisfy
\begin{equation}
    A_n:=P_nAP_n=H_n+(1-\gamma)P_n,
\end{equation}
and, recalling the decomposition $H_n=\hat{H}_n\oplus0_{\hilbert_n^\perp}$ and noticing that $P_n=\id_{\hilbert_n}\oplus0_{\hilbert_n^\perp}$,
\begin{equation}
    A_n=\hat{A}_n\oplus\id_{\hilbert_n^\perp}=\left(\hat{H}_n+1-\gamma\right)\oplus0_{\hilbert_n^\perp}.
\end{equation}
As such, defining $\tilde{T}(t)=\e^{-\iu t\tilde{A}_{\rm fin}}$ and $T_n(t)=\e^{-\iu tA_n} $, the first propagator only differs from $\tilde{U}(t)$ by a phase term:
\begin{equation}
    \tilde{U}(t)=\e^{\iu(1-\gamma)t}\tilde{T}(t);
\end{equation}
and, as for $T_n(t)$ and $U_n(t)$, we have
\begin{align}
    T_n(t)&=\e^{-\iu(1-\gamma)t}\e^{-\iu t\hat{H}_n}\oplus\id_{\hilbert_n^\perp};\\
        U_n(t)&=\e^{-\iu t\hat{H}_n}\oplus\id_{\hilbert_n^\perp},
\end{align}
whence
\begin{align}
    \e^{\iu(1-\gamma)t}T_n(t)&=\e^{-\iu t\hat{H}_n}\oplus \e^{\iu(1-\gamma)t}\id_{\hilbert_n^\perp}\nonumber\\
    &=\e^{-\iu t\hat{H}_n}\oplus \id_{\hilbert_n^\perp}+0_{\hilbert_n}\oplus(\e^{\iu(1-\gamma)t}-1)\id_{\hilbert_n^\perp}\nonumber\\
    &=U_n(t)+0_{\hilbert_n}\oplus(\e^{\iu(1-\gamma)t}-1)\id_{\hilbert_n^\perp},
\end{align}
implying
\begin{align}
U_n(t)-\tilde{U}(t)&=\e^{\iu(1-\gamma)t}\left(T_n(t)-\tilde{T}(t)\right)\nonumber\\&+0_{\hilbert_n}\otimes(1-\e^{\iu(1-\gamma)t})\id_{\hilbert_n^\perp};
\end{align}
the first term converges strongly to zero as discussed above, and the second term also converges strongly to zero since $P_n$ converges strongly to the identity, thus completing the proof.
\end{proof}

\section{Proofs of Section~\ref{sec:particle_results}}\label{sec:proofs_sec3}

In this section we shall provide proofs for the convergence statements on the Laplace operator on $(-1,1)$, namely: Propositions~\ref{prop:friedrichs-dirichlet},~\ref{prop:assoc_legendre_diff_1perp},~\ref{prop:alpha-periodic-unitary}, and Theorem~\ref{thm:alpha-periodic-legendre}.

\subsection{Proof of Proposition~\ref{prop:friedrichs-dirichlet}}\label{sec:proof_prof_friedrichs-dirichlet}

\begin{proof}[Proof of Proposition~\ref{prop:friedrichs-dirichlet}]
    Let $q_\Phi:\domain (H_\Phi)\times \domain(H_\Phi)\rightarrow\mathbb{C}$ be the sesquilinear form associated with $H_\Phi$,
    \begin{equation}
        q_\Phi(\psi,\varphi)=-\int_{-1}^{+1}\overline{\psi(x)}\varphi''(x)\;\mathrm{d}x=\int_{-1}^{+1}\overline{\psi'(x)}\varphi'(x)\;\mathrm{d}x,
    \end{equation}
    where, in the last step, we performed an integration by parts and used the fact that all functions in $\operatorname{Span}\Phi$ satisfy Dirichlet boundary conditions. Therefore, $H_\Phi$ is a nonnegative operator, and the form $q_\Phi$ is closable; the closure of $q_\Phi$, which we denote by $\tilde{q}_\Phi$, is defined on $\overline{\domain( H_\Phi)}^{\|\cdot\|_+}$, where the norm $\|\cdot\|_+$ is given by
    \begin{equation}
        \|\psi\|^2_+=\|\psi'\|^2+\|\psi\|^2,
    \end{equation}
    and thus, in particular, coincides with the Sobolev norm $\|\cdot\|_{\mathrm{H}^1}$, cf.~Eq.~\eqref{eq:sobolev}. We claim the following equality:
    \begin{equation}\label{eq:claim}
    \overline{\domain (H_\Phi)}^{\|\cdot\|_+}=\mathrm{H}^1_0((-1,1)).
    \end{equation}

    We begin by proving $\overline{\domain( H_\Phi)}^{\|\cdot\|_+}\subset\mathrm{H}^1_0((-1,1))$. Let $\psi\in\overline{\domain( H_\Phi)}^{\|\cdot\|_+}$, meaning that there exists a sequence $(\psi_n)_{n\in\mathbb{N}}\subset\operatorname{Span}\Phi$ such that $\|\psi_n-\psi\|_+\to0$, i.e.
    \begin{equation}
        \|\psi_n-\psi\|\to0\qquad\text{and}\qquad\|\psi'_n-\psi'\|\to0.
    \end{equation}
    Necessarily, $\psi\in\mathrm{H}^1((-1,1))$; we must prove that it also satisfies Dirichlet boundary conditions. By Proposition~\ref{prop:continuous_representative}, $\psi$ is continuous up to the boundary (in the sense that it admits a continuous representative, with which we hereafter identify it). Now, $\|\psi_n-\psi\|\to0$ implies that there exists a subsequence $(\psi_{n_k})_{k\in\mathbb{N}}\subset\operatorname{Span}\Phi$ such that $\psi_{n_k}(x)\to\psi(x)$ almost everywhere~\cite[p.~312]{jost-partialdifferentialequations-2002}; since $\psi$ is continuous, $\psi_{n_k}(x)\to\psi(x)$ everywhere. But then, in particular,
    \begin{equation}
    \label{proofeq:friedrich-dirichlet-limit}
        \psi(-1)=\lim_{k\to\infty}\psi_{n_k}(-1)=0,\qquad\psi(1)=\lim_{k\to\infty}\psi_{n_k}(1)=0
    \end{equation}
    since each $\psi_{n_k}$ satisfies Dirichlet boundary conditions. This proves the desired inclusion.

    We now proceed to prove the converse inclusion $\overline{\domain( H_\Phi)}^{\|\cdot\|_+}\supset\mathrm{H}^1_0((-1,1))$. Let $\psi\in\mathrm{H}^1_0((-1,1))$, and thus $\psi'\in L^1((-1,1))$; besides, 
    \begin{equation}
        \label{proofeq:friedrich-dirichlet-approx-start}
        \braket{1,\psi'}=\int_{-1}^{+1}\psi'(x)\;\mathrm{d}x=\psi(1)-\psi(-1)=0,
    \end{equation}
    whence $\psi'\in\{1\}^\perp=\overline{\operatorname{Span}\Phi'}$, implying that there exists a sequence $(\xi_n)_{n\in\mathbb{N}}\subset\operatorname{Span}\Phi'$ such that $\|\xi_n-\psi'\|\to0$. But then, by linearity, each $\xi_n$ is the first derivative of a function $\psi_n\in\operatorname{Span}\Phi$, i.e. there exists a sequence $(\psi_n)_{n\in\mathbb{N}}\subset\operatorname{Span}\Phi$ such that
    \begin{equation}
        \|\psi'_n-\psi'\|\to0.
    \end{equation}
    Furthermore, as $\Phi\subset\mathrm{H}^1_0((-1,1))$, both $\psi_n$ and $\psi$ are in $\mathrm{H}^1_0((-1,1))$; as such, Poincarè's inequality (Proposition~\ref{prop:poincare}) applies:
    \begin{equation}
    \label{proofeq:friedrich-dirichlet-approx}
        \|\psi_n-\psi\|\leq\|\psi'_n-\psi'\|\to0,
    \end{equation}
    and thus $\|\psi_n-\psi\|_+\to0$. This proves the desired inclusion, and thus the claimed equality~\eqref{eq:claim}.

    We have thus proven that the domain of the closure $\tilde{q}_\Phi$ of the sesquilinear form associated with $H_\Phi$ is $\mathrm{H}^1_0((-1,1))$; as such, the Friedrichs extension $\tilde{H}_\Phi$ of $H_\Phi$ admits $\mathrm{H}^1_0((-1,1))$ as form domain, and in particular $\domain(\tilde{H}_\Phi)\subset\mathrm{H}^1_0((-1,1))$. By the general properties of Friedrichs extension (Proposition~\ref{prop:def-friedrich}), $\tilde{H}_\Phi$ is also the unique self-adjoint extension of $H_\Phi$ such that $\domain(\tilde{H}_\Phi)\subset\mathrm{H}^1_0((-1,1))$. But the Dirichlet realization $H_{\mathrm{ Dir}}$ of the Laplace operator has domain
    \begin{equation}
        \domain (H_{\mathrm{Dir}})=\mathrm{H}^1_0((-1,1))\cap\mathrm{H}^2((-1,1))\subset\mathrm{H}^1_0((-1,1)),
    \end{equation}
    cf.~Remark~\ref{rem:dirichlet}, and is itself a self-adjoint extension of $H_\Phi$ since $\Phi\subset\mathrm{H}^1_0((-1,1))\cap\mathrm{H}^2((-1,1))$ by assumption. By the aforementioned uniqueness of the Friedrichs extension, $\tilde{H}_\Phi=H_{\mathrm{Dir}}$.
\end{proof}

\subsection{Proof of Proposition~\ref{prop:assoc_legendre_diff_1perp}}
    \label{sec:assoc_legendre_diff_1perp}
    
       \begin{proof}[Proof of Proposition~\ref{prop:assoc_legendre_diff_1perp}]
We will prove the following equality:
    \begin{equation}
        \overline{\operatorname{Span}\left(P_l^m\right)'_{l \geq m}} = \{1\}^\perp,
    \end{equation}
    which, of course, implies the claimed equality $ \overline{\operatorname{Span}\left(p_l^m\right)'_{l \geq m}} = \{1\}^\perp$ since the linear span is unaffected by normalization factors.
       
        Let $m \geq 4$ fixed.
        We start by proving the inclusion 
        \begin{equation}            \overline{\operatorname{Span}\left(P^m_l\right)'_{l \geq m}} \subseteq \{1\}^\perp.
        \end{equation}
        To this purpose, recall that the associated Legendre Polynomials for $m \geq 4$ fulfill Dirichlet boundary conditions (Proposition~\ref{prop:boundary_data_legendre}), whence 
        \begin{equation}\label{eq:inclusion}
            \braket{1,(P^m_l)'}=\int_{-1}^1\frac{\mathrm{d}}{\mathrm{d}x}P^m_l(x)\;\mathrm{d}x=P^m_l(1)-P^m_l(-1)=0.
        \end{equation}
    By linearity, this readily implies $\braket{1,\psi}=0$ for all $\psi\in\operatorname{Span}\left(P_l^m\right)'_{l \geq m}$; furthermore, by the continuity of the scalar product, the same holds for all $\psi\in \overline{\operatorname{Span}\left(P_l^m\right)'_{l \geq m}}$.

    We shall now show that, in fact, the two sets are equal. Let us begin by showing the following equality for all $n\in\mathbb{N}$:
        \begin{equation}
            \label{proofeq:der-al-1perp-1}
            \operatorname{Span}\left(P_l^m\right)_{l=m}^{m+n} = \operatorname{Span}\left((1-x^2)^{m/2} x^l \right)_{l=0}^n \, .
        \end{equation}
        To this purpose, let $\psi\in\operatorname{Span}\left(P_l^m\right)_{l=m}^{m+n}$; from the definition~\eqref{eq:def-assoc-legendre} of associated Legendre polynomials, this means that there exist $c_m,\dots,c_{m+n}\in\mathbb{C}$ such that
        \begin{equation}
            \psi(x)=\sum_{l=m}^{m+n}c_lP^m_l(x)=(1-x^2)^{m/2}\frac{\mathrm{d}^m}{\mathrm{d}x^m}\sum_{l=m}^{m+n}c_lP_l(x)
        \end{equation}
        and thus also
        \begin{equation}
            \psi(x)=(1-x^2)^{m/2}\frac{\mathrm{d}^m}{\mathrm{d}x^m}\sum_{l=0}^{m+n}c_lP_l(x)
        \end{equation}
with $c_0=\ldots=c_{m-1}=0$. But, since $\operatorname{Span}(P_l)_{l=0}^{m+n}=\operatorname{Span}(x^l)_{l=0}^{m+n}$ (Eq.~\eqref{eq:legendre-poly-gramschmidt}), this is equivalent to stating that there exist $\tilde{c}_1,\dots,\tilde{c}_{m+n}\in\mathbb{C}$ such that
\begin{equation}
\begin{split}
    \psi(x)=(1-x^2)^{m/2}\frac{\mathrm{d}^m}{\mathrm{d}x^m}\sum_{l=0}^{m+n}\tilde{c}_lx^l&=(1-x^2)^{m/2}\sum_{l=m}^{m+n}\tilde{c}_l l (l-1) \dots (l-m+1) x^{l-m}\\&=(1-x^2)^{m/2}\sum_{l=0}^{n}\tilde{c}_{l+m} (l+m) \dots (l+1) x^{l}\\
    &=\sum_{l=0}^{n}\tilde{c}_{l+m} (l+m) \dots (l+1) (1-x^2)^{m/2}x^{l},
    \end{split}
\end{equation}
thus being equivalent to $\psi\in\operatorname{Span}\left((1-x^2)^{m/2}x^l\right)_{l= 0}^n$, which proves Eq.~\eqref{proofeq:der-al-1perp-1}. The latter implies
\begin{align}
            \label{proofeq:der-al-1perp-1bis}
            \operatorname{Span}\left((P_l^m)'\right)_{l=m}^{m+n} &= \operatorname{Span}S_n;\\            \label{proofeq:der-al-1perp-1ter}
             \operatorname{Span}\left((P_l^m)'\right)_{l=m}^{\infty} &= \operatorname{Span}S,
        \end{align}
where $S_n:=(v_l)_{l=0}^n$ and $S:=(v_l)_{l=0}^\infty$, with each $v_l(x)$ being the first derivative of $(1-x^2)^{m/2}x^l$, namely:
\begin{align}
                v_l(x) = \diff*{(1-x^2)^{m/2}x^l}{x} &= (1-x^2)^{m/2-1} \left(l x^{l-1} - (m+l)x^{l+1} \right)\nonumber\\
                &\equiv v_{-1}(x)\left(l x^{l-1} - (m+l)x^{l+1} \right)\qquad\qquad \forall l > 0\, , \\
                v_0(x)  = \diff*{(1-x^2)^{m/2}x^0}{x}  &= -m(1-x^2)^{m/2-1}x\nonumber\\
                & \equiv v_{-1}(x)(-mx)\, ,
        \end{align} 
 where we defined $v_{-1}(x) := (1-x^2)^{m/2-1}$. Then the following equality holds:
 \begin{equation}
            \label{proofeq:der-al-1perp-2}
            \operatorname{Span} \left(S_n \cup \{v_{-1}\}\right) = \operatorname{Span}\left(v_{-1}\,x^l\right)_{l=0}^{n+1}\quad \forall n \geq m \, .
\end{equation}
Indeed, let $\psi\in\operatorname{Span} \left(S_n \cup \{v_{-1}\}\right)$. Then there exist $c_{0},c_1,\dots,c_{n+1}\in\mathbb{C}$ such that
\begin{equation}
\begin{split}
   \psi(x)&=v_{-1}(x)\left(c_{0}+c_1(-mx)+\sum_{l=1}^nc_{l+1}\left(l x^{l-1} - (m+l)x^{l+1} \right)\right)\nonumber\\
   &\equiv v_{-1}(x)\sum_{l=0}^{n+1}d_{l}x^l,
\end{split}
\end{equation}
where, by a direct comparison of the two sums in the above equation, the coefficients $d_{0},d_1,d_2,\dots,d_{n+1}$ are defined through
\begin{equation}
\label{eq:proof-plm1perp-coeffs}
\begin{split}
    d_{0}&=c_{0}+c_2;\\
    d_1&=-m\,c_1+2c_3;\\
    d_2&=-(m+1)c_2+3c_4;\\
    \vdots &=   \quad\vdots \\
    d_{l}&=-(m+l-1)c_{l}+(l+1)c_{l+2};\\
        \vdots &=   \quad\vdots \\
    d_{n-1}&=-(m+n-2)c_{n-1}+nc_{n+1};\\
    d_{n}&=-(m+n-1)c_{n};\\
    d_{n+1}&=-(m+n)c_{n+1},
    \end{split}
\end{equation}
which can be recast as a matrix equation: $\mathsf{d}=\mathbb{A}\mathsf{c}$, where $\mathsf{d}=(d_{0},d_1,\ldots,d_{n+1})^\intercal$, $\mathsf{c}=(c_{0},c_1,\ldots,c_{n+1})^\intercal$, and $\mathbb{A}$ is a $(n+1)\times(n+1)$ upper triangular matrix whose diagonal elements are all nonzero; as such, $\mathbb{A}$ is invertible and the relation above can be inverted. Therefore, $\psi\in\operatorname{Span} \left(S_n \cup \{v_{-1}\}\right)$ if and only if $\psi\in\operatorname{Span}(v_{-1}\,x^l)_{l=0}^{n+1}$. This proves Eq.~\eqref{proofeq:der-al-1perp-2} for all $n\in\mathbb{N}$. From Eq.~\eqref{proofeq:der-al-1perp-2}, we then get
\begin{equation}\label{eq:piece1}
     \operatorname{Span} \left(S \cup \{v_{-1}\}\right) = \operatorname{Span}\left(v_{-1}\,x^l\right)_{l=0}^{\infty}.
\end{equation}
At the same time, replacing $m$ with $m-2$ in Eq.~\eqref{proofeq:der-al-1perp-1}, we obtain
\begin{equation}
            \label{proofeq:der-al-1perp-1-bis}
            \operatorname{Span}\left(P_l^{m-2}\right)_{l=m-2}^{m+n-2} = \operatorname{Span}\left((1-x^2)^{m/2-1} x^l \right)_{l=0}^n = \operatorname{Span}\left(v_{-1} x^l \right)_{l=0}^n,
        \end{equation}
        which implies
        \begin{equation}\label{eq:piece2}
            \operatorname{Span}\left(P_l^{m-2}\right)_{l=m-2}^{\infty} = \operatorname{Span}\left(v_{-1} x^l \right)_{l=0}^\infty.
        \end{equation}
Consequently, equating Eqs.~\eqref{eq:piece1} and~\eqref{eq:piece2}, we get
\begin{equation}
     \operatorname{Span} \left(S \cup \{v_{-1}\}\right)= \operatorname{Span}\left(P_l^{m-2}\right)_{l=m-2}^{\infty},
\end{equation}
and finally, taking closures,
\begin{equation}\label{eq:step}
     \overline{\operatorname{Span} \left(S \cup \{v_{-1}\}\right)}= \overline{\operatorname{Span}\{P_l^{m-2}\}_{l=m-2}^{\infty}}=\hilbert,
\end{equation}
where in the last step we used the fact that $(P^l_{m-2})_{l\geq m-2}$ is a complete orthonormal set in $\hilbert$.

We now claim that $v_{-1}\in\overline{\operatorname{Span}(S\cup\{1\})}$. To begin with, because of Eq.~\eqref{eq:step}, for every $\epsilon>0$ there exist $\alpha_\epsilon\in\mathbb{C}$ and $u_\epsilon\in\operatorname{Span}\,S$ such that
\begin{equation}
            \label{proofeq:der-al-1perp-3}
            \norm{1 - (\alpha_\epsilon v_{-1} + u_\epsilon)} < \epsilon \, ,
        \end{equation}
whence, by the Cauchy--Schwarz inequality,
\begin{equation}
    \left|\braket{1,1 - (\alpha_\epsilon v_{-1} + u_\epsilon}\right|\leq\|1\|\norm{1 - (\alpha_\epsilon v_{-1} + u_\epsilon)} < \sqrt{2}\epsilon \, .
\end{equation}
The left-hand side of the equation above can be computed explicitly: indeed, $\braket{1,1}=\int_{-1}^1\mathrm{d}x=2$, while $\braket{1,v_{-1}}$ computes to
\begin{equation}\label{eq:cauchyschwarz}
    \braket{1,v_{-1}}=\int_{-1}^1(1-x^2)^{m/2-1}\;\mathrm{d}x=\frac{\sqrt{\pi} \Gamma(m/2)}{\Gamma((m+1)/2)}=:\kappa_m>0;
\end{equation}
finally, $\braket{1,u_\epsilon}=0$, since $u_\epsilon\in\operatorname{Span}\,S=\operatorname{Span}\{(P^m_l)'\}_{l=m}^\infty$ (Eq.~\eqref{proofeq:der-al-1perp-1ter}) and the latter is contained in $\{1\}^\perp$ (Eq.~\eqref{eq:inclusion}). As such, Eq.~\eqref{eq:cauchyschwarz} becomes
\begin{equation}
    |2-\alpha_\epsilon\kappa_m|<\sqrt{2}\epsilon,
\end{equation}
whence, by the reverse triangle inequality,
\begin{equation}
    \left|2-|\alpha_\epsilon|\kappa_m\right|<|2-\alpha_\epsilon\kappa_m|<\sqrt{2}\epsilon.
\end{equation}
This holds for all $\epsilon>0$. Let us take $\epsilon<1$; then 
\begin{equation}
        \left|2-|\alpha_\epsilon|\kappa_m\right|<\sqrt{2}
\end{equation}
and therefore $|\alpha_\epsilon|\kappa_m\in(2-\sqrt{2},2+\sqrt{2})$; in particular, $\alpha_\epsilon\neq0$. But then Eq.~\eqref{proofeq:der-al-1perp-3} can be written as
\begin{equation}
    |\alpha_\epsilon|\left\|v_{-1}-\frac{1-u_\epsilon}{\alpha_\epsilon}\right\|<\epsilon
\end{equation}
whence, using $|\alpha_\epsilon|\kappa_m\geq2-\sqrt{2}$,
\begin{equation}
    \left\|v_{-1}-\frac{1-u_\epsilon}{\alpha_\epsilon}\right\|<\frac{\epsilon}{|\alpha_\epsilon|}\leq\frac{\epsilon\kappa_m}{2-\sqrt{2}}.
\end{equation}
As $\kappa_m/(2-\sqrt{2})$ is an $\epsilon$-independent constant, we have shown that indeed $v_{-1}$ is in $\overline{\operatorname{Span}(S\cup\{1\})}$.

We can now wrap up our proof. Since $v_{-1}\in\overline{\operatorname{Span}(S\cup\{1\})}$, we have $\operatorname{Span}\,(S\cup\{v_{-1}\})\subset\overline{\operatorname{Span}(S\cup\{1\})}$ and thus, taking closures,
\begin{equation}
    \overline{\operatorname{Span}\,(S\cup\{v_{-1}\})}\subset\overline{\operatorname{Span}(S\cup\{1\})};
\end{equation}
but, on the other hand, we had already proven $\overline{\operatorname{Span}\,(S\cup\{v_{-1}\})})=\hilbert$; therefore the equation above yields $\overline{\operatorname{Span}(S\cup\{1\})}=\hilbert$, which is the desired claim.
    \end{proof}

\subsection{Proof of Proposition~\ref{prop:alpha-periodic-unitary}}
\label{sec:proof-alpha-periodic-unitary}

We will now prove sufficient conditions for the convergence of $U_n(t)$ to the evolution group $U_{W(\alpha)}(t)$ generated by the Laplacian with $\alpha$-periodic boundary conditions. In the following, $\alpha\in[0,2\pi)$ shall be fixed.
\begin{lemma}
\label{lem:alpha-sobolev-dense}
    The subspace of $\mathrm{H}^1((-1,1))$ defined by
    \begin{equation}
        \mathrm{H}^1_{\alpha,\mathrm{per}}((-1,1)) = \left\{ f \in \mathrm{H}^1((-1,1)) \, , \, \, \, f(-1) = \e^{\iu \alpha} f(1) \right\}    \end{equation}
    is closed with respect to the norm $\norm{\cdot}_+$ defined by $\|\psi\|^2_+ = \norm{\psi'}^2+\norm{\psi}^2$.
\end{lemma}
\begin{proof}
    Let $\psi \in \overline{\mathrm{H}^1_{\alpha,\mathrm{per}}((-1,1))}^{\norm{\cdot}_+}$.
    Then there exist $(\psi_n)_{n\in\mathbb{N}} \subset \mathrm{H}^1_{\alpha,\mathrm{per}}((-1,1))$ such that $\norm{\psi-\psi_n}_+ \to 0$, and therefore $\psi \in \mathrm{H}^1((-1,1))$. We follow analogous steps as in the proof of Proposition~\ref{prop:friedrichs-dirichlet}: without loss of generality we can assume that $\psi$ is continuous up to the boundary (see Proposition~\ref{prop:continuous_representative}). 
    $\norm{\psi-\psi_n}_+ \to 0$ implies $\norm{\psi-\psi_n} \to 0$ and thus there exists a subsequence $\psi_{n_k}$ with $\psi_{n_k}(x) \to \psi(x)$ almost everywhere~\cite[p.~312]{jost-partialdifferentialequations-2002}: as $\psi$ is continuous, $\psi_{n_k}(x) \to \psi(x)$ for all $x \in [-1,1]$.
    Consequently,
    \begin{equation}
            \psi(-1) = \lim_{k \to \infty} \psi_{n_k}(-1) = \lim_{k \to \infty} \e^{\iu \alpha} \psi_{n_k}(1) = \e^{\iu \alpha} \psi(1) \, .
    \end{equation}
    Thus, $\psi \in \mathrm{H}^1_{\alpha,\mathrm{per}}((-1,1))$.
\end{proof}
Note that, differently from the Dirichlet case, $\domain(H_{W(\alpha)}) \subsetneq \mathrm{H}^1_{\alpha,\mathrm{per}}((-1,1)) \cap \mathrm{H}^2((-1,1))$.

\begin{lemma}
\label{prop:alpha-periodic}
    Let $\Phi = (\phi_l)_{l \in \nnum}$ be a complete orthonormal set of $L^2((-1,1))$ satisfying the following additional properties
    \begin{itemize}
        \item[(i)] $\Phi\subset \domain(H_{W(\alpha)})$;
        \item[(ii)] There exists $\phi_0 \in \Phi$ with $\phi_0(1) \neq 0$ such that $\overline{\operatorname{Span}\tilde{\Phi}'}=\{1\}^\perp$,
    \end{itemize}
    where $\tilde{\Phi} = \{\phi(x)-\frac{\phi(1)}{\phi_0(1)} \phi_0(x) , \, \phi \in \Phi \}$, and $\tilde{\Phi}'$ denotes the set of derivatives of elements of $\tilde{\Phi}$. Then the symmetric operator $H_\Phi$ defined by
    \begin{equation}
        \domain (H_\Phi)=\operatorname{Span}\Phi,\qquad H_\Phi\psi=-\psi''
    \end{equation}
    has a Friedrichs extension equal to the $\alpha$-periodic  Laplacian $H_{W(\alpha)}$ with domain $\domain(H_{W(\alpha)})$.
\end{lemma}
\begin{proof}
    We follow the same strategy of the proof of Proposition~\ref{prop:friedrichs-dirichlet}, starting with showing that the quadratic form $q_\Phi: \domain(H_\phi) \times \domain(H_\Phi) \to \cnum $ given by $q_\Phi(\psi,\varphi) = \braket{\psi,H_\Phi \varphi}$ is positive:
    \begin{equation}
        q_\Phi(\psi,\varphi) = -\int_{-1}^{+1} \overline{\psi(x)}\varphi''(x)  \dl x = \int_{-1}^{+1} \overline{\psi'(x)}\varphi'(x) \dl x - \left[\overline{\psi(x)}\varphi'(x)\right]_{-1}^{+1} \geq 0 \, .
    \end{equation}
    The boundary term vanishes as $\overline{\psi(-1)}\varphi'(-1) = \overline{\e^{\iu \alpha} \psi(1)}\e^{\iu \alpha} \varphi'(1) = \overline{\psi(1)} \varphi'(1)$.
    Thus, $q_\Phi$ is closable with closure $\tilde{q}_\Phi$.
   The norm $\|\cdot\|_+$, defined as in Definition~\ref{def:friedrichs_extension}, is again equal to
    \begin{equation}
        \norm{\psi}_+^2 = q_\Phi(\psi,\psi) + \norm{\psi}^2 = \norm{\psi'}^2+\norm{\psi}^2
    \end{equation}
     We claim the following equality:
    \begin{equation}
         \domain(\tilde{q}_\Phi) = \overline{\domain(H_\Phi)}^{\norm{}_+} = \mathrm{H}^1_{\alpha,\mathrm{per}}((-1,1)) \,.
    \end{equation}
    The right-hand side is indeed closed with respect to $\norm{}_+$, as shown in Lemma~\ref{lem:alpha-sobolev-dense}.
    
    We begin by proving $\overline{\domain(H_\Phi)}^{\norm{}_+} \subset \mathrm{H}^1_{\alpha,\mathrm{per}}((-1,1))$.
    Let $\psi \in \overline{\domain(H_\Phi)}^{\norm{}_+}$.
    Then there exists a sequence $(\psi_n)_{n\in\mathbb{N}} \in \domain(H_\Phi)$ such that $\norm{\psi-\psi_n}_+ \to 0$, and therefore, $\psi \in \mathrm{H}^1((-1,1))$.
    Besides, following the same steps as in the proof of Lemma~\ref{lem:alpha-sobolev-dense}, one readily shows
    \begin{equation}
        \psi(-1)  = \e^{\iu \alpha} \psi(1) \, ,
    \end{equation}
    whence $\psi\in\mathrm{H}^1_{\alpha,\mathrm{per}}((-1,1))$.
    
    We now prove the inclusion $ \mathrm{H}^1_{\alpha,\mathrm{per}}((-1,1)) \subset \overline{\domain(H_\Phi)}^{\norm{}_+} $.
    Let $\psi \in  \mathrm{H}^1_{\alpha,\mathrm{per}}((-1,1))$, i.e. $\psi \in \mathrm{H}^1((-1,1))$ and $\psi(-1) =  \e^{\iu \alpha} \psi(1)$, and define 
    \begin{equation}
        \tilde{\psi}(x) = \psi(x) - \frac{\psi(1)}{\phi_0(1)} \phi_0(x) \, .
    \end{equation}
    Then, as per assumption $\phi_0(1) \neq 0$, $\tilde{\psi}(\pm 1) = 0$, and as both $\psi,\phi_0 \in \mathrm{H}^1((-1,1))$, we get $\tilde{\psi}\in \mathrm{H}^1_0((-1,1))$.
    According to assumption \textit{(ii)}, $\operatorname{Span}\tilde{\Phi}' = \{1\}^\perp$.
    Following analogous steps as in the proof of Proposition~\ref{prop:friedrichs-dirichlet}, cf.~Eqs.~\ref{proofeq:friedrich-dirichlet-approx-start}--\ref{proofeq:friedrich-dirichlet-approx}, there exists a sequence $(\tilde{\psi}_n)_{n\in\mathbb{N}} \subset \operatorname{Span}\tilde{\Phi}$ such that $\norm{\tilde{\psi} - \tilde{\psi}_n}_+ \to 0$.
    Let
    \begin{equation}
        \psi_n(x) = \tilde{\psi}_n(x) + \frac{\psi(1)}{\phi_0(1)} \phi_0(x) \, .
    \end{equation}
    Per definition, $\operatorname{Span}\tilde{\Phi} \subset \operatorname{Span} \Phi$. Then, as $\phi_0 \in \Phi$ and $\tilde{\psi} \in \operatorname{Span}\tilde{\Phi}$, all $\psi_n \in \operatorname{Span}\Phi$ as well.
    But 
    \begin{equation}
        \norm{\psi-\psi_n}_+ = \norm*{\tilde{\psi} + \frac{\psi(1)}{\phi_0(1)} \phi_0 - \tilde{\psi}_n - \frac{\psi(1)}{\phi_0(1)} \phi_0}_+ = \norm{\tilde{\psi} - \tilde{\psi}_n}_+ \to 0 \, ,
    \end{equation}
    proving the desired inclusion.
    Thus, the form domain of the closure $\tilde{q}_\Phi$ is equal to $\mathrm{H}^1_{\alpha,\mathrm{per}}((-1,1))$.
    For the domain of the $\alpha$-periodic Hamiltonian $H_{W(\alpha)}$ the following inclusions hold:
    \begin{equation}
        \domain(H_\Phi) \subset \domain(H_{W(\alpha)}) \subset \mathrm{H}^1_{\alpha,\mathrm{per}}((-1,1))\, .
    \end{equation}
    Thus, $H_{W(\alpha)}$ is a self-adjoint extension of $H_\Phi$ and its domain is a subset of the form domain of $\tilde{q}_\Phi$; by uniqueness of the Friedrichs extension (Proposition~\ref{prop:def-friedrich}) we get the desired equality $\tilde{H}_\Phi = H_{W(\alpha)}$.  
\end{proof}

\begin{proof}[Proof of Proposition~\ref{prop:alpha-periodic-unitary}]
    We claim that $H_{W(\alpha)}$ is the Friedrichs extension of the operator $H_{\mathrm{fin}} $ given by
    \begin{equation}
     \domain(H_{W(\alpha)}) = \bigcup_{n \in \nnum} P_n \hilbert \, ,
    \end{equation}
    where the $P_n$ are the projections on the first $n$ elements of $\Phi$. 
    Indeed, $H_{W(\alpha)} \psi = - \psi''$, $\domain(H_{W(\alpha)}) = \operatorname{Span} \Phi$, and as $\Phi$ satisfies the conditions of the previous Lemma~\ref{prop:alpha-periodic}, $H_{W(\alpha)}$ is the Friedrichs extension of $H_{\mathrm{fin}} $.
    Then, using Theorem~\ref{prop:friedrichs}, the result follows.
\end{proof}    

\subsection{Proof of Theorem~\ref{thm:alpha-periodic-legendre}}
\label{sec:proof-thm-alpha-periodic-legendre}
\begin{proof}[Proof of Theorem~\ref{thm:alpha-periodic-legendre}]

    We show that the basis $$(\phi_l)_{l \in \nnum} = \mathrm{GS}(f_0,(p_k^m)_{k \geq m}) $$ fulfills the two conditions of Proposition~\ref{prop:alpha-periodic-unitary}.
    To this end, notice that we can write $\phi_l$, the $l$th element of the Gram-Schmidt orthogonalization of $f_0$ and $(p_k^m)_{k \geq m}$, as
    \begin{equation}
        \label{proofeq:thm-alphaperiodic-gs-1}
        \phi_l(x) = \beta^{(l)} f_0(x) + \sum_{k = 1}^{l}\gamma_k^{(l)} p_{k+m-1}^m(x)\, , 
    \end{equation}
    with coefficients $\beta^{(l)}$ and $\gamma_k^{(l)}$ in $\cnum$.
    In particular, $\phi_0 = f_0$.
    Since the associated Legendre polynomials have zero boundary data (Proposition~\ref{prop:boundary_data_legendre}), $\phi_l(\pm 1) = \beta^{(l)} f_0(\pm 1)$ and $\phi_l'(\pm 1) = \beta^{(l)} f_0'(\pm 1)$.
    Therefore $\phi_l$ fulfills condition (\textit{i}) of Proposition~\ref{prop:alpha-periodic-unitary}: $\phi_l \in \domain(H_{W(\alpha)})$ for all $l$.

    We now show that they also fulfill condition (\textit{ii}).
    To begin with, we show $\gamma_l^{(l)} \neq 0$ for $l \geq 1$.
    As $f_0(1) \neq 0$, $f_0$ is linearly independent from the (in itself linearly independent) set $(p_{k+m-1}^m)_{k = 1}^l$.
    Thus, $p_{l+m-1}^m$ is also linearly independent from the set $(\phi_k)_{k =0}^{l-1}$.
    Therefore, when adding the function $p_{l+m-1}^m$ to the orthogonal set $(\phi_k)_{k =0}^{l-1}$ using the Gram--Schmidt algorithm, the resulting coefficient $\gamma_l^{(l)}$ is non-zero:  $\gamma_l^{(l)} \neq 0$.
    Furthermore, using Eq.~\eqref{proofeq:thm-alphaperiodic-gs-1} and noticing that $\phi_l(1)=\beta^{(l)}f_0(1)$ since all other terms in the sum vanish there, $\tilde{\phi}_l \in \tilde{\Phi}$ is given by
    \begin{equation}
     \label{proofeq:thm-alphaperiodic-gs-tildephi}
    \begin{split}
        \tilde{\phi}_l(x) & = \phi_l(x) - \frac{\phi_l(1)}{f_0(1)} f_0(x) = \beta^{(l)} f_0(x) + \sum_{k = 1}^{l}\gamma_k^{(l)} p_{k+m-1}^m(x) - \frac{\beta^{(l)} f_0( 1)}{f_0(1)}f_0(x) \\
         & = \sum_{k = 1}^{l}\gamma_k^{(l)} p_{k+m-1}^m(x) \, .
    \end{split}
    \end{equation}
    Next, we show the following equality for $n \in \nnum$ and $n \geq 1$:
    \begin{equation}
       \operatorname{Span}(\tilde{\phi}_l)_{l = 1}^n =  \operatorname{Span}(p_l^m)_{l = m}^{n+m-1}
    \end{equation}
    Let $\psi \in \operatorname{Span}(\tilde{\phi}_l)_{l = 1}^n$.
    Then, there exist coefficients $(c_l)_{l = 1}^n$ such that
    \begin{equation}
        \psi = \sum_{l =1}^n c_l \tilde{\phi}_l = \sum_{l =1}^n \sum_{k = 1}^{l}\gamma_k^{(l)} p_{k+m-1}^m = \sum_{k = 1}^{n} d_k p_{k+m-1}^m\, ,
    \end{equation}
    where we used~\eqref{proofeq:thm-alphaperiodic-gs-tildephi} for the second equation.
    Comparing the coefficients, we get
    \begin{equation}
        d_k = \sum_{l = k}^n \gamma_k^{(l)} \quad \forall \, 1 \leq k \leq n\, , 
    \end{equation}
    which we can write, similarly to Eq.~\eqref{eq:proof-plm1perp-coeffs}, as a matrix equation $\mathsf{d} = \mathbb{A} \mathsf{c}$, where $\mathsf{d} = (d_1,d_3, \dots d_n)^\perp$, $\mathbf{c} = (c_1,c_2,\dots, c_n)^\perp$ and $\mathbb{A}$ is an $n\times n$ lower-triangular matrix whose diagonal elements  $\gamma^{(l)}_l$ are, as shown above, nonzero; as such, $\mathbb{A}$ is invertible.
    Thus, $\psi \in \operatorname{Span}(\tilde{\phi}_l)_{l = 1}^n$ if and only if  $\psi \in \operatorname{Span}(p_l^m)_{l = m}^{n+m-1}$.
    Taking $n$ to infinity and using $\tilde{\phi}_0 = f_0-f_0 = 0$, we get
    \begin{equation}
        \operatorname{Span}(\tilde{\phi}_l)_{l = 0}^\infty = \operatorname{Span}(p_l^m)_{l = m}^{\infty} \, .
    \end{equation}
    Using Proposition~\ref{prop:assoc_legendre_diff_1perp}, we finally obtain
    \begin{equation}
        \overline{\operatorname{Span}(\tilde{\phi}_l')_{l = 0}^\infty} = \overline{\operatorname{Span}((p_l^m)')_{l = m}^{\infty}} = \{1\}^\perp \, ,
    \end{equation}
    which is condition \textit{ii}) of Proposition~\ref{prop:alpha-periodic-unitary},  thus concluding the proof.
\end{proof}
        
\section{Concluding remarks}\label{sec:conclusion}
   
    In this work we have provided a detailed discussion on the convergence---or lack thereof---of the solutions of the Sch\"odinger equation generated by finite-dimensional truncations of a given quantum Hamiltonian $H$ with finite ground state energy on a Hilbert space $\hilbert$ via a family of projectors $(P_n)_{n\in\mathbb{N}}$ converging to the identity. Not surprisingly, the success or failure of this approximation method crucially depends on the specific choice of $P_n$: for example, if one discretizes with respect to a given orthonormal basis $(\phi_l)_{l\in\mathbb{N}}$, different choices of basis can lead to very different results. 
    
    Surprisingly enough, however, a relatively simple criterion for the outcome of the method can be found. To begin with, one restricts $H$ to the union of all finite-dimensional spaces $P_n\hilbert$ (which, as a direct consequence of infinite-dimensionality, is not simply equal to the whole $\hilbert$), thus obtaining a symmetric, not self-adjoint operator $H_{\rm fin}$. Then:
    \begin{enumerate}
        \item if $H_{\rm fin}$ to this space is still \textit{essentially self-adjoint} (i.e., it admits a single self-adjoint extension, necessarily coinciding with $H$), then the correct solution is obtained in the limit $n\to\infty$, and we do not need to worry;
        \item if $H_{\rm fin}$ has multiple self-adjoint extensions, with $H$ only being one of them, then the truncated solutions will reproduce, in the limit $n\to\infty$, the one corresponding to a specific self-adjoint extension of $H_{\rm fin}$ among all others: its Friedrichs extension.
    \end{enumerate}
    Let us rephrase this in a slightly more appealing way. Finite-dimensional truncations are guaranteed to work only if they ensure that ``not too much information about the physics'' of the system is lost, so that $H_{\rm fin}$ is still essentially self-adjoint. If, instead, we lose such information when truncating, physics cannot help us choosing the correct self-adjoint extensions---in this case, mathematics makes this choice for us, always selecting the Friedrichs extension among all others. Since $H_{\rm fin}$ can either admit one or infinitely many self-adjoint extensions, in the second case there are simply infinitely many ways we can take the wrong choice. And once again, numerics cannot generally help us in this case: what we will see is that the solutions, as $n\to\infty$, are converging to a legitimate, normalized wavefunction. Without mathematics, we just cannot trust this wavefunction to correspond to the physics we want to reproduce. Morally, not only must the basis be suited to the particular physical situation we aim to reproduce: they should also be unsuited to similar, but distinct ones.

    To put this machinery in motion in a case where we do have the exact, analytical solutions at hand, we went back to the essentials, and analyzed the particle in a box---a case in which different self-adjoint extensions of $H_{\rm fin}$ are conveniently mapped in different boundary conditions---proving that it is indeed possible to find choices of discretization bases such that things go spectacularly wrong: namely, one can as well find themselves in the position of the unlucky student in the ``thought homework'' presented at the start of this paper, reproducing in the limit $n\to\infty$ the \textit{right} solution for the particle in the \textit{wrong} box. 
    
    In this case, at the very least, we could come to the aid of our (imaginary) unfortunate student by just computing the exact analytical solution corresponding to the desired boundary conditions, and comparing it with the numerical one. For most complex quantum systems of research interest, e.g. in the context of quantum chemistry, obtaining analytical solutions is however out of question. In such cases, a wrong choice of basis can cause analogous issues that just go undetected: all we observe is that our computers yield a legitimate solution for each $n$, and that the solutions stay normalized as $n$ grows. 
    
    Our work motivates a more careful study of finite-dimensional truncations in the broader context, and stimulates a discussion between ``practitioners'' performing numerics and theoreticians applying functional analysis. In particular, it would be interesting to bridge between rather abstract domain issues on the one hand and practical convergence issues encountered in the numerics, as well as sufficiently general tricks to avoid them. This will, of course, require thinking outside the box.

    \section*{Acknowledgments}

 We thank S. Schraven for his insightful comments. DL acknowledges financial support by Friedrich-Alexander-Universit\"at Erlangen-N\"urnberg through the funding program ``Emerging Talent Initiative'' (ETI).
   
    \appendix

\section{Exact time evolution for Dirichlet and periodic boundary conditions}
\label{sec:app-exact-timeevolution}
In this appendix, we calculate the exact time evolution of the the Dirichlet Laplacian, $U_{\mathrm{Dir}}(t)$ and the periodic Laplacian $U_{\mathrm{per}}(t)$ at time $t = \frac{4}{\pi}$.
This is used in Example~\ref{ex:dirichlet_vs_neumann} to show that the two time evolutions can be maximally different at specific times, cf.~Figure~\ref{fig:dir_periodic_sin}. 

In the following, as in Section~\ref{sec:particle-in-a-box}, we set $\hilbert = L^2((-1,1))$, and we denote by $H_{\mathrm{Dir}}$ and $H_{\mathrm{per}}$ the Laplace operator with Dirichlet and periodic boundary conditions (Definition~\ref{def:particular_bcs}), with $U_{\rm Dir}(t)$, $U_{\rm per}(t)$ being the corresponding unitary propagators. Furthermore, given $\psi_0\in\hilbert$, we use the following notation:
\begin{equation}
 \psi_{\mathrm{Dir}}(x;t) = (U_{\mathrm{Dir}}(t)\psi_0)(x)\, ,\qquad  \psi_{\mathrm{per}}(x;t) = (U_{\mathrm{per}}(t)\psi_0)(x)\, .
\end{equation}

\begin{proposition}
\label{prop:app-dirichlet-exact}
    Let $\psi_0 \in \hilbert$ and $t = \frac{4}{\pi}$.
    Then
    \begin{equation}
       \psi_{\mathrm{Dir}}(x;4/ \pi) = -\psi_0(-x).
    \end{equation}
\end{proposition}
\begin{proof}
    The eigenfunctions $\psi_j$ and eigenvalues $E_j$ of the Dirichlet Laplacian are given by (cf.~Example~\ref{ex:bc-eigenfunctions}):
    \begin{equation}
        \label{eq:app-dir-eigenfunctions}
         E_j = \frac{j^2 \pi^2}{4}, \quad  \psi_j(x) = \sin\left(\frac{j\pi}{2}(x+1)\right), \quad  j \geq 1\, .
    \end{equation}
    For odd $j$ the wavefunctions are symmetric, and for even $j$ anti-symmetric:
    \begin{equation}
        \psi_{2j+1}(-x) = \psi_{2j+1}(x), \quad \psi_{2j}(-x) = -\psi_{2j}(x), \quad n \in \nnum, n \geq 1.
    \end{equation}
    As such, we can define projectors onto the symmetric ($P_{\rm s}$) and anti-symmetric subspace ($P_{\rm a}$) of $\hilbert$ by means of the following expressions:
    \begin{equation}
    \label{eq:app-dir-projectors}
        \begin{split}
            P_{\rm s} &: \hilbert \to \hilbert, \quad P_{\rm s} \psi = \sum_{j = 0}^\infty \braket{\psi_{2j+1},\psi} \psi_{2j+1}, \\
            P_{\rm a} &: \hilbert \to \hilbert, \quad P_{\rm a} \psi = \sum_{j = 1}^\infty \braket{\psi_{2j},\psi} \psi_{2j},
        \end{split}
    \end{equation}
and by construction
\begin{equation}\label{eq:proj}
        (P_{\rm s} \psi)(-x) = (P_{\rm s} \psi)(x), \quad (P_{\rm a} \psi)(-x) = -(P_{\rm a} \psi)(x), \quad P_{\rm s} \psi + P_{\rm a} \psi = \psi, \quad \forall \psi \in \hilbert.
    \end{equation}    
    Using the spectral theorem and Eq.~\eqref{eq:app-dir-eigenfunctions}, the time evolution of an initial state $\psi_0$ reads at time $t$
    \begin{equation}
            U_{\mathrm{Dir}}(t) \psi_0 = \sum_{j = 1}^\infty \e^{-\iu E_j t }\braket{\psi_j,\psi_0} \psi_j\, .
    \end{equation}
    Inserting $t = \frac{4}{\pi}$ and using $E_j = \frac{j^2 \pi^2}{4}$, we get
    \begin{equation}
        \begin{split}
             U_{\mathrm{Dir}}\left(\frac{4}{\pi}\right) \psi_0 & = \sum_{j = 1}^\infty \e^{-\iu \frac{j^2 \pi^2}{4}\frac{4}{\pi}} \braket{\psi_j,\psi_0} \psi_j =  \sum_{j = 1}^\infty (-1)^{j^2} \braket{\psi_j,\psi_0} \psi_j \\
             & = - \sum_{j = 0}^\infty \braket{\psi_{2j+1},\psi} \psi_{2j+1} + \sum_{j = 1}^\infty \braket{\psi_{2j},\psi} \psi_{2j} = -P_{\rm s} \psi_0 + P_{\rm a} \psi_0 \, ,
        \end{split}
    \end{equation}
    where in the last line we split the sum into even and odd parts and used the definition of the projectors in Eq.~\eqref{eq:app-dir-projectors}.
    Finally, by Eq.~\eqref{eq:proj},
    \begin{equation}
        U_{\mathrm{Dir}}\left(\frac{4}{\pi}\right) \psi_0 = -P_{\rm s} \psi_0 + P_{\rm a} \psi_0 = -(P_{\rm s} \psi_0)(-x) -(P_{\rm a} \psi_0)(-x) = -\psi_0(-x)\, ,
    \end{equation}
    thus completing the proof.
\end{proof}

\begin{proposition}
\label{prop:app-per-exact}
     Let $\psi_0 \in \hilbert$ and $t = \frac{4}{\pi}$.
    Then the \textit{periodic} time evolution of $\psi_0$ reads at $t = 4/\pi$
    \begin{equation}
       \psi_{\mathrm{per}}(x;4/ \pi) = \psi_0(x).
    \end{equation}
\end{proposition}
\begin{proof}
    The eigenfunctions $\psi_j$ and eigenvalues $E_j$ of the periodic Laplacian are given by (cf.~Example~\ref{ex:bc-eigenfunctions}):
    \begin{equation}
        \label{eq:app-periodic-eigenfunctions}
        E_j = j^2 \pi^2, \quad \psi_j(x) = \frac{1}{\sqrt{2}}\e^{-\iu j \pi (x+1)}, \quad j \in \znum\, .
    \end{equation}
    As they form a complete orthogonal set, summing over their respective projectors yields the identity:
    \begin{equation}
    \label{eq:app-periodic-id}
        \sum_{j \in \znum} \braket{\psi_j,\psi} \psi_j = \psi \quad \forall \psi \in \hilbert\, .
    \end{equation}
    Using the spectral theorem and Eq.~\eqref{eq:app-periodic-eigenfunctions}, the periodic time evolution of the initial state $\psi_0$ at time $t$ reads
    \begin{equation}
        U_{\mathrm{per}}(t) \psi_0 = \sum_{j \in \znum} \e^{-\iu E_j t} \braket{\psi_j,\psi_0}\psi_j \, .
    \end{equation}
    Inserting $t = \frac{4}{\pi}$, $E_j = j^2 \pi^2$ and using Eq.~\eqref{eq:app-periodic-id},
    \begin{equation}
        U_{\mathrm{per}}\left(\frac{4}{\pi}\right) \psi_0 = \sum_{j \in \znum} \e^{-\iu j^2 \pi^2 \frac{4}{\pi}} \braket{\psi_j,\psi_0}\psi_j = \sum_{j \in \znum}\braket{\psi_j,\psi_0}\psi_j = \psi_0 \, ,
    \end{equation}
    thus completing the proof.
\end{proof}
\begin{corollary}
\label{corr:app-dirichlet-periodic-maxdifferent}
    For $t = \frac{4}{\pi}$ the Dirichlet and periodic time evolutions are maximally different, in the following sense: the operator norm of their difference reads
    \begin{equation}
        \norm*{U_{\mathrm{Dir}}\left(\frac{4}{\pi}\right) -U_{\mathrm{per}}\left(\frac{4}{\pi}\right) } = 2 \, .
    \end{equation}
\end{corollary}
\begin{proof}
Clearly the difference between two unitaries is always bounded by $2$. To show that the equality holds, just consider any even wavefunction $\psi_0$, i.e., $\psi_0(-x)=\psi_0(x)$. Then, by Propositions~\ref{prop:app-dirichlet-exact} and~\ref{prop:app-per-exact},
     \begin{equation}
        \norm*{U_{\mathrm{Dir}}\left(\frac{4}{\pi}\right)\psi_0 -U_{\mathrm{per}}\left(\frac{4}{\pi}\right)\psi_0 } = 2\|\psi_0\| \,,
    \end{equation}
    which proves the claim.
\end{proof}
Finally, we remark that the exact time evolution generated by the Dirichlet Laplacian for arbitrary rational times $t = \frac{p}{q \pi}$, $p,q \in \nnum$ (in units of $\frac{1}{\pi}$) was famously discussed by Berry in~\cite{berry-quantumfractalsboxes-1996}. Similar techniques can be applied to other boundary conditions.

\section{Expansion coefficients in the associated Legendre polynomials}
\label{app:expansion-coefficients}

In this appendix we show how we obtained the numerical data in Figure~\ref{fig:numerics-nt-alt}, and how similar plots can be obtained.
To this extent, let us begin by finding a computable expression for the state-dependent approximation error $\norm*{U_n(t) \psi_{0,W} - U_{W}(t) \psi_{0,W}}$, with $\psi_{0,W}$ being an eigenvector of the corresponding realization $H_W$ of the Laplace operator,
\begin{equation}
    H_W\psi_{0,W}=E_{0,W}\psi_{0,W}.
\end{equation}
In this case, given any orthonormal basis $(\phi_l)_{l\in\mathbb{N}}$, one easily gets
\begin{align}
    \label{eq:numerics-norm-expansion}
     \norm*{U_W(t)\psi_{0,W} -U_n(t) \psi_{0,W}}^2 =&  \sum_{l=0}^\infty \left|\Braket{\phi_l,U_W(t)\psi_{0,W} - U_n(t) \psi_{0,W}} \right|^2 \\
    =&    \sum_{l=0}^{n} \left|\exp(-\iu E_{0,W} t)\Braket{\phi_l,\psi_{0,W}} -  \Braket{\phi_l, U_n(t)\psi_{0,W}} \right|^2 \nonumber\\&        +  \sum_{l = n+1}^\infty \left|\exp(-\iu E_{0,W} t)-1\right|^2 \left| \Braket{\phi_l,\psi_{0,W}} \right|^2
        \label{proofeq:numerics-norm-1} \\
    =&   \sum_{l=0}^{n} \left|\exp(-\iu E_{0,W} t)\Braket{\phi_l,\psi_{0,W}} -  \Braket{\phi_l, U_n(t)\psi_{0,W}} \right|^2\nonumber\\& + \left|\exp(-\iu E_{0,W} t)-1\right|^2 \left(1-\sum_{l=0}^{n} \left|\Braket{\phi_l,\psi_{0,W}} \right|^2 \right) \,,
        \label{proofeq:numerics-norm-2}
\end{align}
where we used the property $\e^{-\iu tH_W}\psi_{0,W}=\e^{-\iu E_{0,W}t}\psi_{0,W}$, as well as the fact that $H_n$ by construction acts trivially on the vectors $\phi_l$ with $l>n$.
No infinite sum is thus required to compute the error in this case.

In the following we choose the normalized associated Legendre polynomials $(p^m_l)_{l\geq m}$ as orthonormal basis, and we consider the two cases of Dirichlet and periodic boundary conditions, for which all eigenvalues and eigenvectors are known. In these cases, to compute the quantity in Eq.~\eqref{proofeq:numerics-norm-2} we need to compute the following quantities:
\begin{itemize}
    \item the expansion coefficients $\braket{\phi_l,\psi_{0,W}}$ of the eigenvector $\psi_{0,W}$ in the basis $(\phi_l)_{l \in \nnum}$;
    \item the quantities $\braket{\phi_l,U_n(t)\psi_{0,W}}$.
\end{itemize}
As we will see, the expansion coefficients of both Dirichlet and periodic boundary conditions can be evaluated exactly; this is done in Sections~\ref{sec:app-coeff-dir}--\ref{sec:app-coeff-per}. To compute $\braket{\phi_l,U_n(t)\psi_{0,W}}$, we adopt the following strategy:
\begin{enumerate}
    \item We evaluate exactly the matrix elements $\braket{p_l^m, H_W p_k^m}$ of the Laplace operator, which gives us the exact expression of $H_n$ for every fixed $n$. This is done in Section~\ref{sec:app-matrixels}.
    \item We then determine the exponential $U_n(t) = \e^{-\iu tH_n}$ numerically via direct matrix exponentiation.
\end{enumerate}
For the latter purpose, we use the python library \verb|mpmath|~\cite{mpmath}, which provides arbitrary precision floating point arithmetic. This allows us to avoid overflow issues when dealing with large factorials and determine $U_n(t)$ with high precision. In particular, we used a precision of $500$ decimal places; at this point further increasing the precision did no change the results. 

Before going through the calculations, we state some auxiliary lemmata on the derivatives of the associated Legendre polynomials that will be needed in the remainder of the appendix.
\begin{lemma}
    \label{lem:legendre_derivative}
    Let $\alpha,l$ be integers such that $\alpha\leq l$. Then
    \begin{equation}
        \diff*[\alpha]{P_l(x)}{x}[x=1] = \frac{2^{-\alpha}}{\alpha!} \frac{(l+\alpha)!}{(l-\alpha)!} \, .
    \end{equation}
    For $\alpha > l$, $\diff*[\alpha]{P_l(x)}{x} = 0$.
\end{lemma}

\begin{proof}
We will use the following representation of Legendre polynomials~\cite{dlmf147-legendre-rodrigues}:
\begin{equation}
    P_l(x) = \sum_{k=0}^l \binom{l}{k}\binom{l+k}{k} \left(\frac{x-1}{2}\right)^k \, .
\end{equation}
In this representation, the derivative at $x=1$ reads
\begin{equation}
    \diff*[\alpha]{P_l(x)}{x}[x=1] =\sum_{k=0}^l \binom{l}{k}\binom{l+k}{k} \diff*[\alpha]{\left(\frac{x-1}{2}\right)}{x}[ x=1] =  2^{-\alpha} \alpha! \, \binom{l}{\alpha}\binom{l+\alpha}{\alpha} \, ,
\end{equation}
whence, for $l \geq \alpha$,
\begin{equation}
    \diff*[\alpha]{P_l(x)}{x}[x=1] = 2^{-\alpha} \alpha! \, \frac{l! \, (l+\alpha)!}{\alpha! \, (l-\alpha)! \, \alpha! \, (l+\alpha-\alpha)! } = \frac{2^{-\alpha}}{\alpha!} \frac{(l+\alpha)!}{(l-\alpha)!} \, .
\end{equation}
\end{proof}
\begin{lemma}
    \label{lem:assoc_legendre_derivative}
    Let $m,\alpha,l$ integers such that $m$ is even and $\alpha\leq l$. Then
    \begin{multline}
         \diff*[\alpha]{P_l^m(x)}{x}[x=1] = (-1)^{m/2} 2^{-\alpha} \alpha!  \left(\frac{m}{2}\right)! \\
         \times \,  \sum_{t=m/2}^{\min(m,\alpha)} \frac{(l+m+\alpha-t)!}{(\alpha-t)! (t-m/2)! (m-t)! (m-t+\alpha)! (l-m-\alpha +t)!} \,.
    \end{multline}
    For $\alpha > l$, $\diff*[\alpha]{P_l^m(x)}{x} = 0$.
\end{lemma}
\begin{proof}
Using the product rule, we get
\begin{multline}
        \diff*[\alpha]{P_l^m(x)}{x}[x = 1]  = (-1)^m \diff*[\alpha]{\left( (1-x^2)^{m/2} \diff*[m]{P_l(x)}x \right)}{x}[x=1] \\
         = \sum_{t=0}^\alpha \binom{\alpha}{t} \diff*[t]{(1-x^2)^{m/2}}{x}[x = 1] \diff*[m+\alpha-t]{P_l(x)}{x}[x = 1]
\end{multline}
The first part of each summand is $0$ for $t < \frac{m}{2}$, as one term $(1-x^2)$ that evaluates to $0$ at $x=1$ always remains.
Therefore, we take $t \geq \frac{m}{2}$.
Furthermore, for $t>m$, $\diff*[t]{(1-x^2)^{m/2}}{x} = 0$.
If $\alpha < m$, then $\alpha$ is the upper limit of the sum.
Then, using Eq.~\eqref{eq:def-legendre},
\begin{equation}
    \diff*[t]{(1-x^2)^{m/2}}{x} = (-1)^{m/2} 2^{m/2}(m/2)! \, \diff*[t-m/2]{P_{m/2}(x)}{x},
\end{equation}
and by Lemma~\ref{lem:legendre_derivative} we get
\begin{equation}
\begin{split}
            & \diff*[\alpha]{P_l^m(x)}{x}[x = 1]         = (-2)^{m/2}(m/2)! \,  \sum_{t=m/2}^{\min(m,\alpha)} \binom{\alpha}{t} \diff*[t-m/2]{P_{m/2(x)}}{x}[x = 1] \diff*[m+\alpha-t]{P_l(x)}{x}[x = 1] \\
         & = (-2)^{m/2}(m/2)! \,  \sum_{t=m/2}^{\min(m,\alpha)} \frac{\alpha!}{t! \, (\alpha-t)!} \frac{2^{-t+m/2} t!}{(t-m/2)! (m-t)!} \frac{2^{-m+t-\alpha} (l+m+\alpha-t)!}{(m-t+\alpha)!(l-m-\alpha+t)!} \\
         & = (-1)^{m/2} 2^{m/2+m/2-m-t+t-\alpha} \alpha!(m/2)! \\
         & \qquad \times \sum_{t=m/2}^{\min(m,\alpha)}  \frac{t!}{t!} \frac{(l+m+\alpha-t)!}{(\alpha-t)! (t-m/2)! (m-t)! (m-t+\alpha)! (l-m-\alpha +t)!} \, ,
\end{split}
\end{equation}
which concludes the proof.
\end{proof}

\subsection{Matrix elements of the Laplacian}
\label{sec:app-matrixels}
By using Lemmata~\ref{lem:legendre_derivative}--\ref{lem:assoc_legendre_derivative}, we will now be able to compute the matrix elements of any realization $H_W$ of the Laplace operator on $(-1,1)$. Here, to avoid cumbersome expressions, we will compute them without the normalization factors.
\begin{proposition}
    \label{prop:matrix-elements}
    Let $m\geq4$ even, $W\in\mathrm{U}(2)$, and $H_W$ being the corresponding realization of the Laplace operator on $L^2((-1,1))$. Then, for every integers $l,k\geq m$,
    \begin{equation}\label{eq:matrix_general}
        \begin{split}
            \Braket{P_l^m , H_W  P_k^m} &= 2 \delta_{\mmod(l+k,2),0} \times \begin{cases}
                h_{lk} & l \geq k \\ h_{kl} & l < k
            \end{cases}\, , \\
            h_{lk} & = \sum_{t=\left\lfloor \frac{m}{2} \right\rfloor + 1}^m (-1)^t   \diff*[m-t]{P_l(x)}{x}[x = 1] \times \\
            & \quad \quad \times \diff*[t-1]{\left( (1-x^2)^{m/2} \diff*[2]{\left( (1-x^2)^{m/2} \diff*[m]{P_k(x)}{x} \right) }{x} \right)}{x}[{x = 1}].
        \end{split}
    \end{equation}
    In particular, for $m=4$,
  \begin{equation}\label{eq:matrix_4}
       h_{lk} = \left(2-\frac{l(l+1)}{12}\right)\frac{(k+4)!}{(k-4)!} + \frac{18}{5!} \frac{(k+5)!}{(k-5)!}.     \end{equation}
\end{proposition}
As anticipated in the main text, the matrix elements do not depend on the particular choice of boundary conditions.
As $H_W$ is self adjoint, we assume $l \geq k$ without loss of generality.
Furthermore, with a slight abuse of notation, in the following we will drop the explicit dependence on the operator on $W$.
In addition, we recall the following property of the associated Legendre polynomials:
\begin{equation}\label{eq:t>m}
    \diff[m-t]{P_l(x)}{x} = \frac{1}{2^l l!} \diff[l+m-t]{(x^2-1)^l}{x},\qquad  m < t \leq l+m.
\end{equation}
We prove an auxiliary lemma before proving the proposition.
\begin{lemma}
    Let $l,k,m$ be integers with $l \geq k \geq m\geq4$ and $m$ being even.
    Then, for every integer $T$ such that $1 \leq T \leq l+m$, the following equality holds:
\begin{multline}
    \label{proofeq:matrixels-recurrence}
    \Braket{P_l^m,H P_k^m} =   \!\sum_{t=1}^T (-1)^t\!  \left.\left( \diff[m-t]{P_l(x)}{x}\right) \diff*[t-1]{(1-x^2)^{m/2} \diff*[2]{ \left( (1-x^2)^{m/2} \diff[m]{P_k(x)}{x} \right)}{x}}{x} \right|_{-1}^1 \\
     + (-1)^{T+1} \int_{-1}^{1} \left( \diff[m-T]{P_l(x)}{x} \right) \diff*[T]{ \left( (1-x^2)^{m/2} \diff*[2]{ \left( (1-x^2)^{m/2} \diff[m]{P_k(x)}{x} \right)}{x}\right)}{x}\dl x
\end{multline}
\end{lemma}
\begin{proof}
To begin with, we note that all expressions $\diff[m-t]{P_l(x)}{x}$ are well-defined by Eq.~\eqref{eq:t>m}, as we assumed $t \leq T \leq l+m$.

The claim will be obtained by performing $T$ integrations by parts. To this extent, we prove the lemma by induction over $T$, starting with $T = 1$:
\begin{equation}
    \begin{split}   
    \braket{P_l^m ,H P_k^m} &= - \int_{-1}^{1} \left( (-1)^m (1-x^2)^{m/2} \diff[m]{P_l(x)}{x}\right) \diff*[2]{ \left( (-1)^m (1-x^2)^{m/2} \diff[m]{P_k(x)}{x} \right)}{x} \dl x \\
    & \begin{multlined}[t]
        = - \left. \left( \diff[m-1]{P_l(x)}{x}\right) (1-x^2)^{m/2} \diff*[2]{\left( (1-x^2)^{m/2} \diff[m]{P_k(x)}{x}\right) }{x} \right|_{-1}^1 \\
         + \int_{-1}^{1} \left( \diff[m-1]{P_l(x)}{x} \right) \diff*[1]{\left( (1-x^2)^{m/2} \diff*[2]{  (1-x^2)^{m/2} \diff[m]{P_k(x)}{x} }{x}\right) }{x} \dl x \, ,
    \end{multlined}
\end{split}
\end{equation}
where we performed a single integration by parts. 

Now assume Eq.~\eqref{proofeq:matrixels-recurrence} to hold for some integer $T$ with $1\leq T \leq l+m-1$. Then it also holds for $T+1$:
\begin{multline}
        \Braket{P_l^m , H P_k^m} =\! \sum_{t=1}^T (-1)^t\!\left. \left( \diff[m-t]{P_l(x)}{x} \right) \diff*[t-1]{\left( (1-x^2)^{m/2} \diff*[2]{ (1-x^2)^{m/2} \diff[m]{P_k(x)}{x} }{x} \right)}{x} \right|_{-1}^1 \\
        + (-1)^{T+1} \int_{-1}^{1} \left( \diff[m-T]{P_l(x)}{x} \right) \diff*[T]{\left( (1-x^2)^{m/2} \diff*[2]{ \left((1-x^2)^{m/2} \diff[m]{P_k(x)}{x} \right)}{x}\right) }{x} \dl x\\
        = \sum_{t=1}^T (-1)^t \left. \left( \diff[m-t]{P_l(x)}{x} \right) \diff*[t-1]{\left( (1-x^2)^{m/2} \diff*[2]{\left( (1-x^2)^{m/2} \diff[m]{P_k(x)}{x} \right)}{x}\right) }{x} \right|_{-1}^1 \\
         + (-1)^{T+1} \left. \left( \diff[m-T-1]{P_l(x)}{x} \right) \diff*[T]{ \left((1-x^2)^{m/2} \diff*[2]{\left( (1-x^2)^{m/2} \diff[m]{P_k(x)}{x} \right)}{x} \right)}{x} \right|_{-1}^1 \\
         - (-1)^{T+1} \int_{-1}^{1}  \left( \diff[m-T-1]{P_l(x)}{x}\right) \diff*[T+1]{\left(  (1-x^2)^{m/2} \diff*[2]{\left( (1-x^2)^{m/2} \diff[m]{P_k(x)}{x}\right) }{x}\right) }{x} .
\end{multline}
Rearranging the terms, the claim is verified, closing the induction.
\end{proof}

\begin{proof}[Proof of Proposition~\ref{prop:matrix-elements}]
After $T \leq l+m$ integrations by parts, the matrix element $\braket{P_l^m,HP_k^m}$ is given by Eq.~\eqref{proofeq:matrixels-recurrence}.
Setting $T = k+m-1$, the remaining integral vanishes and only the sum remains as we will show in the following.
Recalling the definition of the Legendre polynomials (see Definition~\ref{def:legendre}), $P_k$ is a polynomial of degree $k$, and assuming $m$ to be even,
\begin{equation}
 (1-x^2)^{m/2} \diff*[2]{  (1-x^2)^{m/2} \diff*[m]{P_k(x)}{x} }{x}
\end{equation}
is a polynomial of degree $k-2+m$.
Thus, setting $T = k+m-1$, the integrand in Eq.~\eqref{proofeq:matrixels-recurrence} evaluates to 0:
\begin{equation}
    \left( \diff[m-T]{P_l(x)}{x} \right) \diff*[T]{ \left( (1-x^2)^{m/2} \diff*[2]{ \left( (1-x^2)^{m/2} \diff[m]{P_k(x)}{x} \right)}{x}\right)}{x} = 0\, ,
\end{equation}
as the second part of the product evaluates to 0.
The now simplified matrix elements read
\begin{equation}
\label{eq:matrix-els-simp1}
    \Braket{P_l^m,H P_k^m} =\!\!\sum_{t=1}^{k+m-1} (-1)^t\! \left.\left( \diff[m-t]{P_l(x)}{x}\right)\!\diff*[t-1]{(1-x^2)^{m/2} \diff*[2]{ \!\left( (1-x^2)^{m/2} \diff[m]{P_k(x)}{x} \right)}{x}}{x} \right|_{-1}^1 \\
\end{equation}
We now show that we can further restrict the sum to $t \in \left[\left\lfloor \frac{m}{2}\right\rfloor+1,m\right]$.
The following term (cf.~Eq.~\eqref{eq:t>m})
\begin{equation}
    \diff*[m-t]{P_l(x)}{x}\bigg|_{x = \pm 1} = \frac{1}{2^l l!} \diff*[l+m-t]{(x^2-1)^l}{x}\bigg|_{x=\pm 1}
\end{equation}
evaluates to 0 at the boundaries for $m < t \leq l+m$.
Therefore, the sum can be cut off at $m$.
Furthermore, for $t-1 < \frac{m}{2}$, the term $\diff*[t-1]{((1-x^2)^{m/2} g(x))}{x}$ (for an arbitrary polynomial $g$) evaluates to zero at the boundaries.
In particular, 
\begin{equation}
   \left. \diff*[t-1]{\left((1-x^2)^{m/2} \diff*[2]{ \left( (1-x^2)^{m/2} \diff[m]{P_k(x)}{x} \right)}{x}\right)}{x} \right|_{\pm 1} = 0
\end{equation}
for $t-1 < \frac{m}{2}$, and we can simplify Eq.~\eqref{eq:matrix-els-simp1} further to
\begin{equation}
\label{eq:matrix-els-simp2}
\begin{split}
        \Braket{P_l^m,H P_k^m} & = \sum_{t=\left\lfloor \frac{m}{2}\right\rfloor+1}^{m} (-1)^t  \left.F_{m,l,k,t}(x)\right|_{-1}^1 \\
        F_{m,l,k,t}(x) & = \left( \diff[m-t]{P_l(x)}{x}\right) \diff*[t-1]{(1-x^2)^{m/2} \diff*[2]{ \left( (1-x^2)^{m/2} \diff[m]{P_k(x)}{x} \right)}{x}}{x} \, ,
\end{split}
\end{equation}
where we newly introduced the polynomial $F_{m,l,k,t}(x)$ for $\left\lfloor \frac{m}{2}\right\rfloor+1 \leq t \leq m$.
We now show that (for even $m$)
\begin{equation}
    F_{m,l,k,t}(-x) = (-1)^{k+l+1}F_{m,l,k,t}(x)\,. 
\end{equation}
Indeed, for the Legendre polynomials the following parity relation holds: $P_l(-x) = (-1)^l P_l(x)$.
Furthermore, $(1-x^2)^{m/2}$ is a symmetric function, and the derivative of a symmetric function is antisymmetric and vice versa.
Thus, 
\begin{equation}
    F_{m,l,k,t}(-x) = (-1)^{m-t+l+t-1+2+m+k}F_{m,l,k,t}(x) = (-1)^{l+k+1}F_{m,l,k,t}(x)\,
\end{equation}
proving the claim.
Therefore, the evaluation at the boundaries reads 
\begin{multline}
    \left.F_{m,l,k,t}(x)\right|_{-1}^1 = F_{m,l,k,t}(1)- F_{m,l,k,t}(-1) \\ = (1-(-1)^{l+k+1})F_{m,l,k,t}(1) = 2\delta_{\mmod(k+l,2),0} F_{m,l,k,t}(1)\, .
\end{multline}
Inserting this into Eq.~\eqref{eq:matrix-els-simp2} proves Eq.~\eqref{eq:matrix_general}.

To conclude, in the case $m=4$ ($m/2=2$), we obtain a sum with two summands from $t=3$ to $t=4$.
We use Eq.~\eqref{eq:def-legendre} and Lemma~\ref{lem:legendre_derivative} to calculate the terms explicitly.
For $t=3$, the second derivative needs to be fully applied to $\diff*[t-1]{(1-x^2)^{2}}{x}[1]$ to not evaluate to 0, yielding
\begin{multline}
    \diff*{P_l(x)}{x}[x =1]  = \diff*[2]{\left((1-x^2)^2  \diff*[2]{\left((1-x^2)^2 \diff*[4]{P_k(x)}{x}\right)}{x}\right)}{x}[x = 1] \\
    = \frac{1}{2} \frac{(l+1)!}{(l-1)!}  \diff*[2]{(1-x^2)^2}{x}[1]\diff*[2]{(1-x^2)^2}{x}[x = 1] \diff*[4]{P_k(x)}{x}[x = 1] = \frac{l(l+1)}{12} \frac{(k+4)!}{(k-4)!} \, .
\end{multline}
For $t=4$, the product rule is used, yielding
\begin{align}
    P_l(1) & \diff*[3]{\left((1-x^2)^2  \diff*[2]{\left((1-x^2)^2 \diff*[4]{P_k(x)}{x}\right)}{x}\right)}{x}[x = 1] \\
    & \begin{multlined}
        = \binom{3}{3} \diff*[3]{(1-x^2)^2}{x}[x = 1] \diff*[2]{(1-x^2)^2}{x}[x = 1] \diff*[4]{P_k(x)}{x}[x = 1]\\
         +  \binom{3}{2} \diff*[2]{(1-x^2)^2}{x}[x = 1] \binom{3}{3}\diff*[3]{(1-x^2)^2}{x}[x = 1] \diff*[4]{P_k(x)}{x}[x = 1] \\
         + \binom{3}{2} \diff*[2]{(1-x^2)^2}{x}[x = 1] \binom{3}{2}\diff*[2]{(1-x^2)^2}{x}[x = 1] \diff*[4+1]{P_k(x)}{x}[x = 1] 
    \end{multlined} \\
    & = 2\frac{(k+4)!}{(k-4)!} + \frac{18}{5!} \frac{(k+5)!}{(k-5)!} \, ,
\end{align}
which proves Eq.~\eqref{eq:matrix_4}, thus completing the proof.\end{proof}

\subsection{Expansion coefficients of the Dirichlet eigenfunctions}
\label{sec:app-coeff-dir}
In this section, we calculate the expansion coefficients of the Dirichlet eigenfunctions $\psi_j, j \in \nnum$, $j \geq 1$.
Recall that they are given by (see Example~\ref{ex:bc-eigenfunctions})
\begin{equation}
\label{eq:app-dirichlet-eigenfunction}
    \psi_j(x) = \sin\left(\frac{j \pi}{2}(x+1)\right)\, , \quad j \geq 1 \, , \\
\end{equation}
Here, we state the non-normalized expansion coefficients.
\begin{proposition}
    \label{prop:coeffs-sin-exp}
    Let $j,m,l$ be integers such that $j\geq1$, $l\geq m\geq4$, and $m$ is even, and $\psi_j$ as defined in Eq.~\eqref{eq:app-dir-eigenfunctions}.
    Then
    \begin{equation}
        \Braket{P_l^m,\psi_j} = \frac{-2}{j \pi} \left( (-1)^j - (-1)^l \right) \sum_{t=t_{\text{min}}}^{t_{\text{max}}} (-1)^t \left( \frac{2}{j \pi} \right)^{2 t} \left. \diff*[2t]{P_l^m(x)}{x} \right|_{x = 1}\, ,
    \end{equation}
    where $t_{\text{min}} = \left\lfloor \frac{m}{4} - \frac{1}{2} \right\rfloor + 1$ and $t_{\text{max}} = \left\lfloor \frac{l}{2} \right\rfloor$.
\end{proposition}
We prove the claim using integration by parts.
\begin{lemma}
    \label{lem:app-sin-exp-part}
   Let $j,m,l$ be integers such that $j\geq1$, $l\geq m\geq4$, and $m$ is even, and $\psi_j$ as defined in Eq.~\eqref{eq:app-dir-eigenfunctions}.
    Then, for every integer $0 \leq T \leq \frac{l}{2}$, the following equality holds:
    \begin{multline}
    \label{proofeq:coeff-sin-1}
    \braket{ P_l^m, \psi_j} = \frac{-2}{j \pi} \sum_{t=0}^T (-1)^t \left( \frac{2}{j \pi} \right)^{2t} \cos\left( \frac{j}{2} \pi (x+1) \right) \diff*[2t]{P_l^m(x)}{x} \bigg|_{-1}^{1} \\
    -  (-1)^T \left( \frac{2}{j \pi} \right)^{2T+2}  \int_{-1}^1 \sin\left( \frac{j}{2} \pi (x+1) \right) \diff*[2T+2]{P_l^m(x)}{x} \dl x \,  .
\end{multline}
\end{lemma}
\begin{proof}
This will be proven by performing $2T+2$ integrations by parts. After two integrations by parts, the expansion coefficients read
    \begin{equation}
    \begin{split}
        & \!\! \braket{ P_l^m, \psi_j} = \int_{-1}^1 \sin\left( \frac{j}{2} \pi (x+1) \right) P_l^m(x) \dl x \\
        &= \frac{-2}{j \pi} \cos\left( \frac{j}{2} \pi (x+1) \right) P_l^m(x) \bigg|_{-1}^{1} + \frac{2}{j \pi} \int_{-1}^1 \cos\left( \frac{j}{2} \pi (x+1) \right) \diff{P_l^m(x)}{x} \dl x \\
        &= \frac{-2}{j \pi} \cos\left( \frac{j}{2} \pi (x+1) \right) (1-x^2)^{m/2} P_l^m(x) \bigg|_{-1}^{1} + \left( \frac{2}{j \pi} \right)^2 \sin\left( \frac{j}{2} \pi (x+1) \right) \diff*{P_l^m(x)}{x} \bigg|_{-1}^{1} \\
        & \quad - \left( \frac{2}{j \pi} \right)^2 \int_{-1}^1 \sin\left( \frac{j}{2} \pi (x+1) \right) \diff*[2]{P_l^m(x)}{x} \dl x \, .
    \end{split}
\end{equation}
As $\sin\left( \frac{j}{2} \pi (\pm 1 + 1) \right) = 0$, the second boundary term vanishes, proving the claim for $T=0$.
Assuming that Eq.~\eqref{proofeq:coeff-sin-1} holds for $0 \leq T \leq l/2-1$, after two further integrations by parts we get
\begin{align}
    \braket{ P_l^m, \psi_j} &= \frac{-2}{j \pi} \sum_{t=0}^T (-1)^t \left( \frac{2}{j \pi} \right)^{2t} \cos\left( \frac{j}{2} \pi (x+1) \right) \diff*[2t]{P_l^m(x)}{x}  \bigg|_{-1}^{1} \\
    &\quad -(-1)^T \left( \frac{2}{j \pi} \right)^{2T+2}  \int_{-1}^1 \sin\left( \frac{j}{2} \pi (x+1) \right) \diff*[2T+2]{P_l^m(x)}{x} \dl x \\
    & = \label{proofeq:coeff-sin-3} \frac{-2}{j \pi} \sum_{t=0}^T (-1)^t \left( \frac{2}{j \pi} \right)^{2t} \cos\left( \frac{j}{2} \pi (x+1) \right) \diff*[2t]{P_l^m(x)}{x} \bigg|_{-1}^{1} \\
    & \label{proofeq:coeff-sin-4} \quad - (-1)^T \left( \frac{2}{j \pi} \right)^{2T+2}  \frac{-2}{j \pi} \cos\left( \frac{j}{2} \pi (x+1) \right) \diff*[2T+2]{P_l^m(x)}{x} \bigg|_{-1}^{1} \\
    &\quad + (-1)^T \left( \frac{2}{j \pi} \right)^{2T+2}  \frac{-2}{j \pi} \int_{-1}^1 \cos\left( \frac{j}{2} \pi (x+1) \right) \diff*[2T+3]{P_l^m(x)}{x} \dl x \\
    &=  \frac{-2}{j \pi} \sum_{t=0}^{T+1} (-1)^t \left( \frac{2}{j \pi} \right)^{2t} \cos\left( \frac{j}{2} \pi (x+1) \right) \diff*[2t]{P_l^m(x)}{x} \bigg|_{-1}^{1} \\
    & \label{proofeq:coeff-sin-5}\quad -(-1)^T \left( \frac{2}{j \pi} \right)^{2T+2}  \left(\frac{2}{j \pi}\right)^2 \sin\left( \frac{j}{2} \pi (x+1) \right) \diff*[2T+3]{P_l^m(x)}{x} \bigg|_{-1}^{1} \\
    &\quad + (-1)^T \left( \frac{2}{j \pi} \right)^{2T+2}   \left(\frac{2}{j \pi}\right)^2 \int_{-1}^1 \sin\left( \frac{j}{2} \pi (x+1) \right) \diff*[2T+4]{P_l^m(x)}{x} \dl x \, .
\end{align}
Plugging Eq.~\eqref{proofeq:coeff-sin-4} into the sum~\eqref{proofeq:coeff-sin-3}, and using $\sin\left( \frac{j}{2} \pi (\pm 1 + 1) \right) = 0$ in Eq.~\eqref{proofeq:coeff-sin-5}, we prove the assumption for the case $T+1$.
\end{proof}

\begin{proof}[Proof of Proposition~\ref{prop:coeffs-sin-exp}]
We show that the integral in Eq.~\eqref{proofeq:coeff-sin-1} vanishes for $T = t_{\mathrm{max}}= \left\lfloor \frac{l}{2} \right\rfloor$.
Then, 
\begin{equation}
    2T+2 = 2\left\lfloor \frac{l}{2} \right\rfloor+2 \geq l+1 > l \, ,
\end{equation}
and as the associated Legendre polynomials are polynomials of degree $l$, the derivative $\diff*[2T+2]{P_l^m(x)}{x} = 0$.
Hence, after $T = t_{\mathrm{max}}$ integrations by part, we get (cf.~Lemma~\ref{lem:app-sin-exp-part})
\begin{equation}
        \label{proofeq:coeff-sin-new}
        \braket{ P_l^m, \psi_j} = \frac{-2}{j \pi} \times \sum_{t=0}^{t_{\mathrm{max}}} (-1)^t \left( \frac{2}{j \pi} \right)^{2t} \cos\left( \frac{j}{2} \pi (x+1) \right) \diff*[2t]{P_l^m(x)}{x} \bigg|_{-1}^{1} \\
\end{equation}
On the other hand, for $2t \leq \frac{m}{2}-1$, $\diff*[2t]{(1-x^2)^{m/2}}{x}$ evaluates to 0 at the boundaries, such that the summands vanish for
\begin{equation}
    \label{eq:coeff-min}
    t < t_{\text{min}} = \left\lfloor \frac{m}{4} - \frac{1}{2} \right\rfloor + 1.
\end{equation}
Plugging the boundaries for $t$ into Eq.~\eqref{proofeq:coeff-sin-new}, we get
\begin{equation}
    \Braket{P_l^m, \psi_j} = \frac{-2}{j \pi} \times \sum_{t=t_{\text{min}}}^{t_{\text{max}}} (-1)^t \left( \frac{2}{j \pi} \right)^{2 t} \left( \cos\left( \frac{j \pi}{2} (x+1) \right) \diff*[2t]{P_l^m(x)}{x} \right) \bigg|_{-1}^{1}\, .
\end{equation}
Furthermore, by using the following parity properties:
\begin{align}
    \cos\left( \frac{j \pi}{2} (1+1) \right) = (-1)^j\, , \quad \cos\left( \frac{j \pi}{2} (-1+1) \right) = +1\, , \\ 
    \diff*[2\alpha]{P_l^m(x)}{x}\big|_{x = -1} = (-1)^l \diff*[2\alpha]{P_l^m(x)}{x}\big|_{x = 1} \quad \forall \alpha \in \nnum \, .
\end{align}
we finally get
\begin{equation}
    \Braket{P_l^m, \psi_j} = \frac{-2}{j \pi} \left( (-1)^j - (-1)^l \right) \times \sum_{t=t_{\text{min}}}^{t_{\text{max}}} (-1)^t \left( \frac{2}{j \pi} \right)^{2 t} \diff*[2t]{P_l^m(x)}{x}[x = 1] \, .
\end{equation}
\end{proof}
\subsection{Expansion coefficients of the periodic eigenfunctions}
\label{sec:app-coeff-per}
To conclude, we calculate the expansion coefficients of the (non-normalized) periodic eigenfunctions $\psi_j, j \in \nnum$, $j \in \znum$, which are given by (cf.~Example~\ref{ex:bc-eigenfunctions})
\begin{equation}
\label{eq:app-per-eigenfunction}
    \psi_j(x)  = \e^{-\iu j \pi (x+1)}, \quad j \in \znum\, . \\
\end{equation}
\begin{proposition}
    \label{prop:coeffs-per-exp}
    Let $j \in \znum$, and $m,l$ be integers such that $l\geq m \geq 4$ and $m$ is even.
    Furthermore, let $t_{\mathrm{min}} = \left\lfloor \frac{m}{4} - \frac{1}{2} \right\rfloor$ and $t_{\mathrm{max}} = \left\lfloor \frac{l}{2}\right\rfloor$.
    Then
    \begin{equation}
        \braket{P_l^m,\psi_j} = \begin{cases}
            2 \sum_{t=t_{\text{min}}-1}^{t_{\text{max}}} \frac{(-1)^t}{(j \pi)^{2t+2}}\diff[2t+1]{P_l^m(x)}{x}[x = 1] & l \text{ even} \\[10pt] 
            -2 \iu  \sum_{t=t_{\text{min}}-1}^{t_{\text{max}}}\frac{(-1)^t}{(j \pi)^{2t+1}} \diff[2t]{P_l^m(x)}{x}[x = 1] & l \text{ odd} 
        \end{cases}
    \end{equation}
    For $j =0$, the expansion coefficients of the constant function for even $m$ and even $l$ are given by
    \begin{equation}
        \label{eq:coeff-legendre-1}
        \Braket{P_l^m, 1} = \int_{-1}^1 P_l^m(x) \dl x = \frac{2m}{l} \frac{((l/2)!)^2 (l+m)! }{((l-m)/2)! ((l+m)/2)! (l+1)!} \, .
    \end{equation}
    For even $m$ and odd $l$, they are 0. 
\end{proposition}
Again, we use integration by parts:
\begin{lemma}
    \label{lem:app-coeff-per}
    Let $j \in \znum$, and $m,l$ be integers such that $l\geq m \geq 4$ and $m$ is even, and $\psi_j$ as defined in Eq.~\eqref{eq:app-per-eigenfunction}.
    Then, for every integer $T$ such that $0 \leq T \leq l$, the following equality holds:
    \begin{multline}
    \label{proofeq:coeff-exp-1}
    \braket{P_l^m, \phi_j}  = \int_{-1}^1  P_l^m(x) \e^{\iu j \pi (x+1)} \dl x
     = \sum_{t=0}^T (-1)^t \frac{1}{(\iu j \pi)^{t+1}}\left. \e^{\iu j \pi (x+1)} \diff*[t]{P_l^m(x)}{x}\right|_{-1}^1 \\
     - (-1)^T \frac{1}{(\iu j \pi)^{T+1}} \int_{-1}^{1} \e^{\iu j \pi (x+1)} \diff*[T+1]{P_l^m(x)}{x} \dl x \, .
\end{multline}
\end{lemma}
\begin{proof}
This will be proven by performing $T+1$ integrations by parts. We show the statement using induction over $T$.
    For $T = 0$, a single integration by parts yields
    \begin{equation}
    \braket{P_l^m,\phi_j}  = \int_{-1}^1 \e^{\iu j \pi (x+1)} P_l^m(x) \dl x
     = \frac{1}{\iu j \pi}\left. \e^{\iu j \pi (x+1)} P_l^m(x) \right|_{-1}^1 - \frac{1}{\iu j \pi} \int_{-1}^{1} \e^{\iu j \pi (x+1)} \diff{P_l^m(x)}{x} \dl x\, .
\end{equation}
Assuming that Eq.~\eqref{proofeq:coeff-exp-1} holds for $0 \leq T \leq l-1$, for $T+1$ we get, again using integration by parts,
\begin{multline}
        \braket{P_l^m, \phi_j}  = \sum_{t=0}^T (-1)^t \frac{1}{(\iu j \pi)^{t+1}}\left. \e^{\iu j \pi (x+1)} \diff*[t]{P_l^m(x)}{x}\right|_{-1}^1 \\
         \quad - (-1)^T \frac{1}{(\iu j \pi)^{T+1}} \int_{-1}^{1} \e^{\iu j \pi (x+1)} \diff*[T+1]{P_l^m(x)}{x} \dl x\\
         = \sum_{t=0}^T (-1)^t \frac{1}{(\iu j \pi)^{t+1}}\left. \e^{\iu j \pi (x+1)} \diff*[t]{P_l^m(x)}{x}\right|_{-1}^1  +(-1)^{T+1} \frac{1}{(\iu j \pi)^{T+2}} \left. \e^{\iu j \pi (x+1)} \diff*[T+1]{P_l^m(x)}{x}\right|_{-1}^1 \\
         \quad - (-1)^{T+1} \frac{1}{(\iu j \pi)^{T+2}}\int_{-1}^{1} \e^{\iu j \pi (x+1)} \diff*[T+2]{P_l^m(x)}{x} \dl x \, .    
\end{multline}
Combining the first line into a sum from $t=0$ to $T+1$, we get the desired equality for $T+1$, thus closing the induction.
\end{proof}
\begin{proof}[Proof of Proposition~\ref{prop:coeffs-per-exp}]
    As the associated Legendre polynomials are polynomials of degree $l$, the integral in Eq.~\eqref{proofeq:coeff-exp-1} vanishes for $T = l$:
    \begin{equation}
        \diff*[l+1]{P_l^m(x)}{x} = 0\, .
    \end{equation}
    Thus, using Lemma~\ref{lem:app-coeff-per}, the expansion coefficients read
    \begin{equation}
        \braket{P_l^m, \phi_j}  = \int_{-1}^1  P_l^m(x) \e^{\iu j \pi (x+1)} \dl x
     = \sum_{t=0}^l (-1)^t \frac{1}{(\iu j \pi)^{t+1}}\left. \e^{\iu j \pi (x+1)} \diff*[t]{P_l^m(x)}{x}\right|_{-1}^1 \, .
    \end{equation}
    Furthermore, for $t \leq \frac{m}{2}-1$, the derivative $\diff*[2t]{(1-x^2)^{m/2}}{x}$ evaluates to 0 at the boundaries again, and the first $m/2-1$ terms vanish:
    \begin{equation}
        \label{proofeq:coeff-exp-2}
        \braket{P_l^m,\phi_j} = \sum_{t=m/2}^l (-1)^t \frac{1}{(\iu jh  \pi)^{t+1}}\left. \e^{\iu j \pi (x+1)} \diff*[t]{P_l^m(x)}{x}\right|_{-1}^1 \, .
    \end{equation}
As $\left. \e^{\iu j \pi (x+1)} \right|_{x = \pm 1} = 1$ holds, which terms contribute to the sum only depends on $l$.
Recalling that we assume $m$ to be even, for $m/2\leq t\leq l$ we have
\begin{equation}
    \diff*[t]{P_l^m(x)}{x}[x=-1] = (-1)^{l+t}\diff*[t]{P_l^m(x)}{x}[x=-1]\, .
\end{equation}
In the following, we introduce $\tilde{t} \in \nnum $ to write an even $t$ as $2 \tilde{t}$ and odd $t$ as $2 \tilde{t}+1$.
We distinguish between even and odd $l$:
\begin{enumerate}[(i)]
    \item Even $l$: Then $\left.\diff*[t]{P_l^m(x)}{x}\right|_{-1}^1 = 0$ for even $t$. Thus only terms with odd $t = 2\tilde{t}+1$ contribute to the sum.
    In this case, $(-1)^t / \iu^{t+1} = (-1)^{2\tilde{t}+1}/\iu^{2\tilde{t}+2} = (-1)^{\tilde{t}}$.
    \item Odd $l$: Then $\left.\diff*[t]{P_l^m(x)}{x}\right|_{-1}^1 = 0$ for odd $t$. Thus only terms with even $t = 2 \tilde t$ contribute.
    Then, $(-1)^t / \iu^{t+1} = (-1)^{2 \tilde{t}}/\iu^{2 \tilde{t}+1} = -\iu (-1)^{\tilde{t}}$.
\end{enumerate}
These two cases result in the two cases in the proposition.
For each case, we can replace the sum over $t$ in Eq.~\eqref{proofeq:coeff-exp-2} with a sum over $\tilde{t}$ from $t_{\text{min}}$ to $t_{\text{max}}$, with
\begin{equation}
    t_{\mathrm{min}} = \left\lfloor \frac{m}{4} - \frac{1}{2} \right\rfloor, \quad t_{\mathrm{max}} = \left\lfloor \frac{l}{2}\right\rfloor \, ,
\end{equation}
concluding the proof.
\end{proof}

We finally point out that the expansion coefficients of the Neumann eigenfunctions $\varphi_j(x) = \cos\left(\frac{j \pi}{2}(x+1)\right)$ (see Example~\ref{ex:bc-eigenfunctions}) for $j \in \nnum$ could be calculated in the same way. They read
    \begin{equation}
        \braket{\varphi_j , P_l^m} = \left(\frac{2}{j \pi}\right)^2 \left( (-1)^j - (-1)^{l+1} \right) \sum_{t=t_{\text{min}}-1}^{t_{\text{max}}} (-1)^t \left( \frac{2}{j \pi} \right)^{2 t} \left. \diff*[2t+1]{P_l^m(x)}{x} \right|_{x=1} \, .
    \end{equation} 
    For $j = 0$, they are given by $\Braket{P_l^m,1}$ as in Eq.~\eqref{eq:coeff-legendre-1}~\cite{jepsen-integralassociatedlegendre-1955}.

\setlength{\emergencystretch}{1em}
\bibliographystyle{myquantum}
\bibliography{refs}

\end{document}